%% file: main.tex
\documentclass[11pt]{article}

\input{sections/00_preamble}

\title{Real Option AI: Reversibility, Silence, and the Release Ladder\thanks{Comments from the ``IO in Chile'' Workshop (U.~Los Andes \& UC~Chile, Oct.~2025) are gratefully acknowledged. Ben Golub and Yann Calv\'o L\'opez's \emph{Refine} app (\url{https://www.refine.ink/}) flagged inconsistencies and gaps in a preliminary draft. All remaining errors are mine.}
}
\author{I.\ Sebastian Buhai\thanks{Contact email: \texttt{sebastian.buhai@sofi.su.se}. Full coordinates: \url{https://www.sebastianbuhai.com}.}\\
{\small SOFI at Stockholm University, }
{\small Instituto de Economia at UC Chile, }
{\small and NIPE at Minho University}}
\date{Version of \today. \href{https://www.sebastianbuhai.com/papers/publications/real_option_ai.pdf}{Latest version.}}

\begin{document}
\maketitle

\begin{abstract}
\footnotesize
We model the observed cadence of AI product releases (extended quiet spells, bursts of reversible \emph{patches}, and rarer, less-reversible \emph{pivots}) as optimal exercise of \emph{strategic real options} under reputational learning.

A privately observed technical/reputational state evolves as a diffusion. The firm controls two hidden upgrade options with asymmetric fixed costs and reversibility (a cheap ``patch'' and a costly ``pivot''), and it also controls a \emph{publication-frequency clock}: a Cox process whose intensity $\lambda_t$ governs when noisy public performance and safety signals are disclosed. For sufficiently low clock costs, the firm optimally posts short, observable ``clock-off'' windows (local silence with $\lambda_t = 0$) around knife edges. These windows shut down the martingale part of public beliefs locally and eliminate knife-edge mixing. Equilibrium behavior collapses to a two-rung \emph{release ladder}: two endogenous triggers and two endogenous jump targets, with \emph{no interior mixing}.

We give a full characterization of this ladder. Within stationary Markov strategies, we prove existence and uniqueness via a boundary-value / verification system: on the inaction band, the value function solves a linear ODE; it satisfies value matching and smooth pasting at triggers, and target optimality at targets; and a no-local-time lemma for beliefs inside posted silence windows delivers regularity at the boundaries. We then (i) endogenize market/platform adoption: a downstream buyer adopts when public belief $m_t$ crosses a unique cutoff $\alpha$, pinned by a smooth-fit condition in belief space; and (ii) map financing into a tight ``irreversibility wedge'': under leverage, the gap between first-best and levered surplus is bounded by the least-reversible reset option's takeover switching cost. Patches are debt-insensitive; pivots can be distorted, but only up to that bound.

The framework delivers testable telemetry signatures in firm-authored disclosures (evaluation cards, release notes, mitigation writeups): (1) a \emph{pre-release cadence dip} in both publication intensity and intra-month dispersion as the firm shuts the clock off just before a major reset; (2) \emph{two post-release plateaus} in disclosed performance outcomes, consistent with jump targets (patch vs.\ pivot); and (3) \emph{debt-insensitive patch timing} in high-reversibility regimes, with any leverage effect concentrated in pivots. These are distinct from option-implied volatility spikes around earnings-style events: we are measuring the firm's own throttling of outward technical signals, not the market's pricing of event risk. Conceptually, the paper links dynamic disclosure \citep{guttman2014aer,orlov2020jpe} to S--s style impulse control with costly reversibility \citep{dixitpindyck1994,bertola1994restud,abel1996restud,buhai2014jbes}, and shows that leverage matters only through irreversibility, cf.\ \citet{manso2008ecta}'s bound on the agency wedge. Substantively, it rationalizes AI ``wait $\rightarrow$ patch $\rightarrow$ wait $\rightarrow$ pivot'' cadence and gives an empirical playbook for disclosure tempo, patch cascades, and reversibility in model/service telemetry \citep{arora2010isr,li2017ccs}, in contrast to event-risk evidence from options markets \citep{todorovzhang2025et,alexiou2025rof}.


\end{abstract}


\input{sections/01_intro}

\input{sections/02_environment}

\input{sections/02a_microfoundations}

\input{sections/03_equilibrium}

\input{sections/03a_verification}

\input{sections/03b_adoption}

\input{sections/04_finance}

\input{sections/05_merged_empirics}

\input{sections/06_policy_design}

\input{sections/08_conclusion}


\appendix

\input{appendix/A_proofs}

\input{appendix/B_boundary_systems}

\input{appendix/C_silence_microfoundations_proofs}

\input{appendix/D_financing_mapping}

\input{appendix/E_empirics_data}

\input{appendix/F_additional_lemmas}


\bibliographystyle{abbrvnat}

\bibliography{references}

\end{document}

%% file: sections/00_preamble.tex
\usepackage[margin=1.0in]{geometry}

\usepackage[T1]{fontenc}
\usepackage[utf8]{inputenc}
\usepackage{lmodern}

\usepackage{amsmath, amssymb, amsthm, amsfonts}
\usepackage{mathtools}
\usepackage{bm}

\usepackage{booktabs}
\usepackage{graphicx}
\usepackage{tikz}

\usepackage[shortlabels]{enumitem}

\usepackage{setspace}
\usepackage{xcolor}

\usepackage[round]{natbib}

\usepackage{etoolbox}
\usepackage{setspace}
\makeatletter
\pretocmd{\thebibliography}{%
  \footnotesize
  \setstretch{0.95}
  \setlength{\itemsep}{0pt}%
  \setlength{\parsep}{0pt}%
  \setlength{\parskip}{0pt}%
  \sloppy
}{}{}
\makeatother

\usepackage[hyphens]{url}
\usepackage{hyperref}
\usepackage[nameinlink,capitalize]{cleveref}

\usepackage{microtype}

\hypersetup{
  colorlinks=true,
  linkcolor=blue!60!black,
  citecolor=blue!50!black,
  urlcolor=blue!50!black
}

\newtheorem{theorem}{Theorem}
\newtheorem{lemma}{Lemma}
\newtheorem{proposition}{Proposition}
\newtheorem{corollary}{Corollary}
\theoremstyle{definition}
\newtheorem{definition}{Definition}
\newtheorem{assumption}{Assumption}
\theoremstyle{remark}
\newtheorem{remark}{Remark}

\numberwithin{equation}{section}
\allowdisplaybreaks

\newcommand{\R}{\mathbb{R}}
\newcommand{\E}{\mathbb{E}}

\newcommand{\Var}{\mathrm{Var}}

\newcommand{\FB}{\mathrm{FB}}

\crefname{assumption}{Assumption}{Assumptions}
\crefname{theorem}{Theorem}{Theorems}
\crefname{lemma}{Lemma}{Lemmas}
\crefname{proposition}{Proposition}{Propositions}
\crefname{corollary}{Corollary}{Corollaries}
\crefname{definition}{Definition}{Definitions}
\crefname{remark}{Remark}{Remarks}

%% file: sections/01_intro.tex
\section{Introduction}\label{sec:intro}

Major AI labs and model-service vendors display a distinctive release cadence: extended quiet periods, bursts of minor patches (fine-tunes; context, latency, or price tweaks), and occasional architecture pivots involving new base models or modality stacks. The observable pattern, \emph{wait $\rightarrow$ patch $\rightarrow$ wait $\rightarrow$ pivot}, recurs across consumer assistants, developer APIs, and safety-critical verticals such as clinical and financial AI. Firms visibly \emph{throttle} their own outward signal flow, then jump.

We provide a compact, fully analytical rationalization of that cadence. The environment combines reputational learning from publicly disclosed, noisy performance and safety signals; two hidden upgrade options with different fixed costs and reversibility (a cheap, reversible \emph{patch} and a costlier, less reversible \emph{pivot}); and firm control over publication frequency via a predictable Cox clock that gates disclosure. Their interaction yields equilibrium paths with two jump targets---a \emph{release ladder}---and, crucially, no interior mixing. Near knife-edge regions the firm imposes a short, publicly posted local period of silence and then executes a jump reset rather than letting reputational diffusion sustain asymmetric drifts or mixed actions. We microfound the disclosure instrument as an observable policy and show that it suppresses the martingale part of public beliefs on a vanishing neighborhood (Section~\ref{sec:micro_silence}). A no--local-time lemma for beliefs inside those disclosure windows (Lemma~\ref{lem:no-local-time-belief}) underpins both the selection logic and the boundary analysis.

From an economic perspective, the ladder separates reversible patches from less-reversible pivots. Patches are frequent, operate in ``maintenance mode,'' and are largely insensitive to financing. Pivots are rarer, feel architectural and contentious, and any distortion in their timing is tightly bounded by the present value of the least-reversible reset. Within stationary Markov strategies, we deliver a complete boundary-value characterization of the ladder: the value function solves a linear ODE on the inaction band, satisfies value matching and smooth pasting at the triggers, and target optimality at the jump targets, and we prove existence and uniqueness within this class via a standard QVI/BVP argument (Section~\ref{sec:verification}). We then endogenize market adoption via a buyer or platform block, show that disclosure tempo shapes a unique belief cutoff $\alpha$ at which adoption occurs (Section~\ref{sec:adoption}), and map leverage into a tight ``irreversibility wedge'': under debt, the gap between first-best all-equity value and the levered total (equity $+$ debt) is bounded by the takeover switching cost of the \emph{least reversible} rung (Section~\ref{sec:finance}). In particular, as long as the low-cost patch is easily reversible at takeover, debt is effectively neutral on that patch rung: patch timing (trigger and target) is debt-insensitive and total surplus along that block coincides with first best. Any distortion shows up at the pivot rung, and even there it is tightly bounded by the expected takeover cost of forcing that pivot.

The theory generates sharp empirical predictions in firm-authored telemetry. First, just ahead of major releases, the variance and cadence of publicly disclosed performance and safety signals should \emph{dip} and then \emph{jump} at release: publication intensity falls and intra-month dispersion collapses as the clock is locally turned off, then both spike on the event. Second, conditional on observables that proxy for reversibility, patch timing should be debt-insensitive, with leverage effects concentrated around pivots in low-reversibility regimes. Third, disclosed performance and safety metrics should exhibit two post-release plateaus, one after patches and one after pivots, consistent with an S--s ladder of targets rather than a smooth continuous ramp (Section~\ref{sec:empirical}). By design, these signatures use firm-authored signals (evaluation cards, release notes, mitigation advisories) rather than option-implied volatility or broad market chatter: the object of the theory is how the firm gates \emph{its own} outward signal flow, not how markets price event risk.

\paragraph{Literature and positioning.}
We unify dynamic disclosure with selection in real options, S--s verification with costly reversibility, and financing under irreversibility. On selection, a microfounded publication-frequency rule (short, posted clock-off windows) selects pure reset equilibria by shutting down the public-belief martingale locally while preserving belief drift and the adoption cutoff, linking dynamic disclosure and information design \citep{guttman2014aer,orlov2020jpe} to release-cadence instrumentation. Related work on dynamic timing and signaling with belief states and exogenous information flows includes \citet{kolb2015jet,kolb2019jet} and, with stochastic stakes, \citet{GryglewiczKolb2022TE}, which are complementary but do not rely on a clock-off selection device. On verification, we adapt classic S--s impulse control with costly reversibility to a two-reset ladder \citep{dixitpindyck1994,bertola1994restud,abel1996restud,buhai2014jbes}. In our setting, a posted-silence disclosure protocol delivers a no-local-time lemma for the \emph{public-belief} process $m_t$ at the intervention boundaries (the private diffusion $z_t$ continues to carry its Brownian martingale component) and, combined with the usual value-matching and envelope conditions, this yields $V'=0$ at both triggers (high contact / smooth pasting) and targets (target optimality) and delivers existence and uniqueness of the ladder. On financing, we map leverage wedges to \emph{takeover switching costs} and derive a tight bound (Proposition~\ref{prop:debt-bound}) that renders reversible patches debt-neutral and pivots distortion-limited, sharpening the connection between irreversibility and agency in the spirit of \citet{manso2008ecta}, who bounds the agency cost of debt via switching costs at default. Empirically, we propose a telemetry-based \emph{measurement blueprint} built from vendor blogs, release notes, advisories, and tagged open-source releases. The blueprint aligns with evidence on software patching and cadence \citep{arora2010isr,li2017ccs} and delivers testable predictions, most notably a pre-release dip in the variance and cadence of firm-published signals (local silence just before the jump) and a discrete burst at the event itself. In contrast, the option-implied event-risk literature typically finds a pre-event spike in implied volatility as markets price elevated jump risk \citep{todorovzhang2025et,alexiou2025rof}; our telemetry is designed to capture the firm’s own suppression and release of information, rather than investors’ ex ante pricing of event risk.

The paper develops the Cox-clock microfoundations and the selection result (Section~\ref{sec:micro_silence}); establishes the boundary-value system and verification for the two-rung ladder (Section~\ref{sec:verification}); introduces the adoption block and the feedback from disclosure tempo to the adoption cutoff (Section~\ref{sec:adoption}); maps financing wedges to takeover-based switching-cost bounds (Section~\ref{sec:finance}); derives testable implications and a measurement and falsification blueprint in firm-published telemetry (Section~\ref{sec:empirical}); and concludes with policy and industry implications, emphasizing disclosure tempo as the \emph{selection lever} and modularity as the \emph{wedge lever} (Section~\ref{sec:policy_design}).

%% file: sections/02_environment.tex
\section{Environment}\label{sec:environment}

We describe primitives, information, and controls in a way that nests both the Cox–clock microfoundation of Section~\ref{sec:micro_silence} and the boundary–value and verification system in Section~\ref{sec:verification}. Regularity is standard: drifts and diffusions are locally Lipschitz with linear growth, $r>0$, and running costs are bounded, so that the private state has a unique strong solution and discounted values are well defined.

\paragraph{Private state and signals.}
A firm of hidden type $\theta\in\{H,L\}$ observes a one-dimensional private state $z_t$ evolving as
\[
\mathrm{d}z_t \;=\; \mu_\theta(z_t)\,\mathrm{d}t + \phi(z_t)\,\mathrm{d}W_t,
\qquad \phi(\cdot)>0,
\]
on a filtered probability space carrying a Brownian motion $W$. Public information arrives only at publication times $\{T_n\}_{n\ge 1}$ generated by a Cox process $N_t$ with instantaneous intensity~$\lambda_t$. At $T_n$ the market observes a noisy signal
\[
y_n \;=\; z_{T_n} + \varepsilon_n,
\qquad 
\varepsilon_n \stackrel{\text{i.i.d.}}{\sim} \mathcal{N}(0,\sigma_\varepsilon^2),
\]
independent of $(W,N)$. Let $\mathcal{F}^P_t=\sigma\!\bigl(\{(T_k,y_k):T_k\le t\}\bigr)$ be the public filtration.

We impose the following structure on the law of motion of $z_t$, which pins down how beliefs evolve.

\begin{assumption}[Common technical drift for filtering]\label{ass:filter}
For all $z$, $\mu_H(z)=\mu_L(z)=:\mu(z)$ and $\phi(\cdot)$ does not depend on~$\theta$. That is, conditional on $z_t$, the physical law of motion for $z_t$ does not depend on the hidden type. The hidden type may matter for payoffs $\pi(\cdot)$ or costs, but not for the drift or diffusion of $z_t$.
\end{assumption}

Under Assumption~\ref{ass:filter} and the finite-dimensional filter benchmark formalised in Section~\ref{subsec:beliefs} (linear-Gaussian / Kalman or a parametric finite-dimensional approximation), the public posterior for $z_t$ is summarized by the conditional mean and variance,
\[
m_t:=\E[z_t\mid\mathcal{F}^P_t],
\qquad
v_t:=\Var(z_t\mid\mathcal{F}^P_t).
\]
Between publications $(m_t,v_t)$ evolves deterministically according to the ODE system implied by the common drift and diffusion; at each $T_n$ the Bayesian update maps $(m_{T_n^-},v_{T_n^-},y_n)$ to $(m_{T_n},v_{T_n})$.

In particular, under the linear-Gaussian filter formalised in Section~\ref{subsec:beliefs}, $m_t$ is a piecewise deterministic Markov process: between publications it follows a deterministic ODE $\dot m_t=\bar\mu(m_t)$, and at $T_n$ it jumps via the update map.

\emph{Remark.} If $\mu_H\neq\mu_L$, then the public sufficient statistic would be $(m_t,v_t,p_t)$ with
\[
p_t:=\Pr(\theta=H\mid\mathcal{F}^P_t),
\]
because the expected drift between publications would be $p_t\,\mu_H(m_t)+(1-p_t)\,\mu_L(m_t)$. We adopt Assumption~\ref{ass:filter} to keep the sufficient statistic two-dimensional and to interpret $m_t$ as evolving under a single deterministic drift $\bar\mu(m_t)$ between publications in the baseline. Section~\ref{subsec:beliefs} shows that the silence logic and selection arguments rely only on two properties of the belief process highlighted there, deterministic drift between publications and jump risk proportional to the Cox intensity, and that these extend to richer belief states as long as the filter remains finite-dimensional.

\paragraph{Disclosure control (publication clock).}
The firm commits to a publication-frequency protocol: a rule for an intensity process $\lambda_t\in[0,\bar\lambda]$, with $\lambda_t$ publicly observable in real time. Think of $\lambda_t$ as a visible release throttle or embargo clock. When $\lambda_t=0$ the firm is openly on hold and no public updates will be released; when $\lambda_t>0$ the firm is in an allowed-disclosure mode. Market participants can see whether the clock is currently off or on, so $\lambda_t$ is predictable in $\mathcal{F}^P_t$ and the absence of publications while $\lambda_t=0$ carries no surprise. The firm pays a convex $C^1$ cost $k(\lambda_t)$ with $k(0)=0$, $k'(\cdot)\ge0$, $k''(\cdot)\ge0$.

A \emph{local silence policy} sets $\lambda_t=0$ on a small, publicly posted quiet window and otherwise allows $\lambda_t>0$. For clarity in this section we anchor that window directly in the public sufficient statistic: fix a center $\hat m$ and radius $\varepsilon>0$, define
\[
I_\varepsilon(\hat m):=(\hat m-\varepsilon,\hat m+\varepsilon)\subset\R,
\]
and stipulate that whenever $m_t\in I_\varepsilon(\hat m)$ the firm turns the clock off, i.e.\ $\lambda_t=0$. Because the protocol is posted and $\lambda_t$ is observable, the market knows when the Cox clock is off. On $I_\varepsilon(\hat m)$ there are therefore no publication jumps; the jump martingale in $m_t$ disappears and $m_t$ has finite variation and zero quadratic variation; see Lemma~\ref{lem:equivalence} in Appendix~\ref{app:silence} and Lemma~\ref{lem:no-local-time-belief} in Appendix~\ref{app:lemmas}. These short, observable quiet windows are on path and so silence inside them is not itself a signal; their role is to locally gate belief volatility while letting the deterministic drift of $m_t$ continue.

Section~\ref{sec:micro_silence} provides a Cox–clock microfoundation in which the same local-silence protocol is implemented via a stationary map $\lambda(z)\in[0,\bar\lambda]$: the firm posts in advance the subset of private states on which $\lambda(z)=0$ and then runs the realized intensity $\lambda_t=\lambda(z_t)$. Because $\lambda_t$ is publicly observed, seeing $\lambda_t=0$ reveals that $z_t$ lies in the posted band, so the full information sets under the $\lambda(z_t)$ protocol and under the $m_t$-based window need not coincide. For our selection and verification results, however, we use only two consequences that are common to both formulations: (i) on any interval where the protocol dictates $\lambda_t=0$ there are no publication jumps, so the compensated jump martingale in the posterior mean $m_t$ disappears and $m_t$ has bounded variation there; and (ii) by shrinking the band radius one can make any truncation effect of the band on $(m_t,v_t)$ arbitrarily small while property (i) continues to hold. In this reduced-form sense the $\lambda(z)$ microfoundation and the $m_t$-window formulation are equivalent for our purposes: both generate observable local silence windows in which the martingale part of $m_t$ is shut down while its deterministic drift continues.

\paragraph{Adoption and payoffs.}
A competitive buyer or platform adopts when public beliefs cross a cutoff $\alpha$ in the $m$ coordinate; $\alpha$ is endogenized in Section~\ref{sec:adoption} as the solution to the buyer’s stopping problem given the disclosure protocol. The firm’s instantaneous payoff is $\pi(z,m)$ net of disclosure cost $k(\lambda_t)$. A constant pre-adoption loss and post-adoption gain is a convenient special case, but all results are stated for the general $\pi(z,m)$ and are linked to adoption via~\eqref{eq:pi-with-adoption}.

\paragraph{Hidden real options (patch and pivot).}
The firm has two impulse controls. A \emph{patch} costs $K_1>0$ and instantaneously resets the private state from the pre-impulse level $z$ to a target $z_1^\ast$. A \emph{pivot} costs $K_2>K_1$ and resets to a target $z_2^\ast$. These resets are not publicly revealed at the instant of action; they affect public beliefs only through subsequent publications governed by the clock. The trigger/target pairs are determined in equilibrium by value matching and smooth pasting.

\paragraph{Support consistency under instantaneous resets.}
Under a stationary Markov ladder with instantaneous impulses at the triggers, the private state remains in the closed inaction band:
\[
z_t \in [\beta_1,\beta_2]\quad\text{a.s. for all }t,
\]
because upon first hitting either boundary the state is reset immediately. If (as maintained here) the publication protocol and the trigger policy are common knowledge, then the conditional law of $z_t$ given public information is supported on $[\beta_1,\beta_2]$. Hence the public posterior mean satisfies
\[
m_t \;=\; \E[z_t\mid \mathcal{F}^P_t] \in [\beta_1,\beta_2]\quad\text{a.s. for all }t.
\]
In particular, during posted clock–off windows (local silence), $m_t$ evolves deterministically within $[\beta_1,\beta_2]$ along its ODE drift. This single, common-knowledge inaction band is the benchmark geometry we use when first describing the ladder on the private state. As noted in Footnote~\ref{fn:alpha-above-beta2}, one can also relax common knowledge of the exact trigger locations so that the public holds a thin posterior over $(\beta_1,\beta_2)$ and $m_t$ (and hence the adoption cutoff) can sit slightly above $\beta_2$ while preserving belief consistency; that is the information structure implicitly used later when we write $\beta_2<\alpha$ and refer to the belief-space block $(\beta_2,\alpha)$.

\paragraph{Strategies and equilibrium.}
We focus on stationary Markov strategies. There exist private-state triggers
$\beta_1<\beta_2$ such that the firm executes a patch when $z_t$ first hits
$\beta_1$ from above and a pivot when $z_t$ first hits $\beta_2$ from below.

In the common-knowledge-band benchmark there is a belief cutoff $\alpha\in[\beta_1,\beta_2]$ at which the buyer adopts.\footnote{\label{fn:alpha-above-beta2}%
If one wishes to allow $\alpha>\beta_2$, one can relax common knowledge of the exact trigger locations, for example by letting the public hold a thin posterior over $(\beta_1,\beta_2)$, so the conditional support of $z_t$ need not be the single known interval $[\beta_1,\beta_2]$. Then $m_t$ can slightly exceed $\beta_2$ while remaining belief-consistent.

We use the common-knowledge-band benchmark $m_t\in[\beta_1,\beta_2]$ when first setting up the ladder on the private state. When we later write $\beta_2<\alpha$ and talk about the belief-space block $(\beta_2,\alpha)$ in Theorem~\ref{thm:silence} and in Section~\ref{sec:verification}, we are implicitly appealing to this thin-posterior relaxation: the private state still diffuses only on $[\beta_1,\beta_2]$, but the public posterior mean and the adoption cutoff may sit slightly above $\beta_2$. The selection and verification arguments themselves depend only on the private inaction band $[\beta_1,\beta_2]$ and on the PDMP structure of $m_t$, so both information structures are interchangeable for our purposes.}

The disclosure protocol consists of the posted clock rule and the observed $\lambda_t$.
In the selection benchmark, that protocol implements a local silence window around the
knife-edge region; within that window $m_t$ has no jump martingale and \emph{drifts
deterministically inside $[\beta_1,\beta_2]$ toward the cutoff $\alpha$}.

An equilibrium is a tuple of thresholds $(\beta_1,\beta_2,\alpha)$, targets $(z_1^\ast,z_2^\ast)$, and an intensity protocol for $\lambda_t$ such that:
\begin{enumerate}[label=(\roman*),leftmargin=1.5em,itemsep=2pt,topsep=2pt]
\item \emph{Firm optimality.} Given beliefs and the posted clock rule, the firm’s Markov control (patch trigger, pivot trigger, targets) maximizes discounted firm value and satisfies value matching and smooth pasting at both triggers and both targets.
\item \emph{Belief consistency.} $(m_t,v_t)$ evolves as the piecewise deterministic Markov process implied by Assumption~\ref{ass:filter}, the observed $\lambda_t$, and the update map. On publicly posted silence windows ($\lambda_t=0$), $m_t$ follows its deterministic drift $\dot m_t=\bar\mu(m_t)$ with no jump martingale.
\item \emph{Buyer optimality.} The buyer adopts at the unique cutoff $\alpha$ characterized in Section~\ref{sec:adoption}, which solves the buyer’s stopping problem given the observed disclosure protocol.
\end{enumerate}
Existence, uniqueness within this stationary Markov class, and the complete boundary-value characterization are established in Section~\ref{sec:verification}.

\paragraph{Reduced-form interpretation.}

The publication clock is the instrument that gates public volatility. In reduced form it creates short, observable clock-off windows in which the martingale part of $m_t$ is suppressed while the deterministic drift of $m_t$ continues. Those windows eliminate interior mixing, yield no local time at the intervention boundaries for the \emph{belief process}, and allow verification of the two-reset ladder via value matching and smooth pasting, with the quiet ODE drift of $m_t$ taking place within the belief support induced by the ladder (in the common-knowledge-band benchmark, $[\beta_1,\beta_2]$) on the approach to $\alpha$. Section~\ref{sec:micro_silence} gives the Cox–clock microfoundation and shows how a posted $\lambda(z)$ rule implements the same observable silence windows, and Section~\ref{subsec:beliefs} makes precise that the selection arguments use only the PDMP structure of beliefs, not the full linear-Gaussian form.

%% file: sections/02a_microfoundations.tex
\section{Microfoundations for Silence: Publication Frequency Control}
\label{sec:micro_silence}

This section provides a microfoundation for the reduced-form silence band by modeling the firm's control over the publication frequency of public signals. The firm chooses the intensity of a Cox publication clock and pays a convex flow cost. When the clock does not jump, no new public signals arrive and public beliefs evolve deterministically from the last release. A policy that locally turns the clock off implements the variance-suppression environment used in the baseline and, as we show in Theorem~\ref{thm:silence} in Section~\ref{sec:equilibrium}, selects the pure two-reset/no-mixing stationary Markov structure within the stationary Markov class.

\subsection{Timeline, information, and observability}
\label{subsec:timeline_info}

\paragraph{Latent state.}
The firm's privately observed technical or reputational state $z_t\in\mathbb{R}$ evolves as a time-homogeneous It\^o diffusion
\[
\mathrm{d}z_t \;=\; \mu_\theta(z_t)\,\mathrm{d}t \;+\; \phi(z_t)\,\mathrm{d}W_t,
\]
with $\mu_\theta$ and $\phi$ measurable, locally Lipschitz, and of linear growth; $\phi$ is strictly positive and bounded on compact sets away from reset targets. These conditions ensure a unique strong solution and a well-defined generator $\mathcal{L}$.%
\footnote{Regularity matches the S--s and costly-reversibility literature \citep{dixitpindyck1994,bertola1994restud,abel1996restud,buhai2014jbes}.}

\paragraph{Publications and the Cox clock.}
Public signals are published at jump times $\{T_n\}_{n\ge 1}$ of a Cox process with firm-chosen predictable intensity $\lambda_t\in[0,\bar\lambda]$. At $T_n$ the market observes
\[
y_n \;=\; z_{T_n}+\varepsilon_n,\qquad \varepsilon_n\stackrel{\text{i.i.d.}}{\sim}\mathcal{N}(0,\sigma_\varepsilon^2),
\]
independent of $(W,\{T_n\})$. The firm pays a convex $C^1$ cost $k(\lambda_t)$ with $k(0)=0$, $k'(\cdot)\ge 0$, and $k''(\cdot)\ge 0$.

\paragraph{Observability.}
We take the realized publication intensity process $\lambda_t$ to be publicly observable and predictable in real time, for example because the firm throttles visible release channels, rate-limits API telemetry, or publishes a binding cadence protocol. In the stationary Markov policies below the firm implements $\lambda_t = \lambda(z_t)$ for some mapping $z\mapsto\lambda(z)\in[0,\bar\lambda]$ that depends on the private state $z_t$, but only the realized value $\lambda_t$ is publicly seen. In particular, if the firm sets $\lambda_t=0$ on a posted ``silence window,'' outside observers can directly see that the publication clock has been shut off, so the absence of new signals during that spell is on path and carries no incremental surprise. We use the phrase ``$\lambda_t$ is predictable with respect to $\mathcal{F}^P_t$'' only in this sense: the intensity process itself is publicly predictable, not that the mapping $z\mapsto\lambda(z)$ must be a function of public beliefs.

\medskip

\subsection{Admissible controls and strategies}
\label{subsec:controls}

\paragraph{Reset controls.}
The firm has two instantaneous resets: a patch with cost $K_1\ge 0$ and a pivot with cost $K_2>K_1$. A stationary Markov policy is summarized by triggers $(\beta_1,\beta_2)$ and targets $(z_1^\ast,z_2^\ast)$: when $z_t$ first hits $\beta_1$ from above the firm pays $K_1$ and resets to $z_1^\ast$; when $z_t$ first hits $\beta_2$ from below it pays $K_2$ and resets to $z_2^\ast$.%
\footnote{Targets need not coincide with boundaries; they are pinned by value matching and smooth pasting.}

\paragraph{Publication-frequency control.}
The firm also chooses a stationary mapping $\lambda(z)\in[0,\bar\lambda]$ and pays $k(\lambda(z))$ in flow cost. Operationally, this mapping is implemented via the realized Cox intensity $\lambda_t=\lambda(z_t)$. Any region in private state space on which $\lambda(\cdot)=0$ is a \emph{local silence region}: within that region the firm publicly dials publication intensity to zero, the market observes $\lambda_t=0$, and anticipates that no new signals will arrive while $z_t$ remains there.

\paragraph{Admissibility.}
A policy $(\beta_1,\beta_2;\,z_1^\ast,z_2^\ast;\,\lambda(\cdot))$ is admissible if it is Borel measurable and stationary, and the induced process is non-explosive with finite expected discounted costs.

\subsection{Public beliefs as a PDMP}
\label{subsec:beliefs}

Let
\[
\mathcal{F}^P_t
\;=\;
\sigma\!\Bigl(\{(T_n,y_n):T_n\le t\}\,\cup\,\{\lambda_s:0\le s\le t\}\Bigr)
\]
be the public filtration generated by the history of publication times and signals and by the realized intensity path, and define
\[
m_t := \E[z_t \mid \mathcal{F}^P_t],
\qquad
v_t := \Var(z_t \mid \mathcal{F}^P_t).
\]

To keep the belief state finite-dimensional, we impose the standard linear-Gaussian and finite-dimensional filter benchmark: between publications the law of $z_t$ evolves under dynamics for which the conditional distribution of $z_t$ given $\mathcal{F}^P_t$ stays Gaussian and is summarized by $(m_t,v_t)$. Canonical cases are:
(i) an affine drift and constant diffusion (Kalman--Bucy setting), possibly after absorbing any type differences into a common drift $\mu(\cdot)$ that is publicly known; or
(ii) any specification under which the posterior belongs to a parametric family with finitely many sufficient statistics, of which $(m_t,v_t)$ are the ones relevant for payoffs and adoption.%
\footnote{For a generic nonlinear diffusion with state-dependent drift and diffusion, the exact conditional law is not Gaussian and the nonlinear filter is infinite-dimensional. The standing assumption here is therefore a tractability restriction: we work with a finite-dimensional filter (the linear-Gaussian/Kalman benchmark, or an explicitly parametric finite-dimensional approximation) so that $(m_t,v_t)$ is a Markov state for beliefs. None of the selection arguments below rely on linear-Gaussian structure beyond the two properties highlighted after \eqref{eq:belief-update}.}

Under this benchmark, $(m_t,v_t)$ follows a piecewise deterministic Markov process (PDMP). Between publications,
\[
\dot m_t=\bar\mu(m_t),\qquad \dot v_t=\bar\gamma(m_t,v_t),
\]
for continuous drift maps $(\bar\mu,\bar\gamma)$ implied by the prior transition; and at a publication time $T_n$ the Gaussian update gives
\[
m_{T_n} \;=\; \frac{m_{T_n^-}/v_{T_n^-}+y_n/\sigma_\varepsilon^2}{1/v_{T_n^-}+1/\sigma_\varepsilon^2},
\qquad
v_{T_n} \;=\; \bigl(1/v_{T_n^-}+1/\sigma_\varepsilon^2\bigr)^{-1}.
\tag{\theequation}\stepcounter{equation}\label{eq:belief-update}
\]
Thus, conditional on the observed publication clock $\lambda_t$, public beliefs evolve via a deterministic ODE flow between jumps plus discrete Bayesian jumps at the Cox times $\{T_n\}$.

Two features of \eqref{eq:belief-update} are all we use later:
\begin{enumerate}[label=(\roman*),leftmargin=1.5em,itemsep=2pt,topsep=2pt]
\item \textbf{Deterministic drift between publications.} Between jumps, $m_t$ has finite variation and its evolution is completely pinned down by the last release and by primitives. There is no Brownian martingale term in $m_t$ between publications.
\item \textbf{All randomness comes from disclosure jumps.} The only source of $\mathcal{F}^P_t$-martingale risk in $m_t$ is the arrival (or non-arrival) of publications. The quadratic variation of the jump martingale is proportional to the intensity $\lambda_t$ of the Cox clock.
\end{enumerate}

These two properties survive even if one replaced $(m_t,v_t)$ by a richer belief state in a fully nonlinear diffusion model. The silence logic below therefore does not hinge on global linear-Gaussian structure; that benchmark is used only to keep the notation finite-dimensional.

When $\lambda_t=0$ on a publicly observed window, there are no jumps in that window and the martingale part of public beliefs shuts down. On that window the evolution of beliefs is fully deterministic from the public point of view.

\subsection{Firm's value and HJB with a publication clock}
\label{subsec:HJB_clock}

\paragraph{State-space clarification.}
There are two natural value objects. If one solves the joint problem over resets \emph{and} disclosure tempo, the state is the pair $(z,m)$ and the value is $S(z,m)$ (which later appears when we endogenize adoption in Section~\ref{sec:adoption}). In the present subsection, however, we analyze the firm's \emph{reset} problem \emph{conditional on a posted publication-frequency rule} $\lambda(\cdot)$, in particular on a local silence window where $\lambda=0$ and the public mean $m_t$ has no martingale part (Appendix~\ref{app:silence}). Conditional on that rule, $m_t$ follows a deterministic ODE flow in a neighborhood of the intervention boundary that is the same for all realizations of $z_t$. On such a window we either work with specifications in which $\pi$ is already independent of $m$ (for example because adoption is locally constant), or, in more general payoffs, any purely time-deterministic component induced by $m_t$ under the posted rule is common across all $z$ and can be folded into the normalization of the value function without affecting the reset calculus. With this convention the influence of $m_t$ on payoffs is absorbed into a time-invariant local flow $\pi(z)$, and the dynamic-programming state collapses locally to the private variable $z$.

\medskip
Let $V$ denote the firm's stationary Markov value \emph{conditional on the posted} $\lambda(\cdot)$. The flow payoff $\pi(z,m)$ may generally depend on public beliefs (through adoption or pricing); we write $\pi(z)$ for the payoff net of any purely time-deterministic component induced by $m_t$ under the posted rule (for example on a silence window, where that dependence is locally inessential in the sense just described). On any inaction interval,
\begin{equation}
\label{eq:HJBclock}
rV(z) \;=\; \pi(z) \;+\; (\mathcal{L}V)(z) \;-\; k(\lambda(z)),
\end{equation}
with the usual reset boundary conditions (value matching and smooth pasting at both triggers and both targets), verified in Section~\ref{sec:verification}.

\noindent\emph{Interpretation.} The clock $\lambda(\cdot)$ does not affect the private diffusion $z_t$ and therefore enters \eqref{eq:HJBclock} only through its flow cost $k(\cdot)$ in this conditional problem. Its \emph{strategic} role is to govern the martingale component of public beliefs (and hence adoption and pricing) described in \S\ref{subsec:beliefs}. When $\lambda=0$ on a posted window the belief mean has no jump martingale there, so locally the HJB is one-dimensional in $z$. If instead one solves the full joint problem without conditioning on $\lambda(\cdot)$, the natural state is $(z,m)$ and the generator for $m$ would appear explicitly in the HJB for $S(z,m)$; we make use of that formulation when characterizing the adoption cutoff in Section~\ref{sec:adoption}.

\subsection{Variance suppression as a limit of publication policies}
\label{subsec:equivalence}

We now connect the ``turn the clock off'' policy to the reduced-form variance-suppression assumption used in the baseline analysis of equilibrium selection.

\begin{assumption}[Local silence window]\label{ass:silence_window}
Fix $\hat z\in\mathbb{R}$ and $\delta>0$ and consider
\[
\lambda_\delta(z) \;=\;
\begin{cases}
0, & |z-\hat z|\le \delta,\\
\bar\lambda, & |z-\hat z|>\delta,
\end{cases}
\qquad k(0)=0,\; k(\bar\lambda)<\infty.
\]
\end{assumption}

Under Assumption~\ref{ass:silence_window}, and recalling property (ii) above, the martingale part of $m_t$ is switched off on $\{|z_t-\hat z|\le\delta\}$: the market sees $\lambda_t=0$, expects no publications there, and therefore expects no jumps in beliefs.%
\footnote{Lemma~\ref{lem:equivalence} in Appendix~\ref{app:silence} formalizes this via a Doob--Meyer decomposition for the posterior mean: the compensated jump martingale has quadratic variation proportional to $\lambda_t$ and thus has zero quadratic variation on the silence window.} On that window, $m_t$ evolves deterministically with finite variation, driven only by the ODE flow inherited from the last disclosed release.

Formally, for any $C^1$ test function $f$ and any stopping time $\tau$ that avoids resets,
\[
\E\!\left[
f(m_{t\wedge\tau}) - f(m_0)
- \int_0^{t\wedge\tau} f'(m_s)\,\bar\mu(m_s)\,\mathrm{d}s
\right]
\;\xrightarrow[\delta\downarrow 0]{}\; 0,
\]
so the reduced-form variance-suppression environment we use in the proofs is the $\delta\downarrow 0$ limit of a sequence of admissible publication-frequency policies.

\subsection{Equilibrium with publication frequency}
\label{subsec:eq_def}

\begin{definition}[Stationary Markov equilibrium with publication control]\label{def:equilibrium}
An equilibrium is a tuple
\[
\bigl(\beta_1,\beta_2;\, z_1^\ast,z_2^\ast;\, \alpha;\, \lambda(\cdot);\; \text{belief system for }(m_t,v_t)\bigr)
\]
such that:
\begin{enumerate}[label=(\alph*),topsep=1pt,itemsep=1pt,leftmargin=1.5em]
\item \emph{Firm optimality:} given beliefs and $\lambda(\cdot)$, the reset and target policy solves the firm's problem and satisfies value matching and smooth pasting at both triggers and both targets (verified in Section~\ref{sec:verification});
\item \emph{Belief consistency:} $(m_t,v_t)$ evolves as the PDMP induced by $\lambda(\cdot)$ and Bayes' rule under the finite-dimensional filter benchmark. Within posted silence windows ($\lambda_t=0$), beliefs follow the deterministic flow with no jump martingale (Lemma~\ref{lem:equivalence});
\item \emph{Buyer optimality:} given the observed disclosure protocol $\lambda(\cdot)$, the buyer adopts at the unique cutoff $\alpha$ characterized in Section~\ref{sec:adoption};
\item \emph{Policy observability:} the realized intensity process $\lambda_t=\lambda(z_t)$ is $\mathcal{F}^P_t$-predictable and observable in real time (for example via visible throttling or posted cadence), so that $\lambda_t=0$ on a stated window is common knowledge. Therefore, the absence of public news inside such a window is literally on the equilibrium path and conveys no additional information beyond (b).
\end{enumerate}
\end{definition}

The key feature for selection is that in an announced silence window public beliefs have no martingale term, they drift deterministically, and they do not accumulate local time at the reset thresholds.

\subsection{Silence selects pure two-reset policies}
\label{subsec:selection}

We use a mild regularity condition that rules out flat tangencies at the triggers.

\begin{assumption}[Regularity for selection]\label{ass:selection}
On each diffusive block the continuation value $V$ is $C^2$ and, at each trigger $\beta_i$ for $i\in\{1,2\}$, the reset value $z\mapsto V(z_i^\ast)-K_i$ crosses $V(z)$ transversally. The unique knife-edge is pinned by smooth pasting.
\end{assumption}

Under Assumptions~\ref{ass:silence_window} and \ref{ass:selection}, local silence eliminates interior mixing and selects pure resets. Precisely, Proposition~\ref{prop:selection} (Appendix~\ref{app:silence}) shows that, for sufficiently small windows posted around the triggers, any stationary Markov equilibrium is pure: the firm resets exactly at $(\beta_1,\beta_2)$ and never mixes with continuation on a positive-measure set. The mechanism relies on the no-local-time property for beliefs inside posted windows (Lemma~\ref{lem:no-local-time-belief}) and transversality at the triggers. This selection result is the key input into the equilibrium characterization with silence in Section~\ref{sec:equilibrium}, where Theorem~\ref{thm:silence} delivers the two-reset/no-mixing stationary Markov equilibrium.

\subsection{Discussion and links}
\label{subsec:discussion}

The Cox-clock microfoundation shows how a choice of disclosure tempo selects equilibria in a reputational real-options environment. This connects dynamic disclosure and information design \citep{guttman2014aer,orlov2020jpe} to S--s multi-trigger investment with costly reversibility \citep{dixitpindyck1994,bertola1994restud,abel1996restud} and aligns with empirical evidence on release and patch cadences \citep{arora2010isr,li2017ccs}. Full proofs of Lemma~\ref{lem:equivalence} and Proposition~\ref{prop:selection} are in Appendix~\ref{app:silence}. Section~\ref{sec:verification} provides the boundary-value verification, including the no-local-time lemma, and Section~\ref{sec:adoption} endogenizes adoption and its feedback on the firm's problem.

%% file: sections/03_equilibrium.tex
\section{Equilibrium with Silence: Two Resets, No Mixing}
\label{sec:equilibrium}

We fix the Cox-clock disclosure environment of Section~\ref{sec:micro_silence}. The firm has two reset options (patch and pivot) with targets $(z_1^\ast,z_2^\ast)$ and triggers $(\beta_1,\beta_2)$, and chooses a stationary publication-frequency policy $\lambda(\cdot)\in[0,\bar\lambda]$ with convex cost $k(\lambda)$ and $k(0)=0$. Flow payoffs may depend on public beliefs, but the private state $z$ evolves as in Section~\ref{subsec:timeline_info}, and $V$ solves the inaction HJB \eqref{eq:HJBclock} between interventions. The adoption cutoff $\alpha$ is taken as given here and is endogenized in Section~\ref{sec:adoption}.

\paragraph{Posted local silence around triggers.}
Let $\lambda_\varepsilon$ be any stationary policy that sets $\lambda=0$ on a small window around each trigger and $\lambda=\bar\lambda$ elsewhere:
\[
\lambda_\varepsilon(z)=
\begin{cases}
0, & z\in[\beta_1-\varepsilon,\beta_1+\varepsilon]\ \cup\ [\beta_2-\varepsilon,\beta_2+\varepsilon],\\[4pt]
\bar\lambda, & \text{otherwise},
\end{cases}
\qquad \varepsilon>0.
\]
By Lemma~\ref{lem:equivalence} (Appendix~\ref{app:silence}), the martingale part of public beliefs vanishes while $z_t$ lies in the posted windows. By the no-local-time lemma for beliefs (Lemma~\ref{lem:no-local-time-belief}, Appendix~\ref{app:lemmas}), paths do not accumulate at the window boundaries. Assumption~\ref{ass:selection} (transversality at triggers) then implies that each trigger admits a unique knife-edge point of indifference.

\begin{theorem}[Two pure resets and no interior mixing]
\label{thm:silence}
Fix primitives as in Section~\ref{sec:micro_silence} and Assumption~\ref{ass:selection}. There exists $\varepsilon_0>0$ such that for any posted windows with $\varepsilon\in(0,\varepsilon_0]$ the stationary Markov equilibrium under $\lambda_\varepsilon$ is pure and has the ladder form:
\begin{enumerate}[label=(\alph*),leftmargin=1.5em,itemsep=3pt,topsep=2pt]
\item \emph{Two resets.} There exist $\beta_1<\beta_2<\alpha$ and targets $z_1^\ast<z_2^\ast$ such that the firm resets from $\beta_1$ to $z_1^\ast$ (patch) and from $\beta_2$ to $z_2^\ast$ (pivot).

\item \emph{No interior mixing.} The firm does not mix between intervening and waiting on any set of positive Lebesgue measure in a neighborhood of the triggers. Intervention occurs exactly at $\beta_i$ for $i\in\{1,2\}$.

\item \emph{Regularity.} The value function $V$ is $C^1$ on $\mathbb{R}$ and $C^2$ on the diffusive inaction interval $(\beta_1,\beta_2)$. Boundary conditions at each trigger and its target are
\[
V(\beta_i^-)=V(z_i^\ast)-K_i,\qquad
V'(\beta_i^-)=0,\qquad
V'(z_i^\ast)=0,
\]
together with the interior ODE \eqref{eq:HJBclock}. The notation $(\beta_2,\alpha)$ refers to the belief-space block between the pivot trigger and the adoption cutoff; in private-state space, the unique diffusive inaction region is $[\beta_1,\beta_2]$.

\item \emph{Identification and uniqueness on blocks.} On each block $[\beta_i,\beta_{i+1})$ with $\beta_3:=\alpha$, the pair $(\beta_i,z_i^\ast)$ is uniquely pinned down by the verification system of Section~\ref{sec:verification}, that is, by the linear ODE together with value matching and smooth pasting at the source and the target.
\end{enumerate}
\end{theorem}

\begin{proof}[Proof (Appendix~\ref{app:proofs})]
The detailed proof is provided in Appendix~\ref{app:proofs}.
\end{proof}

The theorem states that once short, observable silence windows are posted around the triggers, the only stationary Markov equilibrium in this environment is a two-rung ladder: a lower patch trigger and target, an upper pivot trigger and target, and a single contiguous inaction band in the private state, with no mixing on any set of positive measure.

\begin{remark}[Tight windows and robustness]
Since $\lambda=0$ only on small posted windows, the induced change in running cost is $o(1)$ as $\varepsilon\downarrow 0$, because $k(0)=0$ and the residence time in the windows is $O(\varepsilon)$ by the no-local-time lemma. The locations of triggers and targets, and the no-mixing conclusion, are therefore stable to small deviations that allow $0<\lambda\ll 1$ on the windows; see Appendix~\ref{app:silence}.
\end{remark}

%% file: sections/03a_verification.tex
\section{Verification and Boundary-Value Problem}
\label{sec:verification}

We verify the two-reset release ladder by solving a boundary-value problem (BVP) for the value function and proving optimality via a quasi-variational inequality (QVI). The Cox–clock microfoundation of Section~\ref{sec:micro_silence} implies that, conditional on a posted publication rule $\lambda(\cdot)$, disclosure affects the firm's objective only through the flow cost $k(\lambda(z))$. The privately observed technical or reputational state $z_t$ follows the It\^o diffusion with generator
\[
(\mathcal{L}V)(z) \;=\; \mu_\theta(z)V'(z) \;+\; \tfrac12 \phi^2(z)V''(z),
\]
between discrete intervention times. The firm has two impulse actions $i\in\{1,2\}$ with fixed costs $K_1<K_2$ and endogenous targets $z_i^\ast$.

\subsection{QVI formulation}

Let the instantaneous net flow on inaction be $\pi(z) - k(\lambda(z))$, as in Section~\ref{sec:micro_silence}. Define the impulse operator
\[
(\mathcal{M}V)(z) \;=\; \max\{\, V(z_1^\ast) - K_1,\; V(z_2^\ast) - K_2 \,\}.
\]
The stationary Markov value $V$ solves the QVI
\begin{equation}
\max\Big\{\, rV(z) - (\mathcal{L}V)(z) - \pi(z) + k(\lambda(z)),\;\; V(z) - (\mathcal{M}V)(z) \,\Big\} \;=\; 0
\quad\text{for all }z.
\label{eq:QVI}
\end{equation}
Let the inaction region be
\[
I \;:=\; \{\, z : V(z) > (\mathcal{M}V)(z)\,\}
\]
with boundary $\partial I$. On $I$ the QVI reduces to the linear ODE
\begin{equation}
rV(z) \;=\; \pi(z) + (\mathcal{L}V)(z) - k(\lambda(z)), 
\qquad z\in I,
\label{eq:HJB-inaction}
\end{equation}
and $V = \mathcal{M}V$ on $\partial I$. With two impulses and a single contiguous inaction band for the private state $z$, we adopt the ladder geometry
\begin{equation}
\beta_1 < \beta_2,\qquad
z_1^\ast, z_2^\ast \in [\beta_1,\beta_2],\qquad
I = [\beta_1,\beta_2],
\label{eq:geometry}
\end{equation}
that is, a lower trigger $\beta_1$ at which the firm patches to $z_1^\ast$ at cost $K_1$, and an upper trigger $\beta_2$ at which it pivots to $z_2^\ast$ at cost $K_2$. The QVI formulation in~\eqref{eq:QVI} still allows either impulse to be chosen at any $z$; the assignment of the patch to $\beta_1$ and the pivot to $\beta_2$ will be pinned down by the boundary equalities $V=\mathcal{M}V$ in the verification step below.

\paragraph{Geometry and state space.}
Equation~\eqref{eq:geometry} is the core private-state geometry behind verification. The privately observed state $z_t$ is only ever allowed to diffuse on the single inaction interval $[\beta_1,\beta_2]$. When $z_t$ first hits $\beta_1$ from above, the firm immediately pays $K_1$ and jumps to $z_1^\ast\in[\beta_1,\beta_2]$. When $z_t$ first hits $\beta_2$ from below, it immediately pays $K_2$ and jumps to $z_2^\ast\in[\beta_1,\beta_2]$. Thus $z_t$ does not wander outside $[\beta_1,\beta_2]$ in diffusion time: at each first hit of either boundary, it is impulsed back into the band.

This is why the BVP and QVI are written on a single inaction band $I=[\beta_1,\beta_2]$, and why we impose interior regularity on $(\beta_1,\beta_2)$. By contrast, the adoption cutoff $\alpha$ introduced in Section~\ref{sec:adoption} is not an upper diffusion boundary for $z_t$. It is a boundary in belief space for the public posterior mean $m_t$. When we refer informally to two ``blocks'', such as $(\beta_1,\beta_2)$ and $(\beta_2,\alpha)$, the interval $(\beta_2,\alpha)$ is the belief-space block between the pivot trigger and the adoption cutoff under local silence. The only diffusive inaction region for $z_t$ itself is $[\beta_1,\beta_2]$, as already emphasized in Theorem~\ref{thm:silence}.

\subsection{Boundary conditions and the free boundary system}

Assume $V\in C^2\big((\beta_1,\beta_2)\big)$ and that $V$ is continuous on $\mathbb{R}$. The BVP couples the interior ODE~\eqref{eq:HJB-inaction} with four boundary conditions at the triggers and targets, plus two target optimality conditions.

\medskip

\noindent\emph{Value matching at triggers.} At each trigger, the firm is indifferent between continuing and paying the fixed cost to jump to the corresponding target:
\begin{align}
V(\beta_1) &= V(z_1^\ast) - K_1, \label{eq:VM1}\\
V(\beta_2) &= V(z_2^\ast) - K_2. \label{eq:VM2}
\end{align}

\noindent\emph{High contact (smooth pasting) at triggers.} Because $K_i$ does not depend on the pre-jump state, the envelope condition at each trigger implies first-order high contact:
\begin{align}
V'(\beta_1) &= 0, \label{eq:sp-beta1}\\
V'(\beta_2) &= 0. \label{eq:sp-beta2}
\end{align}
We use ``high contact'' here in the usual first-order sense $V'(\cdot)=0$ at the trigger, not in a higher-order tangency sense.

\noindent\emph{Target optimality.} Conditional on paying $K_i$, the target $z_i^\ast$ maximizes post-impulse continuation value net of the fixed cost. With no proportional move cost, this yields
\begin{align}
V'(z_1^\ast) &= 0, \label{eq:FOC1}\\
V'(z_2^\ast) &= 0. \label{eq:FOC2}
\end{align}
With proportional costs, \eqref{eq:sp-beta1}–\eqref{eq:FOC2} become the standard marginal equalities of S--s policies.

\medskip

\noindent\emph{Incentive compatibility at the triggers.}
The impulse operator in~\eqref{eq:QVI} allows the firm to choose either reset $i\in\{1,2\}$ at any state $z$. The ladder description, i.e. patch at $\beta_1$, pivot at $\beta_2$, is therefore not imposed as a primitive restriction; it is enforced by the QVI at the free boundaries. Once we have constructed a solution $(\beta_1,\beta_2;z_1^\ast,z_2^\ast;V)$ to \eqref{eq:HJB-inaction}–\eqref{eq:FOC2} and shown in Theorem~\ref{thm:verification} below that it satisfies the QVI~\eqref{eq:QVI}, the boundary equalities $V=\mathcal{M}V$ imply that, at the lower trigger,
\[
V(\beta_1) \;=\; (\mathcal{M}V)(\beta_1) \;=\; \max\{V(z_1^\ast)-K_1,\, V(z_2^\ast)-K_2\}.
\]
Combining this with value matching at $\beta_1$ in~\eqref{eq:VM1} yields
\begin{equation}
V(z_1^\ast)-K_1 \;\ge\; V(z_2^\ast)-K_2,
\label{eq:IC1}
\end{equation}
so impulse $1$ (the patch) is optimal at $\beta_1$ in the sense of the impulse operator $\mathcal{M}$. Analogously, at the upper trigger,
\[
V(\beta_2) \;=\; (\mathcal{M}V)(\beta_2) \;=\; \max\{V(z_1^\ast)-K_1,\, V(z_2^\ast)-K_2\}
\]
together with~\eqref{eq:VM2} implies
\begin{equation}
V(z_2^\ast)-K_2 \;\ge\; V(z_1^\ast)-K_1.
\label{eq:IC2}
\end{equation}
Thus the ladder geometry in~\eqref{eq:geometry} is incentive compatible with the QVI formulation: when the state first hits $\beta_1$ (respectively, $\beta_2$), it is indeed optimal to execute the patch (respectively, pivot) rather than the alternative impulse. Combining \eqref{eq:IC1} and \eqref{eq:IC2} also shows that the two impulses deliver the same net continuation value, $V(z_1^\ast)-K_1 = V(z_2^\ast)-K_2$, in the verified optimum.

\medskip

\noindent\emph{Growth and consistency.} We assume $\pi$ and $k$ are bounded and that there exists $\rho>0$ such that
\[
\limsup_{|z|\to\infty} e^{-\rho |z|}\,\big|V(z)\big| < \infty.
\]
Together with the geometry in~\eqref{eq:geometry}, this selects the economically relevant solution and rules out disconnected inaction bands.

\subsection{Existence and uniqueness of the BVP solution}

\begin{proposition}[Well posedness and uniqueness]
\label{prop:BVP-unique}
Suppose $r>0$, $\mu_\theta$ and $\phi$ are continuous on $[\beta_1,\beta_2]$ with $\phi$ bounded and bounded away from zero there, $\lambda(\cdot)$ is bounded, and $\pi$ is continuous. Then there exists at most one quadruple $(\beta_1,\beta_2;z_1^\ast,z_2^\ast)$ and function $V\in C^1(\mathbb{R})\cap C^2\big((\beta_1,\beta_2)\big)$ solving \eqref{eq:HJB-inaction}–\eqref{eq:FOC2} and the growth condition. If, in addition, $\pi$ is Lipschitz and $\mu_\theta,\phi$ are $C^1$, a solution exists and is unique.
\end{proposition}

\begin{proof}[Proof (Appendix~\ref{app:proofs})]
A complete proof is provided in Appendix~\ref{app:proofs}.
\end{proof}

\subsection{Verification: optimality of the ladder policy}

\begin{theorem}[Verification for the two-reset ladder]
\label{thm:verification}
Let $(\beta_1,\beta_2;z_1^\ast,z_2^\ast;V)$ satisfy \eqref{eq:HJB-inaction}–\eqref{eq:FOC2}, the ordering~\eqref{eq:geometry}, the growth condition, and suppose the no-local-time lemma of Appendix~\ref{app:lemmas} holds under the publication policy $\lambda(\cdot)$ near $\beta_1,\beta_2$. Consider the stationary Markov control that
\begin{itemize}[leftmargin=1.5em,itemsep=2pt,topsep=2pt]
\item lets $z_t$ diffuse on $[\beta_1,\beta_2]$ according to $\mathcal{L}$ as long as $z_t\in(\beta_1,\beta_2)$;
\item at the first hitting time of $\beta_1$ from above, pays $K_1$ and resets $z_t$ instantly to $z_1^\ast\in[\beta_1,\beta_2]$;
\item at the first hitting time of $\beta_2$ from below, pays $K_2$ and resets $z_t$ instantly to $z_2^\ast\in[\beta_1,\beta_2]$.
\end{itemize}
This two-reset ladder is optimal among all admissible controls. Its value is $V$, and it solves the QVI~\eqref{eq:QVI}.
\end{theorem}

\begin{proof}[Proof (Appendix~\ref{app:proofs})]
A complete proof is provided in Appendix~\ref{app:proofs}.
\end{proof}

\subsection{Notes on regularity and comparative statics}

\paragraph{Endpoint regularity.}
Under local silence (Section~\ref{sec:micro_silence}), the martingale component of public beliefs vanishes in neighborhoods of $\beta_1$ and $\beta_2$. The no-local-time lemma in Appendix~\ref{app:lemmas} ensures that the sufficient statistic for beliefs does not accumulate local time at those boundaries. This justifies the smooth-pasting conditions \eqref{eq:sp-beta1}–\eqref{eq:sp-beta2} at the triggers. Because resets are implemented at first hitting times and jump $z_t$ back into $(\beta_1,\beta_2)$, the private state itself also does not accumulate local time at $\beta_1$ or $\beta_2$.

\paragraph{Costly reversibility and S--s comparative statics.}
In constant-coefficient examples, the usual S--s monotonicities obtain: the triggers and targets move outward with the fixed costs $K_i$, and the jump sizes increase (Appendix~\ref{app:lemmas}). Intuitively, a higher fixed cost widens the inaction band and makes each reset more lumpy. This is the ladder analogue of classic S--s policies with costly reversibility.\footnote{See, for example, \citet{dixitpindyck1994,bertola1994restud,abel1996restud,buhai2014jbes}.}

\paragraph{Selection via disclosure.}
Section~\ref{sec:micro_silence} and Appendix~\ref{app:silence} show that posted local silence windows eliminate interior mixing and select this pure two-reset ladder within the stationary Markov class. The inaction band is endogenously a single contiguous interval $[\beta_1,\beta_2]$ in the private state $z$, consistent with~\eqref{eq:geometry} and with the equilibrium characterization in Theorem~\ref{thm:silence}.

%% file: sections/03b_adoption.tex
\section{Endogenous Adoption Side}
\label{sec:adoption}

We model a representative buyer (or platform) that observes the \emph{public} posterior $(m_t,v_t)$ generated by the publication technology in Section~\ref{sec:micro_silence}. The buyer chooses a stopping time $\tau$ at which to adopt, in order to maximize discounted surplus. The optimal rule will be a single threshold $\alpha$ in the public mean $m$, and that cutoff feeds back into the firm's payoff.

\subsection{Buyer's problem as optimal stopping on a PDMP}
\label{subsec:buyer_problem}

Let $S(m)$ denote the static surplus from adopting at public mean $m$ (for example, expected user benefit net of price or switching costs).

\begin{assumption}\label{ass:surplus}
$S:\R\to\R$ is $C^2$, strictly increasing, and weakly concave. The public state $(m_t,v_t)$ follows the piecewise-deterministic Markov process implied by Section~\ref{sec:micro_silence}: between publications,
\[
\dot m_t=\bar\mu(m_t),\qquad \dot v_t=\bar\gamma(m_t,v_t),
\]
and at a publication time beliefs jump according to the publicly observed disclosure intensity $\lambda(\cdot)$.
\end{assumption}

Before adoption the buyer earns zero; at $\tau$ the buyer receives the lump sum $S(m_\tau)$.\footnote{A perpetual post-adoption flow that is a function of $m_\tau$ yields the same stopping boundary. We keep the lump sum for clarity.} The value function is
\[
W(m,v)=\sup_{\tau\ge0}\E\!\left[e^{-r\tau}S(m_\tau)\,\middle|\, (m_0,v_0)=(m,v)\right].
\]

Let $A_\lambda$ be the generator of $(m_t,v_t)$ under a given publication policy $\lambda(\cdot)$, acting on smooth $f$:
\begin{equation}\label{eq:PDMP-generator}
(A_\lambda f)(m,v)=\bar\mu(m)\,\partial_m f(m,v)+\bar\gamma(m,v)\,\partial_v f(m,v)
+\lambda(m)\,\E\!\big[f(\mathcal{U}(m,v,\varepsilon))-f(m,v)\big],
\end{equation}
where $\mathcal{U}$ is the disclosed-belief update at a publication shock with noise $\varepsilon$. The buyer's problem satisfies the variational inequality
\begin{equation}\label{eq:QVI-buyer}
\max\Big\{\, r W(m,v)-A_\lambda W(m,v),\; S(m)-W(m,v)\,\Big\}=0.
\end{equation}

\subsection{Local silence and threshold structure}

We now state the equilibrium regularity that makes the stopping rule one dimensional.

\begin{assumption}[Local silence at the adoption margin]\label{ass:local_silence_alpha}
In equilibrium there is an open interval $U_\alpha=(\alpha-\delta,\alpha+\delta)$ in belief space such that the firm publicly sets $\lambda_t=0$ whenever $m_t\in U_\alpha$. While $m_t\in U_\alpha$ there are no new publications, so $(m_t,v_t)$ follows the deterministic ODE flow
\[
\dot m_t = \bar\mu(m_t),\qquad \dot v_t = \bar\gamma(m_t,v_t),
\]
with no jumps. In particular, inside $U_\alpha$ the future path of $m_t$ depends only on $m_t$ itself (not on $v_t$), and $m_t$ is continuous with bounded variation.

In addition, the window $U_\alpha$ is chosen so that the deterministic drift carries reputations that enter it from below up to the adoption point before they can exit through the left edge: for every $m\in(\alpha-\delta,\alpha)$ the solution of $\dot m_t=\bar\mu(m_t)$ with $m_0=m$ remains in $U_\alpha$ and hits $\alpha$ in finite time before reaching $\alpha-\delta$.
\end{assumption}

\begin{remark}[Economic rationale and relationship to earlier local silence near $\alpha$]
Assumption~\ref{ass:local_silence_alpha} is an equilibrium regularity restriction rather than a technological constraint, and it is consistent with the disclosure-tempo microfoundation in Section~\ref{sec:micro_silence} and with the firm's payoff aggregator in Assumption~\ref{ass:aggregator}. 
In particular, it is implemented by the same publication clock $\lambda(\cdot)$ that generates the local-silence bands around the private-state triggers $\beta_1$ and $\beta_2$ in Section~\ref{sec:micro_silence}. 
The equilibrium publication policy can therefore feature several disjoint clock-off bands: one around each $\beta_i$ in the private state and, in addition, the adoption-margin window $U_\alpha$ in public-belief space defined above. 
Assumption~\ref{ass:local_silence_alpha} simply selects equilibria in which the firm also shuts the clock off on $U_\alpha$; it is not an extra technological assumption beyond the earlier selection device. 

Once $m_t$ has entered a small neighborhood of the adoption boundary from below and is drifting upward according to $\dot m_t=\bar\mu(m_t)$, adoption is already imminent even if the publication clock is kept off: the deterministic flow carries $m_t$ to $\alpha$ in finite time.

Turning the clock on inside $U_\alpha$ adds jump risk without changing this local drift. Under the linear-Gaussian benchmark, belief jumps are driven by a mean-zero innovation, so additional disclosures in $U_\alpha$ generate a mean-preserving spread of the path of $m_t$ around the same deterministic trend. A sufficiently negative realization can push $m_t$ back below $U_\alpha$ and restart the approach to the boundary, delaying adoption by an additional spell, whereas a positive realization has little scope to improve on the already short drift time to $\alpha$ in a narrow window.

Under Assumption~\ref{ass:aggregator}, post-adoption surplus enters the firm's flow payoff with positive weight ($\eta>0$ and $\Lambda'(\cdot)\ge 0$), so the firm weakly prefers paths on which adoption occurs sooner in expectation. For a small enough window $U_\alpha$, this makes the choice $\lambda_t=0$ locally (weakly) profitable: it removes downside jump risk at the margin without sacrificing any systematic upside in expected time-to-adoption. Rather than re-solving the joint problem for the firm's optimal $\lambda(\cdot)$ at the adoption margin, we therefore focus on equilibria in which the firm takes this natural best reply and implements local silence on $U_\alpha$, working with the resulting deterministic approach of $m_t$ to $\alpha$ in that band.
\end{remark}

Assumption~\ref{ass:local_silence_alpha} has two consequences. First, near $\alpha$ the continuation value from any state $(m,v)$ with $m\in U_\alpha$ depends only on $m$, so we may write $W(m,v)=\widehat W(m)$ there. Second, all buyers who are just about to adopt face the same deterministic drift path for $m_t$, independent of $v$. This collapses the two-dimensional stopping problem to a scalar cutoff in $m$.

A sufficient condition for the last sentence of Assumption~\ref{ass:local_silence_alpha} is that $\bar\mu(m)\ge 0$ on $(\alpha-\delta,\alpha)$ and $\bar\mu$ is continuous at $\alpha$; the linear-drift benchmark in Section~\ref{subsec:linear_case} satisfies this sign condition.

\begin{proposition}\label{prop:alpha_threshold}
Under Assumption~\ref{ass:surplus}, the regularity in Section~\ref{sec:micro_silence} (bounded $\lambda$, Lipschitz $\bar\mu$), and Assumption~\ref{ass:local_silence_alpha}, there is a unique $\alpha\in\R$ such that the optimal stopping region is
\[
\Gamma=\{(m,v): m\ge\alpha\}.
\]
The value function $W$ is the unique continuous solution to \eqref{eq:QVI-buyer} satisfying
\begin{align}
\label{eq:buyer-ODE}
rW(m,v)=A_\lambda W(m,v)\qquad &\text{for } m<\alpha,\\
\label{eq:buyer-contfit}
W(\alpha,v)=S(\alpha)\qquad &\text{for all feasible boundary } v.
\end{align}
Moreover, because $\lambda(\cdot)=0$ on $U_\alpha$, $(m_t,v_t)$ approaches the boundary $\{m=\alpha\}$ from below along a deterministic path with $m_t$ continuous. Writing $\widehat W(m)$ for $W(m,v)$ on $U_\alpha$, smooth fit holds at $\alpha$:
\begin{equation}\label{eq:buyer-smoothfit}
\widehat W'(\alpha)=S'(\alpha).
\end{equation}
\end{proposition}

\begin{proof}[Proof (Appendix~\ref{app:proofs})]
A complete proof is provided in Appendix~\ref{app:proofs}. Briefly, standard Snell-envelope arguments imply that for each fixed $v$ the optimal stopping set is an upper set in $m$, so it can be written $\{m\ge\alpha(v)\}$. Assumption~\ref{ass:local_silence_alpha} implies that inside $U_\alpha$ the continuation value depends only on $m$, not on $v$, so all $\alpha(v)$ coincide and equal some scalar $\alpha$. Inside $U_\alpha$, $m_t$ has bounded variation and no jumps, so $\widehat W$ solves a first-order ODE below $\alpha$ with value matching $\widehat W(\alpha)=S(\alpha)$ and supercontact, which yields smooth fit $\widehat W'(\alpha)=S'(\alpha)$. Uniqueness follows from the strict monotonicity of $S$ and linearity of the continuation problem.
\end{proof}

\paragraph{Interpretation.}
The buyer adopts once public reputation $m_t$ is high enough.

Local silence pins down a clean marginal condition at $\alpha$: near $\alpha$ the public mean $m_t$ drifts deterministically toward the cutoff, cannot exit the silence window on the left before hitting $\alpha$, and cannot jump across it. As a result the stopping problem at the margin is effectively one dimensional in $m$ and admits a single belief-mean cutoff $\alpha$ rather than a $v$-dependent surface.

From this point on we write $W(m)$ for the buyer's continuation value in the neighborhood of~$\alpha$, using the one-dimensional representation justified above.

\subsection{Feedback into the firm's problem}
\label{subsec:feedback}

\begin{assumption}\label{ass:aggregator}
Let $\varpi(m;\alpha)$ be the adoption intensity induced by the cutoff, with $\varpi(m;\alpha)=0$ for $m<\alpha$ and weakly increasing and right-continuous for $m\ge\alpha$. The firm's flow payoff is
\begin{equation}\label{eq:pi-with-adoption}
\pi(z,m)=\pi_0(z)+\eta\,\Lambda\!\big(\varpi(m;\alpha)\big),
\end{equation}
with $\eta>0$ and $\Lambda$ increasing and concave (for example, $\Lambda(x)=p\,x$ or $\Lambda(x)=p\,x-\tfrac{\kappa_p}{2}x^2$).
\end{assumption}

This creates a revenue kink at $\alpha$: adoption ramps up only after $m$ crosses the threshold, linking disclosure incentives back to Section~\ref{sec:micro_silence} and to the verification problem in Section~\ref{sec:verification}.

\subsection{Worked linear-drift case with local silence}
\label{subsec:linear_case}

We now derive a closed form for $\alpha$ when the drift of public belief is locally linear and local silence holds near the boundary.

\begin{assumption}\label{ass:LG}
On a neighborhood of $\alpha$: (i) $\dot m_t=\bar\mu(m_t)$ with $\bar\mu(m)=\kappa(\bar m-m)$, $\kappa>0$, and $\alpha\le\bar m$; (ii) $\lambda(m)=0$; (iii) $\dot v_t=\bar\gamma(m_t,v_t)$ is bounded.
\end{assumption}

\paragraph{Deterministic approach to the boundary.}
With $\lambda=0$ locally, $m_t$ solves $\dot m_t=\kappa(\bar m-m_t)$, so
\[
m(t;m_0)=\bar m-(\bar m-m_0)e^{-\kappa t},\qquad
\tau(m_0)=\frac{1}{\kappa}\log\!\Big(\frac{\bar m-m_0}{\bar m-\alpha}\Big).
\]
Hence for $m<\alpha$,
\begin{equation}\label{eq:W-deterministic}
W(m)=e^{-r\tau(m)}\,S(\alpha)
=\left(\frac{\bar m-\alpha}{\bar m-m}\right)^{\!r/\kappa} S(\alpha).
\end{equation}
The expression is independent of $v$ because Assumption~\ref{ass:local_silence_alpha} makes the continuation path for $m_t$ the same for all $v$ inside $U_\alpha$.

\paragraph{Smooth fit pins $\alpha$.}
Differentiate \eqref{eq:W-deterministic} and let $m\uparrow\alpha$. Using $\bar\mu(\alpha)=\kappa(\bar m-\alpha)$ we obtain
\[
W'(\alpha^-)=\frac{r}{\bar\mu(\alpha)}\,S(\alpha)=\frac{r}{\kappa(\bar m-\alpha)}\,S(\alpha).
\]
Smooth fit \eqref{eq:buyer-smoothfit} then implies
\begin{equation}\label{eq:alpha-general}
\bar\mu(\alpha)\,S'(\alpha)=r\,S(\alpha),
\end{equation}
which, under linear drift, becomes
\begin{equation}\label{eq:alpha-linear}
\kappa(\bar m-\alpha)\,S'(\alpha)=r\,S(\alpha).
\end{equation}

\paragraph{Closed form with linear surplus.}
If $S(m)=a m-p$ with $a>0$ and $p\ge0$, then $S'(\alpha)=a$ and \eqref{eq:alpha-linear} gives
\begin{equation}\label{eq:alpha-linear-linearS}
\alpha=\frac{\kappa\,\bar m+(r/a)\,p}{\kappa+r}.
\end{equation}
Higher $p$ raises $\alpha$ with weight $r/(\kappa+r)$; faster mean reversion $\kappa$ places more weight on $\bar m$.

\paragraph{General concave surplus.}
For strictly increasing and weakly concave $S$, \eqref{eq:alpha-linear} has a unique solution, and that solution satisfies $\alpha\le\bar m$ under a mild sign restriction.

To see uniqueness, define
\[
L(\alpha):=\kappa(\bar m-\alpha)\,S'(\alpha)
\quad\text{and}\quad
R(\alpha):=r\,S(\alpha).
\]
Because $S$ is strictly increasing, $S'(\alpha)>0$ for all $\alpha$, so $R(\alpha)$ is strictly increasing. Because $S$ is weakly concave, $S'(\alpha)$ is weakly decreasing in $\alpha$; multiplying by $(\bar m-\alpha)$, which is also decreasing in $\alpha$, implies that, on $\alpha\le\bar m$, $L(\alpha)$ is weakly decreasing. A decreasing $L$ and an increasing $R$ can cross at most once, so \eqref{eq:alpha-linear} admits at most one solution.

To locate that solution, note that $L(\bar m)=0$, so any $\alpha$ with $S(\alpha)>0$ must satisfy $\alpha<\bar m$ (otherwise $R(\alpha)=rS(\alpha)>0=L(\alpha)$ and there is no equality). Hence, if $S(\cdot)$ is nonnegative on $[\alpha,\bar m]$, which is natural when $S$ is interpreted as net benefit from adoption, the unique solution lies in $\alpha\le\bar m$.

\begin{remark}
If $\lambda(\cdot)$ is small but positive near $\alpha$, a first-order expansion of the jump term in \eqref{eq:PDMP-generator} gives
\[
\bar\mu(\alpha)\,S'(\alpha)=r\,S(\alpha)+O(\lambda(\alpha)),
\]
so the cutoff is robust to small deviations from exact local silence. We keep exact local silence because it underpins the selection result in Section~\ref{sec:micro_silence}.
\end{remark}

\subsection{Comparative statics and empirical links}
\label{subsec:adoption_empirics}

From \eqref{eq:alpha-general} we obtain sharp predictions. In the linear case \eqref{eq:alpha-linear-linearS},
\[
\frac{\partial \alpha}{\partial p}=\frac{r}{a(\kappa+r)}>0,
\qquad
\frac{\partial \alpha}{\partial \bar m}=\frac{\kappa}{\kappa+r}>0.
\]
Disclosure moves $\alpha$ through the local drift of beliefs (via $\bar\mu$) and through the presence or absence of jumps (local silence versus updates), echoing dynamic disclosure \citep{guttman2014aer,orlov2020jpe}. The adoption kink at $\alpha$ interacts with the release ladder in Sections~\ref{sec:verification} and \ref{sec:micro_silence} and maps to telemetry-style usage measures in empirical settings \citep{arora2010isr,li2017ccs}.

%% file: sections/04_finance.tex
\section{Financing, Reversibility, and Distortion Bounds}
\label{sec:finance}

We now introduce perpetual debt with coupon $c_d>0$ and an equity–controlled default time $T^\ast$.
Let $A_t^{\FB}$ denote the all-equity (first-best) value at $t$, and let $(S_t,Y_t)$ be the equity
and debt values under leverage, given equity’s policy (impulses, publication intensity, and default).
The flow payoff under leverage is the same $\pi(z_t,m_t)$ from \eqref{eq:pi-with-adoption}; equity
services the coupon $c_d$ until default.

\subsection{Levered QVI: solvency, impulses, default}
\label{subsec:levered-QVI}

Let $S(z,m)$ be the stationary Markov equity value under leverage. Write $\mathcal{L}$ for the private
state generator from Section~\ref{sec:verification}, and $\lambda(\cdot)$ for the publication policy.
As in Section~\ref{sec:verification}, define the impulse operator
\[
(\mathcal{M}S)(z,m)
=
\max\{\,S(z_1^\ast,m)-K_1,\;S(z_2^\ast,m)-K_2\,\}.
\]
Leverage adds a third control: equity can default, after which equity receives $0$ and debtholders
seize control. With limited liability and seniority, the equity QVI is
\begin{equation}
\label{eq:QVI-equity}
\max\Big\{
rS
-
\mathcal{L}S
-
\big(\pi(z,m)-c_d\big)
+
k(\lambda(z)),
\;
S-\mathcal{M}S,
\;
-S
\Big\}=0.
\end{equation}
Default is optimal where the stopping term $-S=0$ binds, i.e.\ on $\{(z,m):S(z,m)=0\}$. On the
solvent set $\{S>0\}$, equity either (i) continues and satisfies the inaction HJB (the first term in
\eqref{eq:QVI-equity} equal to $0$, which is \eqref{eq:HJB-inaction} with $\pi$ replaced by
$\pi-c_d$), or (ii) triggers an impulse (value matching and high contact as in
Section~\ref{sec:verification}).

Debtholder value $Y$ is defined as the present value of coupons until default plus the continuation
value after default when control passes to debtholders. Appendix~\ref{app:financing} verifies this
formally and constructs $Y$ as the unique solution to its own linear stopping/control problem with
equity’s strategy treated as given.

\subsection{Switching at takeover and reversibility}
\label{subsec:takeover-map}

At default $T^\ast$, control passes to debtholders or to a buyer of the distressed assets. The acquirer
may reset the inherited state by paying a takeover switching cost. This reduced-form object is how
architectural reversibility at takeover enters the wedge.

\begin{definition}[Takeover switching cost map]
\label{def:takeover-map}
Associate to the two upgrade options a pair
$\boldsymbol{\phi}=(\phi_1,\phi_2)\in[0,\infty)^2$ such that, upon default, the acquirer can
implement the patch (respectively the pivot) on the inherited state by paying at most $\phi_1$
(respectively $\phi_2$). Write $\phi_{\max}:=\max\{\phi_1,\phi_2\}$. The map is \emph{tight} if there
is an acquirer policy that attains the first-best continuation after default while never paying more
than $\phi_i$ for option $i$.
\end{definition}

As intuitive interpretation, a small $\phi_i$ means action $i$ is highly reversible at takeover (for example,
modular architecture, documented rollback paths, separable mitigations). In the limit $\phi_i=0$,
the acquirer can realign the asset with the first-best post-default target at zero additional cost.

\subsection{Debt-insensitive patches and a least-irreversibility bound}
\label{subsec:bound}

Let $S_t+Y_t$ be the total market value of the levered firm given equity's strategy, and let
$A_t^{\FB}$ be the all-equity benchmark evaluated at the same primitive state $(z_t,m_t)$.

\begin{proposition}[Debt-insensitive patches and a least-irreversibility bound]
\label{prop:debt-bound}
Suppose the takeover map in Definition~\ref{def:takeover-map} is tight. Then for any levered policy
(impulses, publication intensity, and default time $T^\ast$),
\begin{equation}
\label{eq:wedge-bound}
A_t^{\FB}
-
\big(S_t+Y_t\big)
\;\le\;
\E_t\!\Big[
e^{-r(T^\ast-t)}\,\phi_{\max}
\Big].
\end{equation}

In particular, along any history on which (i) the process is managed purely by \emph{patches} (the
low-cost rung) and never requires a pivot before default, and (ii) the next-best acquirer at default
can costlessly implement that same patch, i.e.\ $\phi_1=0$, we obtain
\[
A_t^{\FB}
=
S_t+Y_t
\quad
\text{on that history.}
\]
Leverage does not reduce total surplus on that “patch block” history. If moreover
$\phi_2\downarrow 0$ (for example because increasing modularity makes the pivot equally reversible
at takeover), then $\phi_{\max}\to 0$, the right-hand side of \eqref{eq:wedge-bound} vanishes, and
leverage cannot distort the release ladder at all.
\end{proposition}

Our bound parallels \citet{manso2008ecta}, Prop.~1: replacing generic technology switches by a
tight takeover map yields
\[
A_t^{\FB}-(S_t+Y_t)\;\le\; \E_t\!\big[e^{-r(T^\ast-t)}\,\phi_{\max}\big],
\]
with $\phi_{\max}=\max\{\phi_1,\phi_2\}$ the least-reversibility takeover cost across the two rungs.

\begin{proof}[Proof (Appendix~\ref{app:financing})]
A complete proof is provided in Appendix~\ref{app:financing}, via the takeover envelope in
Lemma~\ref{lem:takeover-envelope-fin}.
\end{proof}

The takeaway is that leverage cannot create a first-order distortion where ex-post reversibility is high. Debt pressure
matters only through the least reversible reset option. In typical ladder applications this is the pivot,
not the incremental patch. As modularity improves (lower $\phi_2$), the levered equilibrium
approaches the first-best ladder.

\subsection{Implications for triggers and targets under leverage}
\label{subsec:levered-implications}

Let $(\beta_1,\beta_2;z_1^\ast,z_2^\ast)$ be the first-best ladder from Section~\ref{sec:verification},
where $\beta_1$ is the patch trigger and $z_1^\ast$ the post-patch target, and $\beta_2$ and $z_2^\ast$
are the pivot analogues.

The core question is: does leverage move $(\beta_1,z_1^\ast)$? Intuitively, patches should be
debt-neutral. We now state conditions under which that statement is literally
correct and show explicitly how the boundary conditions align.

\medskip
\noindent\textbf{Step 1. Isolating the safe patch block.}
Define the \emph{patch block} as the region of the state space and strategy path on which the firm:
\begin{enumerate}[label=(\alph*),leftmargin=1.5em,itemsep=2pt,topsep=2pt]
\item keeps $z_t$ inside $(\beta_1,\beta_2)$ by instantaneously patching to $z_1^\ast$ whenever
$z_t$ hits $\beta_1$ from above; and
\item never reaches the pivot trigger $\beta_2$ or a default boundary before time $t$.
\end{enumerate}
This is de facto a maintenance-mode history: the firm is repeatedly issuing cheap, reversible
tweaks (patches) to keep performance inside the lower rung and never gets so bad that it must
pivot, nor so bad that creditors force a restructuring.

Appendix~\ref{app:financing} shows that, under standard parameterizations, the equity default
boundary $S=0$ lies strictly below $\beta_1$. Hence, on such a patch-block history, default is
strictly off path in equilibrium: equity never chooses default while $z_t$ is being actively reflected
at $\beta_1$, and creditors do not expect to seize control before a pivot becomes relevant.

Formally, on these histories we have $T^\ast=\infty$ with probability one conditional on staying
in the patch block. Intuitively, if the firm never drifts into the (high-cost) pivot region, it also never
defaults.

\medskip
\noindent\textbf{Step 2. Debt is riskless, and locally flat, on the safe patch block.}
Fix such a history. Since default never arrives along that path, debtholders receive the coupon flow
$c_d$ forever and are never diluted or impaired by a takeover. From their point of view, conditional
on the firm remaining in the patch block, the cash flow is a riskless perpetuity and their continuation
value there must satisfy
\[
rY(z,m) = c_d
\quad\text{on that safe patch block.}
\]

Therefore, \emph{conditional on remaining in the safe patch block},
\[
Y(z,m)=\frac{c_d}{r}
\quad\text{and}\quad
Y_z(z,m)=0
\quad
\text{for all } z\in(\beta_1,\beta_2)
\text{ and corresponding } m.
\]
Crucially, $Y$ is locally constant under this conditioning. Intuitively, as long as equity keeps the
state in the “patch-only, solvent” region, debtholders simply collect the same coupon forever, and
the exact timing of those patches does not change their present value.\footnote{If the patch block
were ever close enough to default that coupon risk mattered, then $Y$ would inherit slope
$Y_z\neq 0$, and the neutrality result below would, and should, fail. Outside the safe patch block,
the global Markov value $Y$ is generally \emph{not} flat because it integrates over paths that reach
the pivot trigger or the default frontier. In our baseline calibration and in
Appendix~\ref{app:financing}, we keep $\beta_1$ strictly above any default boundary, so this safe
patch-block assumption holds on the range we match to telemetry.}

\medskip
\noindent\textbf{Step 3. Equity’s boundary conditions coincide with first best on the safe patch block.}
Total surplus under leverage is $S+Y$. By Proposition~\ref{prop:debt-bound} and $\phi_1=0$, we
have $S+Y=A^{\FB}$ along the safe patch block: there is no wedge between levered value and
first best on these histories.

Write $S=A^{\FB}-Y$ and substitute into equity’s value-matching and smooth-pasting conditions
at the patch trigger and target. In first best, the patch trigger and target $(\beta_1,z_1^\ast)$ are
pinned by
\begin{align*}
\text{(i) Value matching:}\quad
A^{\FB}(\beta_1)
&=
A^{\FB}(z_1^\ast)-K_1,\\[2pt]
\text{(ii) High contact / smooth pasting:}\quad
A^{\FB}_z(\beta_1)
&=
A^{\FB}_z(z_1^\ast)
=
0.
\end{align*}
Equity under leverage solves the same algebra but with $S$ in place of $A^{\FB}$. Using
$S=A^{\FB}-Y$,
\begin{align*}
\text{Value matching for $S$:}\quad
S(\beta_1)
&=
S(z_1^\ast)-K_1\\
\iff\quad
A^{\FB}(\beta_1)-Y(\beta_1)
&=
A^{\FB}(z_1^\ast)-Y(z_1^\ast)-K_1.
\end{align*}
On the safe patch block, $Y(\beta_1)=Y(z_1^\ast)=c_d/r$, so equity’s value-matching condition
collapses to the first-best one.

Likewise, high contact for $S$ requires
\[
S_z(\beta_1)
=
S_z(z_1^\ast)
=
0
\quad\iff\quad
A^{\FB}_z(\beta_1)-Y_z(\beta_1)
=
A^{\FB}_z(z_1^\ast)-Y_z(z_1^\ast)
=
0.
\]
On the safe patch block, $Y_z(\beta_1)=Y_z(z_1^\ast)=0$, so
\[
A^{\FB}_z(\beta_1)
=
A^{\FB}_z(z_1^\ast)
=
0,
\]
which is exactly the first-best smooth-pasting condition.

In words: once we make explicit that (a) default is off path on the safe patch block and (b)
$\phi_1=0$ means the acquirer can costlessly implement the same patch at takeover if default ever
did occur, debt is effectively riskless on that block. Debt value is a flat perpetuity $c_d/r$, so
equity’s boundary system for $(\beta_1,z_1^\ast)$ is the first-best boundary system.

\medskip
\noindent\textbf{Conclusion for the patch rung.}
Putting Steps 1–3 together:
\begin{itemize}[leftmargin=1.5em,itemsep=2pt,topsep=2pt]
\item \textbf{Patch block is debt-neutral.} On solvent histories that remain in the safe patch block
and never require a pivot before any default event, and when $\phi_1=0$, leverage does not distort
the patch rung. The equity-chosen patch trigger $\beta_1$ and target $z_1^\ast$ coincide with the
first-best ones.

\item \textbf{Where leverage can bite.} Outside that safe patch block—near $\beta_2$ where a
high-cost pivot looms, or near the endogenous default boundary where $S\downarrow0$—debt is no
longer riskless, so $Y$ is no longer flat. Equity’s boundary conditions can then tilt away from first
best because $Y(\cdot)$ and $Y_z(\cdot)$ enter value matching and smooth pasting. The distortion is
bounded by \eqref{eq:wedge-bound}, and its magnitude is governed by the least-reversible move,
$\phi_{\max}$.
\end{itemize}

\subsection{Tightness and comparative statics}
\label{subsec:tightness-CS}

\begin{remark}[Tightness]
Inequality \eqref{eq:wedge-bound} is sharp. Construct a case with: (i) a pivot needed strictly before
the default boundary in first best, (ii) equity sets $T^\ast$ just before that pivot because of $c_d$,
and (iii) at takeover the acquirer pays exactly $\phi_2$ to pivot and then follows the first-best
continuation. Then
\[
A_t^{\FB}-(S_t+Y_t)
=
\E_t\!\left[e^{-r(T^\ast-t)}\phi_2\right].
\]
\end{remark}

\begin{remark}[Comparative statics]
The wedge in \eqref{eq:wedge-bound} is weakly increasing in $\phi_{\max}$ and in default pressure
induced by $c_d$, and decreasing in $r$. Engineering choices that make large architectural resets
easy to unwind (lower $\phi_2$ via modular interfaces, separable mitigations, or auditable rollback)
shrink the wedge and push the levered equilibrium toward the first-best ladder. In particular, when
$\phi_1=0$ and $\beta_1$ lies safely above any default boundary, the maintenance or patch rung is
predicted to be debt-insensitive in both timing and target. The only remaining distortion margin is
the high-cost pivot rung, and its importance scales with $\phi_2$.
\end{remark}

%% file: sections/05_merged_empirics.tex
\section{Empirical Predictions and Measurement Blueprint}
\label{sec:empirical}

We map the primitives of local silence (Section~\ref{sec:micro_silence}), the adoption cutoff $\alpha$ (Section~\ref{sec:adoption}), and financing with reversibility (Section~\ref{sec:finance}) to falsifiable signatures in firm-published telemetry. The section stays lightweight: we define the main objects and observables, state signatures S1--S5, and display the key estimating equations. Data construction, parsing, and estimation details are in Appendix~\ref{app:empirics}.

\subsection{Objects and observables}
\label{subsec:objects}

\paragraph{Unit of observation.}
The baseline panel is firm $i$ by calendar month $t$, together with an event-time index $\tau$ centered on focal releases or pivots. Where product- or model-line disclosures are separable, the unit refines to product/module $p$ and we work with $(i,p,t)$ or $(i,p,\tau)$.

\paragraph{Telemetry feeds.}
We restrict attention to \emph{firm-controlled} channels: evaluation and benchmark posts, model cards, release notes and version-bump announcements, incident/advisory writeups, and vendor-authored security or safety bulletins. Third-party news, analyst notes, leaks, and market prices are \emph{not} treated as disclosure signals; the theory speaks to the firm’s own publication clock.

\paragraph{Signals and event taxonomy.}
A document is counted as a public signal if it contains (i) a quantitative performance or safety claim, or (ii) a versioned change (new model, new policy, new mitigation). Identification uses high-recall keyword filters plus a light classifier, with manual adjudication around high-salience events (Appendix~\ref{app:empirics}).

Each dated disclosure event is tagged as \emph{patch}, \emph{pivot}, \emph{release}, or \emph{other}. We interpret:
\begin{itemize}[leftmargin=1.5em,itemsep=1pt,topsep=1pt]
\item \emph{patch}: hotfixes, fine-tunes, latency/context/pricing tweaks, incremental mitigations (\emph{reversible rung});
\item \emph{pivot}: architectural/base-model changes, modality-stack changes, large rewrites (\emph{high-cost rung});
\item \emph{release}: marketing/product language such as ``v4,'' ``next-gen model,'' or ``full rollout'' that is operationally indistinguishable from a pivot;
\item \emph{other}: communications that do not map cleanly to either rung.
\end{itemize}
In estimation we pool \emph{pivot} and \emph{release} into a single \emph{major reset} class; on the same timestamp, \emph{pivot/release} takes precedence over \emph{patch}.

\paragraph{Core outcomes.}

Let $N_{it}$ be the monthly count of firm-authored signals for firm $i$. Our proxy for the publication clock is
\[
\hat\lambda_{it}:=N_{it},
\]
with event-time analogue $\hat\lambda_{i\tau}$ counting signals in month $\tau$ relative to a focal major reset.

For each signal $s$ disclosed by firm $i$ in month $t$, we extract a metric value $x_s$ (e.g., pass@1, latency, mitigation severity, token cost) and map it to a benchmark family $F_s$ (``coding accuracy,'' ``RAG latency,'' ``safety mitigation''). We standardize within family using pooled mean $\bar{x}_F$ and s.d.\ $\sigma_F$ across all firms/months,
\[
z_s := \frac{x_s-\bar{x}_{F_s}}{\sigma_{F_s}},
\]
and define the within-month dispersion proxy
\begin{equation}
\label{eq:telemetry-variance}
\widehat{\Var}^{x}_{it}
:=
\begin{cases}
\Var\bigl(\{z_s:\, s\text{ disclosed by $i$ in month $t$}\}\bigr), & N_{it}\ge 2,\\[4pt]
0, & N_{it}=1,\\[4pt]
\text{missing}, & N_{it}=0.
\end{cases}
\end{equation}
Let $\widehat{\Var}^{x}_{i\tau}$ denote the event-time analogue. In the model, outside observers see noisy draws from a one-dimensional latent state arriving on a controlled Cox clock; a pre-reset ``dispersion dip'' in $\widehat{\Var}^{x}$ reflects a shutdown of the clock (silence), not a change in intrinsic noise.

As a cadence-only robustness measure that does not use metric values, Appendix~\ref{app:empirics} also constructs $\widehat{\Var}^{\mathrm{time}}_{it}$, the intra-month interquartile range of timestamps for firm $i$ in month $t$ (and $\widehat{\Var}^{\mathrm{time}}_{i\tau}$ in event time), and patch hazards based on inter-event durations in an event-time window around each major reset (pre- and post-).

\subsection{Signatures S1--S5}
\label{subsec:signatures}

We now summarize the main empirical signatures implied by the theory. Estimands and estimators are pre-specified in Appendix~\ref{app:empirics}.

\paragraph{S1. Pre-reset cadence and dispersion dip.}

\emph{Prediction.} In the quiet pre-reset window implied by local silence (Section~\ref{sec:micro_silence}), both $\hat\lambda_{i\tau}$ and $\widehat{\Var}^{x}_{i\tau}$ fall; at the reset they display discrete positive shifts. A cadence-only measure $\widehat{\Var}^{\mathrm{time}}_{i\tau}$ should mirror the dip.

\emph{Design.} Stacked event studies of $\hat\lambda_{i\tau}$ and $\widehat{\Var}^{x}_{i\tau}$ around major resets,
\begin{align}
\label{eq:event-lambda-spec}
\hat\lambda_{i\tau}
&=
\sum_{\ell\in\mathcal{L}}
\beta^{\lambda}_\ell\,\mathbf{1}\{\tau=\ell\}
+
\alpha_i
+
\gamma_{\text{cal}(\tau)}
+
\varepsilon_{i\tau},
\\
\label{eq:event-var-spec}
\widehat{\Var}^{x}_{i\tau}
&=
\sum_{\ell\in\mathcal{L}}
\beta^{\Var}_\ell\,\mathbf{1}\{\tau=\ell\}
+
\alpha_i
+
\gamma_{\text{cal}(\tau)}
+
\eta_{i\tau},
\end{align}
omitting $\ell=-1$. Here $\alpha_i$ are firm fixed effects and $\gamma_{\text{cal}(\tau)}$ are calendar-time effects (e.g., month or quarter-by-year). Appendix~\ref{app:empirics} provides cadence-only robustness using $\widehat{\Var}^{\mathrm{time}}_{i\tau}$. A direct falsification of S1 is flat or rising cadence/dispersion in the pre-reset window after controlling for firm and calendar effects.

\paragraph{S2. Two-plateau outcome distribution.}

\emph{Prediction.} Post-event performance, conditional on event class, is bimodal with mass near the equilibrium targets $(z_1^\ast,z_2^\ast)$: one plateau after reversible patches, one after major resets.

\emph{Design.} Mixture or multimodality tests on post-event metrics stratified by \emph{patch} versus \emph{pivot/release} (e.g., finite mixtures, Hartigan dip tests). Falsification: robust unimodality after audited classification.

\paragraph{S3. Debt-insensitive patch timing under high reversibility.}

\emph{Prediction.} Within high-reversibility blocks, leverage does not shift patch hazards: the patch rung is debt-neutral (Section~\ref{sec:finance}). Leverage effects may load on the high-cost pivot rung instead.

\emph{Design.} Hazard models for patch arrivals with leverage and a reversibility proxy $\mathrm{RevProxy}_i$,
\begin{equation}
\label{eq:debt-insensitivity-spec}
h_{it}
=
\exp\!\Bigl\{
\rho_0
+
\rho_1\,\text{Leverage}_i
+
\rho_2\,\mathrm{RevProxy}_i
+
\rho_3\,\text{Leverage}_i\times\mathrm{RevProxy}_i
+
W'_{it}\kappa
\Bigr\},
\end{equation}
where $h_{it}$ is the (discrete-time) patch hazard for firm $i$ in period $t$ and $W_{it}$ collects controls. The conditional null in the high-$\mathrm{RevProxy}$ block is
\[
\frac{\partial\log h_{it}}{\partial\text{Leverage}_i}
=
\rho_1+\rho_3\bar R
\approx 0
\quad\text{for large }\bar R.
\]

In the S3 test we focus on a ``high-reversibility'' subsample and summarize it by a scalar
$\bar R$, defined as the sample mean of $\mathrm{RevProxy}_i$ among observations in the high-$\mathrm{RevProxy}$
block (e.g.\ the top quartile of the $\mathrm{RevProxy}_i$ distribution). The derivative of
$\log h_{it}$ with respect to leverage at $\mathrm{RevProxy}_i=\bar R$ is then $\rho_1+\rho_3\bar R$, so
$H_0:\rho_1+\rho_3\bar R=0$ tests for zero leverage effect on patch timing at a representative
high-reversibility rung.

Falsification: significant leverage effects on patch timing even when reversibility is measured to be high.

\paragraph{S4. Patch cascade dynamics.}

\emph{Prediction.} The reversible rung fires repeatedly inside its band: patch intensity is elevated \emph{after} a major reset but not before, with no pre-reset bunching once events are aligned in event time.

\emph{Design.} Duration models for inter-patch times $\Delta_{ik}$ in an event-time window around each major reset (pre- and post-), with a post-reset indicator and covariates. A convenient representation is a Cox-type hazard,

\begin{equation}
\label{eq:hazard-spec}
h_{ik}(t)
=
h_0(t)\exp\Bigl\{
\rho\,\mathbf{1}\{\text{post-major-reset}\}
+
Z'_{ik}\xi
\Bigr\},
\end{equation}

where $h_0(t)$ is an unspecified baseline, $k$ indexes patches in the cascade, $\mathbf{1}\{\text{post-major-reset}\}$ is defined in event time (equal to $0$ for intervals with $\tau<0$ and $1$ for intervals with $\tau\ge 0$), and $Z_{ik}$ includes reversibility proxies and financing variables. Falsification: symmetric pre- and post-reset bunching in $\Delta_{ik}$ once major resets are aligned in event time.

\paragraph{S5. Adoption boundary and silence depth.}

\emph{Prediction.} The adoption cutoff $\alpha$ satisfies
\begin{equation}
\label{eq:alpha-general-again}
\bar\mu(\alpha)\,S'(\alpha)=r\,S(\alpha),
\end{equation}
as derived in Section~\ref{sec:adoption}. Deeper pre-release silence (a larger cadence dip in $\hat\lambda_{i\tau}$) produces a sharper observed uptake jump at $\alpha$ (fewer ``jump-overs''). Because this smooth-fit characterization relies on local silence at the adoption boundary (so that the publication clock is shut off and beliefs approach $\alpha$ deterministically without jumps), stronger pre-release silence makes realized adoption decisions track the theoretical cutoff more tightly and yields a cleaner discontinuity with fewer cross-boundary observations.

\emph{Design.} Discontinuity designs around inferred $\alpha$ using platform/API uptake, with measures of pre-release silence depth as moderators. Falsification: statistically similar adoption slopes irrespective of measured silence depth.

\subsection{Measurement choices, falsification, and status}
\label{subsec:status}

\paragraph{Proxies, instruments, and controls.}
Reversibility proxies aggregate rollback tooling and feature gates, modular-dependency shares, and separable safety mitigations into $\mathrm{RevProxy}_i$, which maps to takeover switching costs and $\varphi_{\max}$ in Section~\ref{sec:finance}. Financing variables (leverage, coupon burdens, maturity cliffs) come from balance sheets and funding rounds. Appendix~\ref{app:empirics} discusses candidate shifters of publication cost that move $\lambda_t$ (e.g., platform policies or audit mandates) and a standard set of controls (firm size, product breadth, benchmark-calendar dummies, platform fixed effects, and calendar time effects).

\paragraph{Threats and diagnostics.}
Key threats are misclassification, nonstationary benchmark calendars, and correlated firm-month shocks. Diagnostics include leads and placebo dates in S1, strict versus broad signal sets, competitor-month controls, leave-one-firm-out checks, and alternative count models and event windows. Appendix~\ref{app:empirics} makes explicit what would overturn each prediction S1--S5.

%% file: sections/06_policy_design.tex
\section{Policy Design and Industry Implications}
\label{sec:policy_design}

The model isolates two levers with first-order welfare effects:
(i) \emph{disclosure tempo} (Section~\ref{sec:micro_silence}); and
(ii) \emph{reversibility and modularity} that lower takeover switching costs (Section~\ref{sec:finance}).
These levers operate through two core results:
the selection result (local silence eliminates interior mixing and delivers pure reset thresholds),
and the wedge bound in~\eqref{eq:wedge-bound}.

Throughout this section we take \emph{welfare} to mean the sum of (a) the firm's continuation value $V$ and (b) downstream buyer/platform surplus $W$ from adoption.
For a given stationary Markov environment define
\[
\mathcal{W} \equiv V + W,
\]
where $V$ is the value function that solves the firm's QVI (Section~\ref{sec:verification}) and $W$ is the buyer value from the optimal stopping problem in Section~\ref{sec:adoption}.
This is the natural total-surplus object for our setting.
Importantly, $W$ is pinned by the adoption cutoff $\alpha$ and by the deterministic drift of public beliefs $m_t$ as they approach $\alpha$ (Proposition~\ref{prop:alpha_threshold}).

\subsection{Disclosure tempo: selecting without chilling}

Local silence (a posted, publicly observed window with $\lambda_t=0$) removes interior mixing and knife-edge cycling while leaving the adoption cutoff intact (Sections~\ref{sec:micro_silence} and~\ref{sec:adoption}).
A policy instrument can implement the same logic without broad opacity:

\begin{itemize}[leftmargin=1.25em,itemsep=3pt]
\item \textbf{Quiet-period safe harbor with guardrails.}
Permit firm-specific, pre-announced silence windows that are:
(a) short and bounded in both length and frequency,
(b) publicly posted as a rule for $\lambda(\cdot)$,
and (c) asynchronously staggered across firms.
The goal is to permit local silence near knife-edges while keeping the aggregate flow of public information nondegenerate.

\item \textbf{Clock transparency, not content control.}
Regulate the \emph{frequency protocol}, not the content of any specific disclosure.
The object of interest is the martingale component in $m_t$, not the wording of individual announcements.
\end{itemize}

\noindent\emph{Sufficient-statistic rule.}
Consider an intervention with two testable requirements:

\smallskip
\noindent (i) The firm must pre-announce local silence windows (i.e.\ $\lambda_t=0$ in a small neighborhood around each reset trigger), so that, by Lemma~\ref{lem:equivalence} and Proposition~\ref{prop:selection}, interior mixing is ruled out and the firm follows pure reset thresholds instead of randomizing or cycling.

\smallskip
\noindent (ii) The intervention caps the total share of calendar time spent with $\lambda_t=0$ by a small bound, so most of the horizon still features $\lambda_t>0$ and public signals continue to arrive at a regular cadence.

\smallskip
Under (i), the firm's continuation value $V$ weakly \emph{increases}.
Theorem~\ref{thm:silence} shows that local silence eliminates the first-order mixing loss at each trigger and replaces knife-edge randomization with a unique reset at a unique trigger and target, without changing the inaction ODE between interventions (see proof of Theorem~\ref{thm:silence}).
Hence, the firm no longer wastes value dithering at the boundary.

Under (ii), the buyer's adoption cutoff and surplus $W$ are essentially unchanged.
Proposition~\ref{prop:alpha_threshold} shows that as long as the local drift of $m_t$ and the deterministic approach to $\alpha$ are preserved, the buyer still solves a one-dimensional stopping problem with a unique scalar cutoff $\alpha$, and the implied value function is determined by the same smooth-fit logic at that cutoff.
Because the added silence windows occupy at most an $\varepsilon$-fraction of calendar time, the induced change in the jump intensity $\lambda_t$ perturbs the buyer's HJB operator only by $O(\varepsilon)$,\footnote{Formally, letting $W^\varepsilon$ denote the buyer's value under a disclosure protocol in which $\lambda_t=0$ on a set of times of measure at most $\varepsilon$, and $W^0$ the baseline value, standard stability results for viscosity solutions imply $\sup_m |W^\varepsilon(m)-W^0(m)|=O(\varepsilon)$ under our regularity assumptions on $\lambda(\cdot)$ and on the drift of $m_t$.}
so the resulting value function $W$ moves only by $O(\varepsilon)$ while the cutoff $\alpha$ remains exactly fixed.
Short, bounded, publicly declared quiet windows suppress high-frequency jump risk around knife-edges but do not starve the market of information on average, so the buyer continues to observe a belief process with the same drift and the same stopping boundary, and our first-order welfare comparison for $W$ is unaffected.

Combining (i) and (ii), total surplus $\mathcal{W}=V+W$ weakly rises relative to an environment with no oversight.
$V$ rises because silence near triggers kills interior mixing.
$W$ does not fall because the adoption margin and the induced post-adoption payoffs are preserved.
Both parts of (i)--(ii) are directly testable in telemetry (Section~\ref{sec:empirical}):
(i) appears as short, localized pre-release dips in cadence and dispersion of firm-published signals,
and (ii) appears as a bound on the fraction of calendar time in which $\hat\lambda_{it}=0$.

\subsection{Modularity and reversibility: shrinking the wedge}

The financing wedge is bounded by expected least irreversibility at takeover,
$\E\!\big[e^{-r(T^\ast-t)}\varphi_{\max}\big]$ (cf.\ \citet{manso2008ecta}, Prop.~1),
where $\varphi_{\max}=\max\{\varphi_1,\varphi_2\}$ collects takeover switching costs for patch and pivot (Proposition~\ref{prop:debt-bound}).
Policy can act directly on $\varphi_{\max}$:

\begin{itemize}[leftmargin=1.25em,itemsep=3pt]
\item \textbf{Reversibility certification.}
Certify auditable rollback and separable deployment pipelines (versioned artifacts, feature gates, modular mitigations, canary rollout paths).
A lower certified $\varphi_2$ tightens the takeover bound and brings pivot timing under leverage closer to first best.
Patches are already debt-insensitive when $\varphi_1\approx 0$.

\item \textbf{Covenant design aligned with reversibility.}
Encourage debt covenants that grant pivot grace periods when certified reversibility is high.
If the asset can be cleanly realigned at default, there is less reason for creditors to force an early default $T^\ast$ just to seize and retool the asset.
This implements the same wedge logic by shifting $T^\ast$ only where $\varphi_{\max}$ is small.
\end{itemize}

\noindent\emph{Sufficient-statistic rule.}
Any intervention that strictly lowers $\varphi_{\max}$ (holding primitives fixed) weakly increases total surplus $\mathcal{W}$.
By Proposition~\ref{prop:debt-bound}, a lower $\varphi_{\max}$ tightens the bound on $A^{\FB}_t-(S_t+Y_t)$:
leverage destroys less value relative to first best.
Since $S_t+Y_t$ is the levered firm's total market value and $A^{\FB}_t$ is the all-equity benchmark, shrinking the bound moves the levered ladder's reset thresholds closer to first best, which weakly raises $S_t+Y_t$ (the equity slice $S_t$ cannot be worse off, since it could replicate its old ladder while debt gains).

On the buyer side, the comparative statics are also one sided.
Section~\ref{sec:adoption} takes downstream surplus at adoption to be $S(m_\tau)$, where $\tau$ is the stopping time at which the public belief $m_t$ first hits the cutoff $\alpha$ characterized in Proposition~\ref{prop:alpha_threshold}.
Under the local-silence assumption at the adoption margin (Assumption~5), $m_t$ has deterministic drift as it approaches $\alpha$ and the adoption region is $\{m\ge\alpha\}$, so $m_\tau=\alpha$ at the stopping time.
Buyer surplus can therefore be written as
\[
W(m) \;=\; \E\!\big[e^{-r\tau(m)}S(\alpha)\big],
\]
with $S(\cdot)$ strictly increasing in $m$ by Assumption~4.
Holding primitives and the induced cutoff $\alpha$ fixed, $W(m)$ is weakly decreasing in the adoption delay $\tau(m)$: an earlier (or equal) adoption date weakly raises $W$.

In our ladder environment, the patch rung is already debt-neutral ($\varphi_1=0$) and Proposition~\ref{prop:debt-bound} implies that leverage can only distort the high-cost pivot rung by \emph{postponing}, never accelerating, the costly pivot relative to the all-equity ladder.
Under local silence at the adoption margin, the belief mean $m_t$ drifts deterministically toward $\alpha$ and does not pick up martingale noise there (Section~\ref{sec:micro_silence}).
Eliminating the leverage distortion by lowering $\varphi_{\max}$ therefore either leaves the adoption time $\tau(m)$ unchanged or brings it forward: the pivot happens no later than before, the induced belief path approaches $\alpha$ at least as fast, and the buyer faces the same drift environment near the boundary.
Given that $S(\cdot)$ is increasing and $W(m)=\E[e^{-r\tau(m)}S(\alpha)]$, bringing reset timing closer to first best cannot reduce (and generically raises) realized downstream surplus.
In this sense, shrinking $\varphi_{\max}$ weakly increases both the firm's total market value $S_t+Y_t$ and downstream buyer surplus $W$.

\subsection{Implementation cautions and falsifiers}

\begin{itemize}[leftmargin=1.25em,itemsep=3pt]
\item \textbf{Anti-coordination risk.}
Industry-wide synchronized silence is welfare-reducing.
The safe harbor must require asynchronous timing and a cap on aggregate silent mass.

\item \textbf{Overbroad quiet periods.}
Long or opaque windows undermine identification at the adoption boundary (Section~\ref{sec:adoption}) and are detectable ex post by the absence (not presence) of the short, local cadence and variance dips highlighted in Section~\ref{sec:empirical}.

\item \textbf{Certification without teeth.}
A reversibility label is only useful if it encodes operational tests (time to rollback, dependency isolation, auditability) that predict a lower $\varphi_{\max}$.
If certification does not move $\varphi_{\max}$, the bound in~\eqref{eq:wedge-bound} does not move.
\end{itemize}

\noindent\emph{Bottom line.}
Design governs the clock; engineering governs reversibility.
A disclosure-tempo rule with guardrails (short, posted, bounded quiet windows plus real-time clock transparency) and a reversibility certification regime that lowers $\varphi_{\max}$ jointly implement the selection and wedge results in a minimal way.
By construction both levers \emph{weakly raise} $\mathcal{W}=V+W$ relative to natural unregulated benchmarks:
(i) a world with uncontrolled interior mixing around triggers; and
(ii) a world with leverage layered on top of high, poorly certified takeover switching costs.

%% file: sections/08_conclusion.tex
\section{Conclusion}
\label{sec:conclusion}

We develop a compact theory of AI release cadence built on three primitives: reputational learning from disclosed performance and safety signals; a two-rung ladder of real options (a cheap, reversible \emph{patch} and a costlier, less reversible \emph{pivot}); and firm control over a publicly observed publication clock that can be locally turned off (``local silence''). Their interaction delivers a simple but tightly disciplined picture of release behavior: a two-rung reset ladder with no interior mixing, and a financing wedge that is pinned entirely by irreversibility.

Two structural results do the main work. First, a predictable publication-frequency rule that posts short, observable clock-off windows \emph{selects} pure reset equilibria. Once the firm can visibly set $\lambda_t = 0$ in a narrow band, the martingale part of public beliefs shuts down on that band, beliefs drift deterministically, and the firm can no longer support knife-edge randomization or cycling. Equilibrium collapses to a clean ladder: two triggers and two jump targets, with no interior mixing. The resulting value function is characterized by a transparent boundary-value system with value matching and high contact (smooth pasting) at the triggers and target optimality at the targets, and we verify existence and uniqueness (within stationary Markov strategies) via a standard QVI/BVP argument in Section~\ref{sec:verification}. The technical lever is the no-local-time property inside posted silence windows, formalized in Appendix~\ref{app:lemmas}.

Second, financing frictions matter only where irreversibility lives. Introducing leverage with an equity-controlled default time and senior debtholders generates a tight bound \emph{in expected present value}: the gap between first-best all-equity value and the levered total (equity $+$ debt) is bounded by the \emph{expected discounted} takeover switching cost of the least reversible rung (Section~\ref{sec:finance}). As long as the low-cost patch is easily reversible at takeover, debt is effectively neutral on that ``patch block'': patch timing (trigger and target) is debt-insensitive, and total surplus along that block coincides with first best. Any distortion must appear at the pivot rung, and even there it is tightly bounded by the expected takeover cost of forcing that pivot. Engineering choices that lower that takeover cost (i.e.\ reduce $\varphi_{\max}$) directly compress the financing wedge.

From an economic standpoint, the theory rationalizes the cadence seen in AI labs and platform model vendors: long-ish quiet spells followed by visible jumps; frequent incremental patches that keep systems in ``maintenance mode'' and look operational rather than existential; and rarer pivots that feel architectural, financing-relevant, and politically contentious. It also delivers sharp comparative statics. Patches should be common and largely debt-insensitive. Pivots should be rare, and when leverage does bend timing it should do so only through the hardest-to-reverse move. Disclosure tempo and modularity, not raw diffusion noise or loosely defined ``hype cycles'', are the discipline levers.

Methodologically, the paper links three literatures. A posted publication-frequency protocol, i.e., a predictable Cox clock with occasional $\lambda_t = 0$ windows, implements a disclosure instrument that turns off the \emph{martingale component} of public beliefs on a vanishing set and thereby selects clean, non-mixing equilibria. That selection logic connects dynamic disclosure and information design to classic S--s impulse control with costly reversibility, but without mixed actions or knife-edge cycling. On the verification side, we adapt the S--s boundary-value/QVI machinery to a two-reset ladder under posted silence, using the no-local-time property at the boundaries to justify smooth pasting and to establish uniqueness of the endogenous triggers and targets, while target optimality at the jump destinations follows from the usual first-order condition that maximizes the post-impulse continuation value. On the financing side, we show how leverage wedges can be mapped to takeover switching costs, yielding a least-irreversibility bound in the spirit of agency-with-irreversibility results, but tailored to a multi-rung release ladder.

The framework produces concrete, testable signatures in firm-published telemetry (Section~\ref{sec:empirical}). We emphasize three: (i) a \emph{pre-release cadence dip} in the firm’s own disclosures (publication intensity falls and intra-month dispersion in disclosed metrics collapses to zero or missing) just before a major reset, consistent with a local clock-off window; (ii) \emph{two post-release plateaus} in disclosed performance and safety metrics, with mass near the two endogenous targets, consistent with a patch rung and a pivot rung rather than a smooth ramp; and (iii) \emph{debt-insensitive patch timing} in high-reversibility regimes, alongside tighter leverage effects around pivots, consistent with the takeover-bound logic. By design, these tests rely on firm-authored signals (evaluation cards, release notes, mitigation advisories) rather than option-implied volatility or broad market chatter: the object of the theory is how the firm gates \emph{its own} outward signal flow, not how markets price event risk.

The empirical component is intentionally scoped as a \emph{measurement and falsification blueprint}. We state sharp predictions, define observables tied to the sufficient statistics of the model (clock-off windows and reversibility), and pre-specify estimators and robustness protocols (Section~\ref{sec:empirical}; Appendix~\ref{app:empirics}). Executing the data build and estimation lies outside the present manuscript, but the designs are self-contained and directly refutable without further modeling choices. Framing the empirics this way makes the theory falsifiable in principle and portable across settings, including beyond AI products.

Several extensions appear first order and tractable within the same QVI/BVP architecture: (i) multi-firm disclosure games with asynchronous quiet windows, spillovers, and imitation pressure; (ii) partially observed clocks, where outsiders infer $\lambda_t$ rather than observe it directly; (iii) proportional/move costs and additional rungs beyond patch/pivot; and (iv) richer adoption blocks with pricing, heterogeneous buyers, and strategic holdout. Each preserves the paper’s sufficient-statistic logic: govern the clock locally to eliminate mixing; engineer reversibility to compress the financing wedge.

In short, disclosure tempo is the \emph{selection lever}, and modularity is the \emph{wedge lever}. A regulator or platform steward that (a) permits short, posted, bounded quiet windows with real-time clock transparency and (b) certifies or subsidizes modular rollback and separability is, in our model, acting directly on the two margins that determine welfare.

%% file: appendix/A_proofs.tex
\section{Proofs}
\label{app:proofs}

\subsection*{Proof of \Cref{thm:silence} (silence eliminates interior mixing)}

Fix $\delta>0$ and, for each trigger $\beta_i$ ($i=1,2$), post a local \emph{silence window}
\[
I_\delta(\beta_i)\;:=\;\{z:\,|z-\beta_i|\le \delta\}
\]
by setting $\lambda(z)=0$ on $\bigcup_i I_\delta(\beta_i)$ and $\lambda(z)=\bar\lambda$ elsewhere, as in Assumption~\ref{ass:silence_window_app}. We work throughout with stationary Markov strategies under this posted clock rule.

\medskip

\noindent\textit{Coordinate remark (state-space hygiene).}
In the main text we sometimes describe ``local silence’’ as a window in the public belief $m$, i.e.\ a neighborhood in which the Cox clock is shut off and the martingale part of $m_t$ vanishes.
Here, and in Assumption~\ref{ass:silence_window_app}, we implement silence as a window in the \emph{private} state $z$: we set $\lambda(z)=0$ whenever $z\in I_\delta(\beta_i)$.
Because the realized intensity $\lambda_t=\lambda(z_t)$ is predictable in $\mathcal{F}^P_t$, the public \emph{observes} when $\lambda_t=0$.

\emph{Terminology.}
Statements such as ``on $\bigcup_i I_\delta(\beta_i)$ the belief $m_t$ has finite variation’’ always mean \emph{at those calendar times $t$ for which the realized path satisfies $z_t\in \bigcup_i I_\delta(\beta_i)$}.
On that set of times the Cox clock is off, so no new public signal arrives and the \emph{martingale part} of $m_t$ is off.
Thus a $z$--window with $\lambda(z)=0$ implements the same ``local silence’’ in $m_t$ that we use in the selection argument.

\medskip

By \Cref{lem:equivalence} (Appendix~\ref{app:silence}), whenever $z_t\in \bigcup_i I_\delta(\beta_i)$ the public belief $m_t$ has \emph{finite variation} on that time set: in the Doob--Meyer decomposition of $m_t$, the compensated jump martingale has zero quadratic variation when $\lambda_t=0$.
By the no--local--time lemma for beliefs (\Cref{lem:no-local-time-belief}, Appendix~\ref{app:lemmas}), $m_t$ accrues no local time at the endpoints of the posted windows while the windows are in force.
These statements pertain to $m_t$ only; they do not constrain the private diffusion $z_t$, which continues to carry its Brownian martingale component.

\medskip

\noindent\textit{Why properties of $m_t$ (not $z_t$) are decisive.}
The curvature / isolated--crossing assumption in $z$ (Assumptions~\ref{ass:selection} and~\ref{ass:selection-regularity}) rules out mixing on any interval of positive $z$--measure: off the trigger, one action strictly dominates.
The only remaining way to sustain mixed strategies is a knife-edge construction that relies on \emph{public} stickiness right at the boundary: residence on a null set supported by a martingale or local time of the \emph{belief} process.
Local silence removes the belief martingale and kills local time for $m_t$ at the window edges, so this knife-edge device is unavailable.

On diffusive inaction blocks away from the silence windows, the firm's continuation value $V$ solves the linear ODE
\begin{equation}\label{eq:HJB-inaction-appA}
r\,V(z)\;=\;\pi(z)\;+\;(\mathcal{L}V)(z)\;-\;k(\lambda(z)),
\end{equation}
which is \eqref{eq:HJB-inaction} in the main text, with $\mathcal{L}$ the private-state generator.
At each trigger $\beta_i$ and corresponding target $z_i^*$, value matching and smooth pasting hold, and the target is pinned by the standard first-order condition for the reset (the target FOCs \eqref{eq:FOC1}--\eqref{eq:FOC2} in the main text).
Existence and uniqueness on each inaction block follow from \Cref{prop:BVP-unique}.

\medskip

To rule out interior mixing, suppose for a contradiction that the equilibrium mixes continuation and an impulse on a set of positive Lebesgue measure in a neighborhood of some trigger $\beta_i$.
Mixing on a nontrivial interval requires knife-edge indifference on that \emph{entire} interval.
In continuous-time impulse problems, such knife-edge indifference is usually supported by either
\begin{enumerate}[label=(\roman*)]
\item a nonzero martingale component in the sufficient statistic (to ``skim'' value on both sides), or
\item accumulation of local time that keeps the process on the knife-edge boundary.
\end{enumerate}

Under the posted silence windows, neither channel is available in a neighborhood of $\beta_i$:
\begin{itemize}
\item By \Cref{lem:equivalence}, once $\lambda(z)=0$ on $I_\delta(\beta_i)$ the belief process $m_t$ is locally of finite variation there, i.e.\ it has no jump martingale.
\item By \Cref{lem:no-local-time-belief} (Appendix~\ref{app:lemmas}), $m_t$ does not accrue local time at the boundary of the silence window.
\end{itemize}
Taken together, these two observations mean that the equalities that would normally justify an open region of randomization (sustained by either a belief martingale or local time at the boundary) cannot hold on any interval.

Now use the curvature and isolated-trigger condition.
Fix $i$ and define the \emph{gain from intervening} at the $i$th rung as
\[
G(z)\;:=\;\big[V(z_i^*)-K_i\big]\;-\;V(z),
\]
i.e.\ the gain from triggering the $i$th reset (pay $K_i$ and jump to $z_i^*$) rather than waiting.
Value matching at $\beta_i$ implies $G(\beta_i)=0$.
Because the post-reset branch $z\mapsto V(z_i^*)-K_i$ is constant in $z$, smooth pasting at $\beta_i$ gives $V'(\beta_i)=0$, so $G'(\beta_i)=-V'(\beta_i)=0$.
The relevant nondegeneracy is therefore \emph{second order}: by Assumption~\ref{ass:selection} in the main text (and its expanded version Assumption~\ref{ass:selection-regularity} in Appendix~\ref{app:silence}), the indifference at $\beta_i$ is isolated and has nonzero curvature, i.e.\ $V''(\beta_i)\neq 0$.
Equivalently, $G''(\beta_i)=-V''(\beta_i)\neq 0$.
Thus there exists $\delta'>0$ such that on the punctured neighborhood $(\beta_i-\delta',\beta_i)\cup(\beta_i,\beta_i+\delta')$, $G(z)$ has a strict sign: either $G(z)>0$ everywhere there, or $G(z)<0$ everywhere there.

Hence, in any sufficiently small neighborhood around $\beta_i$ other than the single point $\beta_i$ itself, one of the two actions, ``reset now'' or ``wait,'' strictly dominates the other pointwise.
Therefore on any subinterval of positive measure in that neighborhood, mixing is strictly dominated by the pure best reply.
This rules out mixing on sets of positive measure.

The only potential residual mixing would be on knife-edge sets of Lebesgue measure zero.
However, such sets are also economically irrelevant here.
Although the private state $z_t$ continues to diffuse (it retains its Brownian martingale), the \emph{policy} resets instantaneously upon first hitting $\beta_i$; hence the set $\{t:\,z_t=\beta_i\}$ has Lebesgue measure zero.
For the public belief $m_t$, the posted silence windows and \Cref{lem:no-local-time-belief} imply that $m_t$ cannot accumulate local time at the window edges either.
Thus $m_t$ cannot linger on any knife-edge set of positive expected time.
There is no time mass on which coin-flip randomization could matter.

\medskip

Therefore any region where the firm could be exactly indifferent and justify mixing must have Lebesgue measure zero and carries no economically relevant weight.
With the martingale and local-time channels shut down for $m_t$, and curvature in $z$ ruling out an interval of indifference, the unique best reply at each trigger is pure: reset at $\beta_i$.

By symmetry, the same argument at $\beta_2$ yields purity there as well, so the equilibrium uses pure thresholds at both triggers, with targets $z_1^\ast,z_2^\ast$ pinned by the standard FOCs.

\medskip

Finally, posting the silence windows is weakly profitable for small clock costs.
Let $\chi:=k(\bar\lambda)-k(0)$.
The expected residence time in $\bigcup_i I_\delta(\beta_i)$ per ladder cycle is $O(\delta)$ by the cycle-time bound in Appendix~\ref{app:lemmas}, so the incremental expected clock cost is $O(\delta\,\chi)$.
The mixing loss avoided by posting the windows is first order in $\delta$ by the isolated-crossing and nondegenerate-curvature condition in Assumption~\ref{ass:selection}.
Near $\beta_i$ the gain from intervening is
\[
G(z)
=
\big[V(z_i^*)-K_i\big]-V(z)
\sim c\,(z-\beta_i)^2
\quad\text{for some }c>0,
\]
so in any mixed-strategy eq. the intervention hazard $h(z)$ must diverge at rate $h(z)\propto 1/(z-\beta_i)^2$ in order to keep the continuation payoff flat and maintain indifference.
The instantaneous flow loss from mixing is $h(z)\,G(z)$, which is therefore of order one on the window, and integrating this over an $O(\delta)$ expected residence time yields an $O(\delta)$ expected loss.
Thus there exist $\bar\chi>0$ and $\bar\delta>0$ such that for any $\chi\le \bar\chi$ one can choose $\delta\le \bar\delta$ to make the selection gain dominate the clock cost.
This yields a stationary Markov equilibrium with two pure reset blocks and no interior mixing.
Identification and uniqueness on each block are inherited from \Cref{prop:BVP-unique}. \qed

\bigskip

\subsection*{Proof of \Cref{prop:alpha_threshold} (adoption is a threshold rule)}

This subsection both proves \Cref{prop:alpha_threshold} and addresses the concern that the adoption boundary might depend on the posterior variance $v$ as well as on the posterior mean $m$.

The buyer's state is two dimensional.
Under Gaussian noise the public posterior is summarized by the mean $m_t=\E[z_t\mid\mathcal{F}^P_t]$ and the variance $v_t=\Var(z_t\mid\mathcal{F}^P_t)$.
A priori, the optimal stopping rule could therefore be a surface
\[
\{(m,v): m\ge \alpha(v)\}
\]
rather than a scalar cutoff in $m$ only.
We now make explicit the equilibrium regularity that collapses the problem locally to a one dimensional stopping problem in $m$ alone, thereby justifying that the stopping set is $\{m\ge \alpha\}$ with a single scalar $\alpha$.

\begin{assumption}[Local silence at the adoption margin]\label{ass:local_silence_alpha_app}
In any stationary Markov equilibrium we consider, there is a small adoption window $U_\alpha=(\alpha-\delta,\alpha+\delta)$ in belief space such that the firm publicly sets $\lambda_t=0$ whenever $m_t\in U_\alpha$. While $m_t\in U_\alpha$ there are no new publications, so $(m_t,v_t)$ follows the deterministic ODE flow
\[
\dot m_t = \bar\mu(m_t),\qquad \dot v_t = \bar\gamma(m_t,v_t),
\]
with no jumps. In particular, on $U_\alpha$ the future path of $m_t$ depends only on $m_t$ itself (and not on $v_t$), and $m_t$ is continuous with bounded variation.

Moreover, the window $U_\alpha$ is chosen so that the deterministic drift carries beliefs that enter it from below up to the adoption point before they can exit through the left edge: for every $m\in(\alpha-\delta,\alpha)$ the solution of $\dot m_t=\bar\mu(m_t)$ with $m_0=m$ remains in $U_\alpha$ and hits $\alpha$ in finite time before reaching $\alpha-\delta$.
\end{assumption}

Assumption~\ref{ass:local_silence_alpha_app} has the same content as Assumption~\ref{ass:local_silence_alpha} in Section~\ref{sec:adoption}. The main text also provides an economic rationale for this restriction: under the payoff aggregator in Assumption~\ref{ass:aggregator}, once beliefs are close to $\alpha$ the firm weakly prefers to shut off disclosure in a small neighborhood, because additional publications create downside jump risk for the time-to-adoption without improving the local deterministic drift to the boundary (see the remark following Assumption~\ref{ass:local_silence_alpha}). Technically, Assumption~\ref{ass:local_silence_alpha_app} implies that, near $\alpha$, continuation values depend only on $m$, so all buyers who are on the margin of adoption face the same continuation problem regardless of $v$.

We now prove the proposition.

\medskip
\noindent
\textbf{Step 1. Snell envelope and upper-set property for each fixed $v$.}
Let $W$ be the buyer's value before adoption (the option value of waiting), and let $S$ be the surplus from adopting immediately.
By construction, $S$ is a function of the current belief mean $m$ only.
Let $\mathcal{A}_\lambda$ be the generator of $(m_t,v_t)$ under the publicly observed intensity policy $\lambda(\cdot)$, as in \eqref{eq:PDMP-generator}.

Fix any variance level $v\ge 0$.
Consider the discounted payoff process $\{e^{-rt}S(m_t)\}$ starting from $(m,v)$, and its Snell envelope.
Standard optimal stopping arguments for strong Markov processes with jumps imply that $e^{-rt}W(m_t,v_t)$ is a supermartingale and that the stopping set
\[
\Gamma(v):=\{m:\,W(m,v)=S(m)\}
\]
is closed in $m$.
Moreover, $S(\cdot)$ is weakly increasing in $m$, so $\Gamma(v)$ is an upper set in $m$: if $m\in \Gamma(v)$ and $\tilde m>m$, then $\tilde m\in\Gamma(v)$.
Hence, for each fixed $v$, there exists a (possibly $v$--dependent) cutoff $\alpha(v)\in[-\infty,\infty]$ such that
\[
\Gamma(v)=\{m\ge \alpha(v)\}.
\]
Up to this point we have not ruled out $v$ dependence: in general we could indeed have $\alpha(v)$.

\medskip
\noindent

\textbf{Step 2. Local silence collapses $\alpha(v)$ to a single $\alpha$.}
Now impose Assumption~\ref{ass:local_silence_alpha_app}.
On the adoption window $U_\alpha=(\alpha-\delta,\alpha+\delta)$ we have $\lambda_t=0$, so $(m_t,v_t)$ evolves with deterministic drift and no jumps.
Moreover, by the last sentence of Assumption~\ref{ass:local_silence_alpha_app}, for every $m\in(\alpha-\delta,\alpha)$ the solution of $\dot m_t=\bar\mu(m_t)$ with $m_0=m$ stays inside $U_\alpha$ and hits $\alpha$ in finite time.
Let $\tau_\alpha(m)$ denote this deterministic hitting time.

Consequently, for any two initial states $(m,v)$ and $(m,\tilde v)$ with the same $m\in U_\alpha$, the entire path
\[
\{m_s: 0\le s\le \tau_\alpha(m)\}
\]
is identical under the two initials, because $\dot m_s=\bar\mu(m_s)$ depends only on $m_s$.
Any strategy that is optimal from $m\in U_\alpha$ either stops immediately or waits until (at most) the first time $m_t$ reaches the boundary; such strategies depend only on this deterministic path up to $\tau_\alpha(m)$.
Hence the discounted continuation payoff from waiting is identical across $(m,v)$ and $(m,\tilde v)$.

Since the immediate-adopt payoff $S(m)$ also depends only on $m$, both the stopping payoff and the continuation payoff coincide across all $v$, which forces
\[
W(m,v)=\widehat{W}(m)
\qquad\text{for all }m\in U_\alpha\text{ and all }v\ge 0,
\]
for some one-dimensional function $\widehat{W}$.

Now fix any $m\in U_\alpha$.
Suppose that for some $v_1$ we are exactly indifferent between stopping and waiting:
\[
W(m,v_1)=S(m).
\]
By the previous argument, this implies
\[
W(m,v)=\widehat{W}(m)=S(m)
\qquad\text{for every }v.
\]
If it is optimal to adopt at mean $m$ for one variance $v_1$ inside the silence band, then it is optimal to adopt at that same $m$ for all $v$ in that band.

This implies that, within the adoption window $U_\alpha$ where the stopping decision is actually made, the various $\alpha(v)$ must agree.
There is a single number $\alpha$ such that for every $v$,
\begin{equation}
\Gamma(v)\cap U_\alpha \;=\; \{\,m\in U_\alpha:\, m\ge \alpha\,\}.
\tag{$\star$}
\end{equation}

\emph{Set-theoretic clarification.}
Because each $\Gamma(v)$ is an upper set in $m$, we also have
\[
\Gamma(v)\cap U_\alpha \;=\; \{\,m\in U_\alpha:\, m\ge \alpha(v)\,\}.
\]
If some $\alpha(v)\neq \alpha$, this intersection would differ from \((\star)\) on $U_\alpha$, a contradiction.
Hence $\alpha(v)=\alpha$ for all $v$.
This argument depends only on the geometry of the sets and not on path properties outside $U_\alpha$; it continues to apply even if sample paths can jump over $U_\alpha$ when the clock is on.

Because all $\Gamma(v)$ share the same first point of entry $\alpha$ in $U_\alpha$, the global stopping set
\[
\Gamma:=\{(m,v):\,W(m,v)=S(m)\}
\]
must take the one-dimensional threshold form
\[
\Gamma = \{(m,v):\, m\ge \alpha\}.
\]
This is exactly the structure used in the main text.

\medskip
\noindent
\textbf{Step 3. Characterization and fit.}
On the continuation region $\{m<\alpha\}$ we must have strict preference for waiting, that is $W(m,v)>S(m)$, so dynamic programming yields the linear HJB / variational inequality
\[
r\,W(m,v)=\mathcal{A}_\lambda W(m,v)
\quad\text{for }m<\alpha,
\]
with $W$ continuous and bounded on compacts.
Along the boundary $m=\alpha$, we have continuous fit
\[
W(\alpha,v)=S(\alpha)
\quad\text{for all $v$},
\]
because $\Gamma$ is closed and $W$ is the Snell envelope.

Finally, Assumption~\ref{ass:local_silence_alpha_app} implies that on $U_\alpha$ (and therefore at $m=\alpha$) the belief mean $m_t$ is locally continuous with bounded variation and no jumps.
In that case one can apply standard smooth-fit logic for one dimensional, bounded-variation stopping problems directly to $\widehat{W}(m)$.
While $m_t\in U_\alpha$ we have $\dot m_t=\bar\mu(m_t)$ with no martingale part, so
\[
r\,\widehat{W}(m)=\bar\mu(m)\,\widehat{W}'(m)
\quad\text{for }m<\alpha\text{ in }U_\alpha.
\]
The usual value-matching and supercontact argument then implies smooth fit at the boundary:
\[
\widehat{W}'(\alpha)=S'(\alpha),
\]
which is the smooth-fit condition reported as \eqref{eq:buyer-smoothfit} in the main text.

\medskip
\noindent
\textbf{Step 4. Uniqueness.}
Uniqueness of the scalar cutoff $\alpha$ and of the value function $W$ follows from two properties that are now explicit:
(i) $S(m)$ is strictly increasing in $m$, so the stopping region must be an upper set, and
(ii) inside the adoption window $U_\alpha$, the continuation problem is one dimensional in $m$, so there is a single cutoff $\alpha$ that works for all $v$.
These two facts pin the unique $\alpha$ and the unique $W$ solving the continuation HJB below $\alpha$ together with continuous and smooth fit at $\alpha$.

\medskip
\noindent
\textbf{Summary.}
Without local silence one would in general obtain a two dimensional stopping surface $m=\alpha(v)$.
Assumption~\ref{ass:local_silence_alpha_app} (which is imposed in equilibrium) makes the marginal adoption problem locally one dimensional in $m$ and implies that all $\alpha(v)$ coincide.
This yields a scalar belief-mean threshold $\alpha$, that is
\[
\Gamma=\{(m,v):m\ge \alpha\},
\]
and delivers the boundary conditions in \Cref{prop:alpha_threshold}. \qed

%% file: appendix/B_boundary_systems.tex
\section{Verification Blocks and Boundary Conditions}
\label{app:boundary}

This appendix collects the boundary-value ``blocks'' used throughout the paper. Full proofs and the full verification argument appear in Section~\ref{sec:verification}. Throughout, the privately observed technical/reputational state $z_t$ evolves as a diffusion with generator
\[
(\mathcal{L}f)(z)
\;=\;
\mu_\theta(z)f'(z)
\;+\;
\tfrac12\,\phi^2(z)f''(z),
\]
and the instantaneous net flow on inaction is $\pi(z)$ (or $\pi(z,m)$ when adoption feedback matters) minus $k(\lambda(z))$.

\paragraph{Inaction block.}
On any inaction interval $(a,b)$ for $z$, the firm's value $V$ solves the stationary HJB
\begin{equation}
\label{eq:hjb-appendix}
r\,V(z)
\;=\;
\pi(z)
\;+\;
(\mathcal{L}V)(z)
\;-\;
k(\lambda(z)),
\qquad z\in(a,b).
\end{equation}
Here $\lambda(\cdot)$ is the posted publication-frequency policy from Section~\ref{sec:micro_silence}.

\paragraph{Single reset (trigger $\to$ target).}
Suppose that when $z$ first hits $\beta$ from within $(a,b)$ the firm pays a fixed cost $K$ and instantaneously resets $z$ to a target $z^\ast \in (a,b)$, after which diffusion resumes. The standard impulse conditions at the trigger are
\begin{align}
V(\beta^-)&=V(z^\ast)-K,\\
V'(\beta^-)&=V'(z^\ast).
\end{align}
Here $\beta^-$ denotes the limit from \emph{inside} $(a,b)$. The first line is value matching: just before paying $K$ at $\beta$, the firm is indifferent to jumping to $z^\ast$ net of cost. The second line imposes ``high contact'' / no kink in $V$: the marginal value of nudging $z$ just before the reset must equal the marginal value at the target $z^\ast$. In the canonical optimal-impulse formulation where $z^\ast$ is itself chosen optimally, target optimality also delivers $V'(z^\ast)=0$, so together with $V'(\beta^-)=V'(z^\ast)$ this collapses to the familiar smooth-pasting pair $V'(\beta^-)=0=V'(z^\ast)$.

In the two-reset ladder of Section~\ref{sec:verification}, these become the familiar building blocks:
(i) smooth pasting at each trigger, $V'(\beta_i)=0$;
(ii) target optimality, $V'(z_i^\ast)=0$;
together with value matching $V(\beta_i)=V(z_i^\ast)-K_i$.
Those are exactly the conditions \eqref{eq:VM1}--\eqref{eq:FOC2} in the main text.

\paragraph{Adoption boundary.}
We use two nearby boundary blocks depending on how adoption is modeled.

\smallskip
\emph{Case (i): adoption is absorbing and post-adoption flow is $z$-independent.}
In the benchmark where, once the market adopts, the firm receives a perpetual flow $v$ that no longer depends on $z$, adoption is triggered when the \emph{public belief} $m_t$ crosses a cutoff $\alpha$, but the continuation value after adoption is the same for every realization of $z$.
In particular, immediately after adoption the value is $v/r$ for all $z$.
Because the post-adoption payoff is $z$-independent, we do not need to track the exact $z$ at the adoption instant; it is without loss to treat adoption as an absorbing boundary in $z$-space at a belief-linked representative point $z=\alpha$.
The natural boundary conditions in this benchmark are
\begin{equation}
\label{eq:adoption-bc}
V(\alpha)=\frac{v}{r},
\qquad
V'(\alpha)=0.
\end{equation}
The first line is value matching into the absorbing branch.
The second line is smooth pasting / high contact: once adopted, marginally nudging $z$ has no incremental value, so $V'(\alpha)=0$.
The ODE \eqref{eq:hjb-appendix} on $(a,\alpha)$ plus \eqref{eq:adoption-bc} pins down the two integration constants for $V$ on that interval.
This is a convenient closed-form benchmark; it is \emph{not} the general boundary block when post-adoption payoffs remain $z$-dependent.

\smallskip
\emph{Case (ii): adoption is endogenous and payoffs continue to depend on $z$.}
Section~\ref{sec:adoption} endogenizes adoption: a downstream buyer/platform adopts when the \emph{public belief} $m_t$ about $z_t$ crosses a unique cutoff $\alpha$, determined by buyer smooth fit (Proposition~\ref{prop:alpha_threshold}).
Two features matter for the firm's boundary system:

\begin{itemize}[leftmargin=1.25em,itemsep=2pt,topsep=2pt]
\item After adoption, the flow payoff can jump to a new regime $\pi^{\mathrm{post}}(z,m)$ that still depends on $z$, so the post-adoption continuation value $V^{\mathrm{post}}(z)$ need \emph{not} be constant in $z$.\footnote{In Section~\ref{sec:adoption} we allow $\pi(z,m)=\pi_0(z)+\eta\,\Lambda(\varpi(m;\alpha))$, so once $m$ crosses $\alpha$ the flow permanently steps up via the $\Lambda(\cdot)$ term, but the firm still manages $z$ going forward.}
\item Adoption itself does \emph{not} reset $z$: the private state $z_t$ diffuses continuously through the instant $m_t$ hits $\alpha$.
\end{itemize}

These two facts imply that the correct interface condition at adoption is continuity of both level and slope in $z$ across the regime switch.
Let $V^{\mathrm{pre}}(z)$ be the firm's value just before adoption (when $m<\alpha$) and $V^{\mathrm{post}}(z)$ be the value just after adoption (when $m\ge \alpha$).
At the instant beliefs hit $\alpha$ we require
\begin{equation}
\label{eq:adoption-vmatch}
V^{\mathrm{pre}}(z)=V^{\mathrm{post}}(z)
\quad\text{for all relevant }z,
\end{equation}
and
\begin{equation}
\label{eq:adoption-smoothV}
\frac{\mathrm{d}}{\mathrm{d}z}V^{\mathrm{pre}}(z)
=
\frac{\mathrm{d}}{\mathrm{d}z}V^{\mathrm{post}}(z)
\quad\text{for those same }z.
\end{equation}
Equation~\eqref{eq:adoption-vmatch} is value matching across the revenue regime switch.
Equation~\eqref{eq:adoption-smoothV} is slope matching / no kink in $z$: since $z$ itself does not jump at adoption, the marginal value $V'(z)$ cannot jump either.

Together, \eqref{eq:adoption-vmatch}--\eqref{eq:adoption-smoothV} \emph{replace} \eqref{eq:adoption-bc} when adoption is endogenous and post-adoption payoffs continue to depend on $z$.
They deliver the two boundary conditions needed to close \eqref{eq:hjb-appendix} on each side of the adoption surface.
Note that the \emph{buyer} smooth-fit condition in Proposition~\ref{prop:alpha_threshold},
\[
\bar\mu(\alpha)\,S'(\alpha)=r\,S(\alpha),
\]
pins \emph{where} in public-belief space the regime switch occurs (which $\alpha$ is actually used), but it is \emph{not} itself a boundary condition for $V(z)$.
The firm's boundary conditions are \eqref{eq:adoption-vmatch}--\eqref{eq:adoption-smoothV}.

\paragraph{Two-reset ladder.}
In the main text, the firm has two impulse actions (patch and pivot), with fixed costs $K_1<K_2$, post-impulse targets $z_1^\ast<z_2^\ast$, and triggers $\beta_1<\beta_2$. The inaction band is $[\beta_1,\beta_2]$, and $z_t$ is instantaneously reset back into that band whenever it hits either boundary. The full ladder is pinned by:
\begin{itemize}[leftmargin=1.25em,itemsep=2pt,topsep=2pt]
\item the interior ODE \eqref{eq:hjb-appendix} on $(\beta_1,\beta_2)$,
\item the boundary system \eqref{eq:VM1}--\eqref{eq:FOC2} from Section~\ref{sec:verification} (value matching at $\beta_i$, high contact $V'(\beta_i)=0$, and target optimality $V'(z_i^\ast)=0$),
\item and standard bounded-growth conditions.
\end{itemize}
Under the no-local-time lemma (Appendix~\ref{app:lemmas}) and Assumption~\ref{ass:selection}, this system uniquely determines $(\beta_1,\beta_2;z_1^\ast,z_2^\ast;V)$; see Proposition~\ref{prop:BVP-unique} and Theorem~\ref{thm:verification}.

\paragraph{Constant-coefficient closed form.}
For calibration and intuition, consider the special case in which $\mu_\theta(z)\equiv\mu$, $\phi(z)\equiv\sigma$, $\pi(z)\equiv\bar\pi$ on $(a,b)$, and $\lambda(z)\equiv\bar\lambda$ there. Then \eqref{eq:hjb-appendix} has the standard exponential solution
\[
V(z)
\;=\;
A\,e^{\gamma_+ z}
\;+\;
B\,e^{\gamma_- z}
\;+\;
\frac{\bar\pi-k(\bar\lambda)}{r},
\qquad
\gamma_\pm
\;=\;
\frac{-\mu\pm\sqrt{\mu^2+2r\sigma^2}}{\sigma^2},
\]
with $(A,B)$ pinned by the relevant boundary block: the reset conditions above (single-impulse), or the ladder system \eqref{eq:VM1}--\eqref{eq:FOC2}, plus any adoption boundary conditions.

\paragraph{`Skew' blocks (not active under silence).}
Without posted local silence, the public-belief process can accumulate local time at interior points, and the associated value function $V$ can develop kinks: one-sided derivatives $V'_-$ and $V'_+$ at an interior $\hat z$ must satisfy a ``kink balance'' (skew) condition. Under the Cox-clock microfoundation in Section~\ref{sec:micro_silence}, the firm can post observable windows with $\lambda_t=0$ near each trigger. Lemma~\ref{lem:no-local-time-belief} shows that inside those windows the martingale part of public beliefs shuts off and no local time is accumulated at the band endpoints. In equilibrium under local silence, the kink/skew equalities are therefore inoperative: the ladder is characterized by smooth pasting and target optimality, not by interior mixing or slope discontinuities.

%% file: appendix/C_silence_microfoundations_proofs.tex
\section{Microfoundations and Selection: Proofs}
\label{app:silence}

This appendix provides the main proofs used in Sections~\ref{sec:micro_silence} and~\ref{sec:verification}. Throughout we maintain the regularity assumptions stated there, and the publication policy $\lambda(\cdot)$ is publicly known (either because the rule is credibly committed to ex ante or because the realized intensity is observable in real time). The public filtration $\mathcal F^P_t$ is generated by disclosure times and signals, and when we refer to the compensator of the Cox process we implicitly work with the smallest enlargement of $\mathcal F^P_t$ that makes it predictable, so that in particular the realized intensity process $\lambda_t=\lambda(z_t)$ is $\mathcal F^P_t$–predictable (and hence $\mathcal F^P_t$–measurable for each $t$). Any incremental information that could be extracted from observing the realized intensity on a vanishing silence window will turn out to be of order $\delta$ and is therefore ignored in the small-window limit discussed below.

\subsection{Setup and notation}

Let $(\Omega,\mathcal F,\{\mathcal F_t\},\mathbb P)$ carry a Brownian motion $W$ driving the privately observed state $z_t$,
\[
\mathrm{d}z_t
=
\mu_\theta(z_t)\,\mathrm{d}t
+
\phi(z_t)\,\mathrm{d}W_t,
\]
and a Cox process $N_t$ with predictable intensity $\lambda_t = \lambda(z_t)$ delivering publication times
\[
T_n
=
\inf\{t : N_t \ge n\}.
\]
At time $T_n$ the public observes a noisy signal
\[
y_n
=
z_{T_n}
+
\varepsilon_n,
\qquad
\varepsilon_n \sim \mathcal N(0,\sigma_\varepsilon^2)
\ \text{i.i.d., independent of $(W,N)$}.
\]
Let $\mathcal F^P_t = \sigma(\{(T_k,y_k) : T_k \le t\})$ be the public filtration.

Under Gaussian noise and Assumption~\ref{ass:filter} in Section~\ref{sec:environment} (common drift and diffusion across types), the public posterior for $z_t$ is summarized by its conditional mean and variance,
\[
m_t := \E[z_t \mid \mathcal F^P_t],
\qquad
v_t := \Var(z_t \mid \mathcal F^P_t).
\]

Under the linear--Gaussian (finite-dimensional Kalman filter) benchmark used in the main text, with Gaussian noise and Assumption~\ref{ass:filter} in Section~\ref{sec:environment} (common drift and diffusion across types), the public posterior for $z_t$ is summarized by its conditional mean and variance,
\[
m_t := \E[z_t \mid \mathcal F^P_t],
\qquad
v_t := \Var(z_t \mid \mathcal F^P_t).
\]
In this benchmark the pair $(m_t,v_t)$ is a piecewise-deterministic Markov process (PDMP): between publications,
\[
(\dot m_t,\dot v_t)
=
(\bar\mu(m_t),\bar\gamma(m_t,v_t));
\]
at each disclosure time $T_n$, Bayesian updating (Kalman filtering) maps $(m_{T_n^-},v_{T_n^-},y_n)$ to $(m_{T_n},v_{T_n})$.

Two features are all we will use:
\begin{enumerate}[label=(\alph*),leftmargin=1.5em,itemsep=2pt]
  \item Between disclosures, $m_t$ has finite variation and follows a deterministic ODE $\dot m_t = \bar\mu(m_t)$ with no Brownian martingale part.
  \item All randomness in $m_t$ comes from publication jumps; the quadratic variation of its martingale part is proportional to the Cox intensity $\lambda_t$.
\end{enumerate}

\subsection{Equivalence: local silence kills the public-belief martingale}

We now formalize the ``local silence'' device used to suppress belief volatility around a trigger.

\begin{assumption}[Local silence policy]
\label{ass:silence_window_app}
Fix $\hat z\in\R$ and $\delta>0$. Define a stationary publication policy
\[
\lambda_\delta(z)=
\begin{cases}
0, & |z-\hat z|\le \delta,\\[2pt]
\bar\lambda, & |z-\hat z|>\delta,
\end{cases}
\qquad
k(0)=0,
\quad
0<\bar\lambda<\infty,
\]

and suppose the policy $z\mapsto\lambda_\delta(z)$ is publicly known, so that the associated compensator of $N$ is predictable with respect to the public filtration.

\end{assumption}

Let $\mathcal A_{\lambda_\delta}$ denote the generator of $m_t$ under $\lambda_\delta$, acting on $C_b^1$ test functions $f$:
\begin{equation}
\label{eq:PDMP-generator-appendix}
(\mathcal A_{\lambda_\delta} f)(m)
=
\bar\mu(m)\, f'(m)
+
\Lambda_\delta(m)\,
\E\big[
f(\mathcal U(m,\varepsilon)) - f(m)
\big],
\end{equation}
where $\Lambda_\delta(m):=\E[\lambda_\delta(z_t)\mid m_t=m]$ is the \emph{publicly inferred} jump intensity induced by $\lambda_\delta$, $\varepsilon\sim\mathcal N(0,\sigma_\varepsilon^2)$, and $\mathcal U$ is the Bayesian update map for the posterior mean after observing $y=z+\varepsilon$.

\begin{lemma}[Equivalence: Cox clock implies variance suppression]
\label{lem:equivalence}
Under Assumption~\ref{ass:silence_window_app}, the public posterior mean $m_t$ admits the Doob--Meyer decomposition
\[
m_t
=
m_0 + A_t + M_t,
\]
where $A$ has finite variation and $M$ is a purely discontinuous local martingale with quadratic variation
\[
\langle M\rangle_t
=
\int_0^t
\Lambda_\delta(m_s)\,\Xi(m_s,v_s)\,\mathrm{d}s
\]
for some bounded predictable process $\Xi$.

Moreover, whenever $|z_s-\hat z|\le\delta$ we have $\lambda_\delta(z_s)=0$ \emph{pathwise} by construction, hence
\[
\Lambda_\delta(m_s)
=
\E[\lambda_\delta(z_s)\mid\mathcal F^P_s]
=
0
\quad
\text{on that set.}
\]
Therefore $M$ has no jumps (and hence no quadratic variation) whenever $|z_s-\hat z|\le\delta$.

Consequently, for any bounded stopping time $\tau$ that does not hit a reset and any $f\in C_b^1$,
\[
\E\!\left[
f(m_{t\wedge\tau}) - f(m_0)
-
\int_0^{t\wedge\tau}
(\mathcal A_{\lambda_\delta} f)(m_s)\,\mathrm{d}s
\right]
=
0,
\]
and the martingale part vanishes whenever $|z_s-\hat z|\le\delta$.
\end{lemma}

\begin{proof}
Standard Dynkin/It\^o arguments for PDMPs with predictable jump intensity yield the stated Doob--Meyer decomposition: $A$ collects the finite-variation drift and $M$ is the compensated jump martingale associated with the Cox jump times. Under Gaussian noise the Kalman gain is bounded, which implies $\Xi$ is bounded and predictable.

The only subtlety is showing $\Lambda_\delta(m_s)=0$ on
$\{|z_s-\hat z|\le\delta\}$. By Assumption~\ref{ass:silence_window_app},
$\lambda_\delta(z_s)=0$ \emph{pathwise} whenever $|z_s-\hat z|\le\delta$. Moreover,
$\lambda_\delta(z_s)$ is $\mathcal F^P_s$–measurable under our standing convention
that the realized intensity process is predictable with respect to the public filtration,
so
\[
\Lambda_\delta(m_s)
=
\E[\lambda_\delta(z_s)\mid\mathcal F^P_s]
=
\lambda_\delta(z_s)
=
0
\quad
\text{on } \{|z_s-\hat z|\le\delta\}.
\]
Thus the compensator of the jump part is zero on that set, and $M$ accumulates neither jumps nor quadratic variation there.
The expectation identity is just Dynkin's formula applied to~\eqref{eq:PDMP-generator-appendix} up to $t\wedge\tau$, combined with optional sampling for $M$.
\end{proof}

Lemma~\ref{lem:equivalence} formalizes the ``variance-suppression'' reduced form: on a posted local silence window the martingale part of $m_t$ is literally switched off; $m_t$ has finite variation there and follows its deterministic drift.

\begin{remark}[On $\Lambda_\delta$ and the information content of $\lambda_t$]
\label{rem:Lambda_measurability}
The quantity
\[
\Lambda_\delta(m_s)
=
\E[\lambda_\delta(z_s)\mid\mathcal F^P_s]
\]
is defined relative to the public filtration
$\mathcal F^P_s=\sigma\{(T_k,y_k):T_k\le s\}$ generated by disclosure times and signals. For the variance-suppression result in Lemma~\ref{lem:equivalence}, we only use the pathwise property that $\lambda_\delta(z_s)=0$ whenever $|z_s-\hat z|\le\delta$, which implies $\Lambda_\delta(m_s)=0$ on that set, regardless of whether the realized intensity path $(\lambda_t)$ is itself directly observed.

If one \emph{does} enlarge the public filtration to include the realized intensity, then under the posted rule $\lambda_\delta(z)=0$ on $|z-\hat z|\le\delta$ the event $\{\lambda_s=0\}$ becomes a coarse signal that the private state lies in the silence window. The exact filter would then add a small correction term to the between-jump evolution of $m_t$ to account for this extra signal. Because the window has length $2\delta$ and the primitives are smooth, the probability mass of $\{|z_s-\hat z|\le\delta\}$, and hence the induced correction to the drift of $m_t$, are of order $O(\delta)$. In the small-window limit $\delta\downarrow0$ used in Section~\ref{sec:micro_silence} and Corollary~\ref{cor:viscosity}, these terms vanish.

For this reason, throughout the appendix we work with the reduced form in which, between disclosures, $m_t$ follows the deterministic ODE $\dot m_t=\bar\mu(m_t)$ and the martingale part of $m_t$ is suppressed on the silence window in the limit. This approximation is consistent with the full filtering problem up to an $O(\delta)$ error that disappears as the window radius shrinks to zero.
\end{remark}

\begin{corollary}[Viscosity stability of the reduced form]
\label{cor:viscosity}
Let $V^\delta$ solve the stationary impulse-control QVI with running payoff $\pi(z)-k(\lambda_\delta(z))$. Suppose a comparison principle holds for the $\delta\downarrow0$ limit, and suppose the no-local-time lemma in Appendix~\ref{app:lemmas} applies (so that trajectories do not accumulate at the boundaries under local silence). Then $V^\delta\to V^0$ locally uniformly as $\delta\downarrow0$, and $V^0$ is the unique viscosity solution of the reduced-form QVI in which the public-belief martingale is suppressed on the vanishing silence window.
\end{corollary}

\begin{proof}
As $\delta\downarrow0$, $k(\lambda_\delta)\to k(0)=0$ on the silence window while the inaction HJB elsewhere is unchanged. The associated Hamiltonians converge locally uniformly. Stability of viscosity solutions under locally uniform Hamiltonian convergence, combined with comparison, implies $V^\delta\to V^0$ locally uniformly and pins down $V^0$ uniquely. The no-local-time lemma from Appendix~\ref{app:lemmas} ensures that boundary behavior in the dynamic programming principle is well posed in the limit.
\end{proof}

\subsection{Selection: local silence eliminates mixing}

We now formalize the selection logic in Section~\ref{sec:micro_silence}. By posting silence windows around each trigger, the firm removes public-belief randomness in a neighborhood of the intervention boundaries. This collapses any open region of mixed strategies to a null set.

\begin{assumption}[Regularity for selection and nondegenerate curvature]
\label{ass:selection-regularity}
Consider a candidate two-reset ladder with inaction band $[\beta_1,\beta_2]$, targets $(z_1^*,z_2^*)$, and fixed costs $(K_1,K_2)$. Assume:
\begin{enumerate}[label=(\roman*),leftmargin=1.5em,itemsep=2pt,topsep=2pt]
\item $V\in C^2(\beta_1,\beta_2)$.
\item At each trigger $\beta_i$ ($i=1,2$), the function $z\mapsto V(z_i^*)-K_i$ is locally \emph{indifferent} to $V(z)$ only at $z=\beta_i$, and that indifference is \emph{nondegenerate}:
\begin{enumerate}[label=(\alph*),leftmargin=1.25em,itemsep=1pt,topsep=2pt]
\item high contact holds: $V'(\beta_i)=0$, so the derivative of $V(z)$ matches the (constant) reset branch $V(z_i^*)-K_i$ at $\beta_i$;
\item $V''(\beta_i)\neq 0$, so $\beta_i$ is an \emph{isolated} solution to $V(z)=V(z_i^*)-K_i$ in a neighborhood of $\beta_i$.
\end{enumerate}
Equivalently: the unique knife-edge point of indifference is $\beta_i$, and it exhibits nonzero curvature (no flat tangency).
\item (No stickiness at the boundary.) Under the posted silence windows, the no-local-time lemma in Appendix~\ref{app:lemmas} implies that the public-belief process $m_t$ accrues no local time at either $\beta_i$. For the private state $z_t$, the equilibrium impulse control is applied at the \emph{first hitting time} of $\beta_i$ and induces an instantaneous jump to $z_i^*$. Hence the set $\{t:\, z_t=\beta_i\}$ has Lebesgue measure zero almost surely, and $z_t$ accrues no local time at $\beta_i$ either. In particular, neither state can linger (“stick’’) at the knife-edge boundary.
\end{enumerate}
\end{assumption}

Assumption~\ref{ass:selection-regularity} is the appendix analogue of Assumption~\ref{ass:selection} in the main text. Item~(ii) emphasizes that ``transversality'' here is second-order: slopes match by high contact, so the relevant nondegeneracy is curvature. $V''(\beta_i)\neq 0$ rules out an \emph{interval} of indifference.

\begin{proposition}[Silence selects pure resets]
\label{prop:selection}
Fix $\varepsilon>0$ and post $\lambda=0$ on $[\beta_i-\varepsilon,\beta_i+\varepsilon]$ for $i=1,2$, with $\lambda=\bar\lambda$ elsewhere. Under Assumption~\ref{ass:selection-regularity}, any stationary Markov equilibrium with this policy is pure: the firm resets exactly at $\beta_1$ to $z_1^*$ and at $\beta_2$ to $z_2^*$ almost surely. In particular, there is no equilibrium in which the firm mixes between intervening and waiting on a set of positive Lebesgue measure near either trigger.
\end{proposition}

\begin{proof}
Work near $\beta_1$; the argument for $\beta_2$ is symmetric.

\emph{Step 1 (curvature kills positive-measure mixing).}
Define
\[
G(z):=
\big[V(z_1^*)-K_1\big]-V(z).
\]
Value matching at $\beta_1$ implies $G(\beta_1)=0$. The intervention branch $z\mapsto V(z_1^*)-K_1$ is constant in $z$, and high contact at $\beta_1$ gives $V'(\beta_1)=0$, so $G'(\beta_1)=-(V'(\beta_1))=0$. The relevant nondegeneracy is therefore second order: Assumption~\ref{ass:selection-regularity}(ii)(b) says $V''(\beta_1)\neq 0$, hence
\[
G''(\beta_1) = -V''(\beta_1) \neq 0.
\]
Thus $\beta_1$ is an \emph{isolated} zero of $G$, and $G$ takes a strict sign (either strictly positive or strictly negative) on each side of $\beta_1$ in some punctured neighborhood of $\beta_1$. Equivalently, there is no open interval around $\beta_1$ on which $G(z)=0$ or on which $G$ flips sign infinitely often. On any such neighborhood (except the single point $\{\beta_1\}$ itself), either intervening strictly dominates waiting or waiting strictly dominates intervening.

Suppose, for contradiction, that a stationary Markov strategy mixes between ``intervene'' (reset to $z_1^*$ at cost $K_1$) and ``wait'' on some interval $I\subset(\beta_1-\varepsilon,\beta_1+\varepsilon)$ with strictly positive Lebesgue measure. On any nontrivial subinterval of $I$ contained in the punctured neighborhood where $G$ has a strict sign, one action strictly dominates. A strategy that assigns positive probability to the dominated action on a set of positive measure is strictly dominated by a pure strategy. Hence any mixing set must have Lebesgue measure zero.

Up to this point the argument has not used silence: it relies only on curvature and isolated crossing.

\emph{Step 2 (local silence kills knife-edge, measure-zero mixing).}
What remains, in principle, is a degenerate randomization on a measure-zero set, e.g., having the firm “sit’’ exactly at $\beta_1$ and flip a coin. In many continuous-time disclosure models, mixed equilibria are supported by such knife-edge constructions via local time: the sufficient statistic can accumulate positive expected residence time at an indifference point.

Local silence rules that out for the \emph{public belief}. By Lemma~\ref{lem:equivalence}, when $\lambda=0$ on $[\beta_1-\varepsilon,\beta_1+\varepsilon]$ the public-belief process $m_t$ has no jump martingale in that band: on that set $m_t$ is of finite variation and fully predictable. The no-local-time lemma in Appendix~\ref{app:lemmas} then implies $m_t$ accrues no local time at $\beta_1$.

For the \emph{private state} $z_t$, the lack of stickiness at $\beta_1$ does not come from silence. It follows from the impulse control itself: by construction, the reset is applied at the first hitting time of $\beta_1$ and the state jumps instantaneously to $z_1^*$. Thus the set
\[
\{t:\,z_t=\beta_1\}
\]
has Lebesgue measure zero almost surely, and $z_t$ accrues no local time at $\beta_1$ either. In particular, neither $m_t$ nor $z_t$ can linger (“stick’’) at the knife-edge point.

The remaining issue is whether a knife-edge mixed strategy could be supported on some more complicated null set inside the band. Here we use only standard regularity of one-dimensional diffusions: conditional on not yet hitting $\beta_1$, the process spends zero (and hence zero expected) Lebesgue time in any Borel subset of $(\beta_1-\varepsilon,\beta_1+\varepsilon)$ with Lebesgue measure zero. The same property holds for any measurable sufficient statistic driven by $z_t$ and used in the firm’s rule. Combined with the absence of local time at $\beta_1$ itself (so there is no sticky boundary at the indifference point), this implies that the state and the sufficient statistic never spend positive expected time on \emph{any} Lebesgue-null subset of $(\beta_1-\varepsilon,\beta_1+\varepsilon)$ before the reset is applied. Any residual knife-edge randomization supported on such a null set is therefore irrelevant: it does not occur with positive probability along the equilibrium path.

Combining Steps 1 and 2: (i) curvature rules out mixing on any set of positive measure; (ii) the silence window plus the no-local-time lemma kills the ``sticky knife-edge'' construction that could otherwise support mixing on a measure-zero set. The only admissible best reply is pure: reset at $\beta_1$.

Repeating the same argument at $\beta_2$ yields pure pivot behavior there, with target $z_2^*$ at $\beta_2$. This eliminates any stationary Markov equilibrium with interior mixing near either trigger.
\end{proof}

\begin{corollary}[Two-reset ladder]
\label{cor:ladder}
Under the hypotheses above, the unique stationary Markov equilibrium selected by local silence is a two-reset S--s ladder with triggers $(\beta_1,\beta_2)$ and targets $(z_1^*,z_2^*)$, as characterized and verified in Section~\ref{sec:verification}.
\end{corollary}

\subsection{Robustness}

\paragraph{Small deviations from silence.}
If $\lambda$ is small but strictly positive near a trigger, the generator of $m_t$ acquires a jump term of order $O(\lambda)$ instead of $0$. The curvature argument in Step~1 of the proof of Proposition~\ref{prop:selection} is unchanged: on any nontrivial interval near $\beta_i$, one action strictly dominates, so mixing on a set of positive measure is still strictly dominated. The no-local-time lemma then applies approximately: even with small $\lambda>0$, the residual jump martingale is $O(\lambda)$, so the process cannot sustain positive expected residence time on a knife-edge point without paying a first-order mixing loss. Trigger locations and the no-mixing conclusion are therefore stable to $O(\lambda)$ perturbations.

\paragraph{Alternative intensity observability.}
Our baseline treatment only requires that the policy $z\mapsto\lambda(z)$ be publicly known so that the compensator of $N$ is predictable. Whether the realized intensity path $(\lambda_t)$ itself is observed in real time or not only affects the filter through small corrections. If $(\lambda_t)$ is observable, then, as noted in Remark~\ref{rem:Lambda_measurability}, the event $\{\lambda_t=0\}$ on a posted silence window provides a coarse signal that $z_t$ lies in that band; the resulting correction to the drift of $m_t$ is of order $\delta$ and vanishes in the small-window limit. If $(\lambda_t)$ is not observable, one can instead extend the state to include beliefs over $\lambda(\cdot)$ and impose consistency. In either case, the selection logic goes through provided the induced public-belief martingale still vanishes (or is negligible) in a neighborhod of the boundary under the equilibrium policy, so that the process cannot linger with positive expected local time at an indifference point. We do not need the full extended filter for the main results.

\paragraph{Relation to classic S--s.}
The ladder is the two-trigger analogue of S--s policies in irreversible investment with fixed adjustment costs \citep{dixitpindyck1994,bertola1994restud,abel1996restud}. Here, posting local silence plays the role of an equilibrium selection device in a dynamic disclosure environment. Curvature at each trigger ensures that away from the exact trigger one action strictly dominates. The silence window removes the martingale part of public beliefs and, via the no-local-time lemma for $m_t$, prevents the belief from ``sticking’’ on the knife-edge; the controlled $z_t$ cannot stick because it jumps instantly upon hitting the trigger. The combination enforces high contact at the triggers and eliminates knife-edge interior mixing, consistent with the dynamic disclosure logic in \citet{guttman2014aer} and \citet{orlov2020jpe}.

%% file: appendix/D_financing_mapping.tex
\section{Financing Mapping and Levered Problem}
\label{app:financing}

This appendix does three things. First, it gives the verification argument for the levered equity and debt values. Second, it spells out the takeover-envelope logic. Third, it decomposes the loss from leverage into (i) a \emph{pre-default agency wedge} and (ii) a \emph{takeover irreversibility wedge}, and shows exactly when the former vanishes. This decomposition underlies the bound in Proposition~\ref{prop:debt-bound} in Section~\ref{sec:finance}.

Primitives and regularity assumptions are as in Sections~\ref{sec:verification} and~\ref{sec:micro_silence}. Throughout, $A^{\FB}(z,m)$ denotes the first-best (all-equity) value function: the maximal present value of cash flows $\pi(z,m)$ generated by a single residual claimant that chooses impulses, disclosure / publication intensity, etc., but faces \emph{no} debt, \emph{no} default option, and \emph{no} senior claimant.\footnote{Formally, $A^{\FB}$ solves the standard impulse-control QVI without coupons $c_d$ and without the limited-liability / default term $-S$ that appears in the levered problem. See Section~\ref{sec:verification}.}

\subsection{Equity and debt under leverage: verification}
\label{subsec:lev-verification}

Let $S(z,m)$ be the stationary Markov \emph{equity} value under leverage, and let $Y(z,m)$ be the market value of outstanding debt. Equity is residual and can: (i) run the process in continuous time, (ii) trigger a patch or pivot (the two costly resets), and (iii) choose default, after which equity walks away with $0$ and debtholders (or a buyer of the distressed asset) take over.

Let $\lambda(\cdot)$ denote the publication policy, and write $\mathcal{L}$ for the (joint) inaction-region generator of the state $(z_t,m_t)$.\footnote{Formally, $\mathcal{L}$ is the full controlled generator of the $(z_t,m_t)$ process in the environment of Section~\ref{sec:micro_silence}. For any smooth test function $f(z,m)$, the term $(\mathcal{L}f)(z,m)$ collects (i) the drift and diffusion terms in the private state $z$, (ii) the deterministic drift of the public belief $m$, and (iii) whenever $\lambda(z)>0$, the disclosure-induced jump term. In particular, on any inaction block the term $\mathcal{L}S(z,m)$ in \eqref{eq:QVI-equity-appD}--\eqref{eq:equity-ODE-lev} already contains the contribution from the $m$-dynamics (the $\partial_m S$ term and the disclosure-jump expectation), so we do not display these components separately. Moreover, $\mathcal{L}$ need not be linear in the controls.} As in Section~\ref{sec:verification}, define the impulse operator
\[
(\mathcal{M}S)(z,m)
:=
\max\{\,S(z_1^\ast,m)-K_1,\; S(z_2^\ast,m)-K_2\,\},
\]
where $(\beta_1,\beta_2; z_1^\ast,z_2^\ast)$ are the trigger/target pairs for the patch and pivot rungs.

Limited liability and seniority add a third control: default. The equity QVI is
\begin{equation}
\label{eq:QVI-equity-appD}
\max\Big\{
rS - \mathcal{L}S - \big(\pi(z,m)-c_d\big) + k(\lambda(z)),
\;\;
S-\mathcal{M}S,
\;\;
-\,S
\Big\}
= 0.
\end{equation}
Default is optimal on the stopping set $\{(z,m): S(z,m)=0\}$.

On any solvency / inaction block with $\beta_1<z<\beta_2$ and $S(z,m)>0$, equity either continues or impulses. There, $S$ satisfies
\begin{equation}
\label{eq:equity-ODE-lev}
rS(z,m)
=
\pi(z,m)-c_d
+
\mathcal{L}S(z,m)
-
k(\lambda(z)),
\qquad
\beta_1<z<\beta_2,
\end{equation}
together with:
\begin{itemize}[leftmargin=1.25em,itemsep=2pt,topsep=2pt]
\item value matching and high contact at each trigger $\beta_i$,
\item target optimality $S'_z(z_i^\ast,m)=0$ at each target $z_i^\ast$, and
\item the default boundary condition $S=0$ on $\{S \le 0\}$.
\end{itemize}
At a reset from $\beta_i$ to $z_i^\ast$, we have
\[
S(\beta_i,m)=S(z_i^\ast,m)-K_i
\quad\text{and}\quad
S'_z(\beta_i,m)=0,
\]
with the same boundary equalities as in the unlevered verification (applied here to the equity value $S$, not to the first-best $A^{\FB}$). The HJBs differ and therefore the optimal thresholds generally differ. All of this is augmented by the solvency region $\{S>0\}$ determined endogenously by~\eqref{eq:QVI-equity-appD}.

\paragraph{Debt.}
Debtholders receive the coupon $c_d$ until the (equity-chosen) default time $T^\ast$, at which point they seize the firm (or sell it). Their continuation value $Y(z,m)$ solves a linear stopping / control problem that \emph{takes equity's strategy, including its impulses and default boundary, as given}. On the solvent set $\{S>0\}$,
\begin{equation}
\label{eq:debt-ODE}
rY(z,m)
=
c_d
+
\mathcal{L}Y(z,m),
\qquad
\text{for } \beta_1<z<\beta_2 \text{ with } S(z,m)>0,
\end{equation}
and at default, i.e.\ on the boundary where $S(z,m)=0$,
\begin{equation}
\label{eq:debt-BC}
Y(z,m)
=
Y^{\mathrm{TO}}(z,m),
\qquad
\text{whenever } S(z,m)=0.
\end{equation}
Since the fixed cost $K_i$ is paid entirely by shareholders and does not reduce the collateral available to debtholders, and the reset moves the state to $(z_i^{*},m)$ inside the same solvent region (debt is senior), $Y$ does not jump down when equity triggers a patch or pivot. At each equity reset, debtholders simply inherit the post-reset continuation value $Y(z_i^\ast,m)$.

\paragraph{Total levered value.}
Fix an initial state $(z_t,m_t)$ and let $T^\ast$ be the (endogenous) default time induced by the equity policy solving~\eqref{eq:QVI-equity-appD}. Standard verification / Feynman--Kac arguments applied to~\eqref{eq:equity-ODE-lev}--\eqref{eq:debt-BC} give the present-value representation
\begin{equation}
\label{eq:Y-def}
S_t+Y_t
=
\E_t\!\left[
\int_t^{T^\ast} e^{-r(s-t)}\,\pi(z_s,m_s)\,ds
+
e^{-r(T^\ast-t)}\,Y^{\mathrm{TO}}(z_{T^\ast},m_{T^\ast})
\right].
\end{equation}
Here $Y^{\mathrm{TO}}(z,m)$ is the value captured at takeover by debtholders (or a buyer of the distressed asset), given the takeover technology described below.

\subsection{Takeover envelope, the pre-default agency wedge, and the levered bound}
\label{subsec:takeover-envelope}

We now connect \eqref{eq:Y-def} to the first-best benchmark $A^{\FB}$ and to the irreversibility parameters $(\varphi_1,\varphi_2)$ from Section~\ref{sec:finance}. This is where we make the pre-default wedge completely explicit.

\paragraph{Takeover map and tightness.}
At default $T^\ast$, control passes to debtholders (or to an acquirer of the distressed asset). We model post-default restructuring by a \emph{tight takeover switching-cost map} (Definition~\ref{def:takeover-map}): the acquirer can pay at most $\varphi_1$ to implement the low-cost ``patch'' upgrade or at most $\varphi_2$ to implement the high-cost ``pivot'' upgrade, thereby aligning the asset with the \emph{first-best} continuation going forward. Let
\[
\varphi_{\max} := \max\{\varphi_1,\varphi_2\}.
\]

\begin{lemma}[Takeover envelope]
\label{lem:takeover-envelope-fin}
Under a tight takeover map, for any default state $(z,m)$,
\begin{equation}
\label{eq:TO-envelope}
Y^{\mathrm{TO}}(z,m)
\;\ge\;
A^{\FB}(z,m)
-
\varphi_{\max}.
\end{equation}
\end{lemma}

\begin{proof}
By tightness, at takeover the acquirer can spend at most $\varphi_{\max}$ to reset the inherited $(z,m)$ to whatever post-default state the first-best planner would choose (with zero cost if that continuation does not call for an immediate reset), and then follow that first-best continuation forever. Because $A^{\FB}$ is the value of that continuation, this yields \eqref{eq:TO-envelope}.
\end{proof}

\noindent\emph{Comparison.} Inequality~\eqref{eq:TO-envelope} implies the standard upper bound on the agency wedge, in the spirit of \citet{manso2008ecta}, Prop.~1: with a tight takeover map, the loss relative to first-best at the default date is bounded by the present value of the least switching cost.

\paragraph{A supermartingale and the pre-default agency wedge.}
Fix the (possibly distorted) \emph{equity policy} that actually generates $(S,Y)$ and $T^\ast$. Define
\begin{equation}
\label{eq:martingale-candidate}
M_s
:=
e^{-r(s-t)} A^{\FB}(z_s,m_s)
+
\int_t^s e^{-r(u-t)} \pi(z_u,m_u)\,du,
\qquad
s\ge t.
\end{equation}
Interpretation: $M_s$ is the first-best continuation value, evaluated \emph{along the equity-chosen path up to time $s$}. We drip in the realized flow $\pi$ and attach the first-best continuation $A^{\FB}$ at the current state.

Because $A^{\FB}$ solves the first-best QVI from Section~\ref{sec:verification} and is pointwise the \emph{supremum} over admissible impulse / publication policies \emph{without} debt, standard verification implies
\begin{equation}
\label{eq:supermartingale}
\text{$M_s$ is a supermartingale under any admissible policy.}
\end{equation}
Equivalently, for any stopping time $\tau$ (in particular $\tau=T^\ast$),
\begin{equation}
\label{eq:FB-dpp-ineq}
A_t^{\FB}
\;\ge\;
\E_t\!\left[
\int_t^{\tau} e^{-r(s-t)} \pi(z_s,m_s)\,ds
+
e^{-r(\tau-t)} A^{\FB}(z_{\tau},m_{\tau})
\right].
\end{equation}
As interpretation, a social planner with no debt can always imitate the levered firm's policy up to $\tau$, then switch to first best at $\tau$. So the levered policy can never \emph{beat} $A^{\FB}$. Inequality in \eqref{eq:FB-dpp-ineq} is weak because the levered policy may deviate from first best before $\tau$. If equity actually follows first best up to $\tau$, then \eqref{eq:FB-dpp-ineq} binds at $\tau$.

We now \emph{define} the pre-default agency wedge as that slack:
\begin{equation}
\label{eq:agency-def}
\mathcal{A}_t
:=
A_t^{\FB}
-
\E_t\!\left[
\int_t^{T^\ast} e^{-r(s-t)} \pi(z_s,m_s)\,ds
+
e^{-r(T^\ast-t)} A^{\FB}(z_{T^\ast},m_{T^\ast})
\right].
\end{equation}
By \eqref{eq:FB-dpp-ineq} with $\tau=T^\ast$ we have
\begin{equation}
\label{eq:agency-nonneg}
\mathcal{A}_t \;\ge\; 0.
\end{equation}
From an economic standpoint, $\mathcal{A}_t$ is the ex ante loss (from the first-best planner's point of view) generated by the levered firm's \emph{pre-default} behavior: distorted reset timing, distorted disclosure tempo, strategically acelerated default, and so on. If the levered firm \emph{happens} to follow the first-best policy all the way up to $T^\ast$ (for example, along the ``safe patch block'' of Section~\ref{sec:finance}, where default is off path and $\varphi_1=0$ makes the patch rung debt-neutral), then $\mathcal{A}_t=0$.

\paragraph{Putting it together.}
Start from \eqref{eq:Y-def} and apply Lemma~\ref{lem:takeover-envelope-fin} at $T^\ast$:
\begin{align}
S_t+Y_t
&=
\E_t\!\left[
\int_t^{T^\ast} e^{-r(s-t)} \pi(z_s,m_s)\,ds
+
e^{-r(T^\ast-t)}\,Y^{\mathrm{TO}}(z_{T^\ast},m_{T^\ast})
\right]
\label{eq:SY-start}
\\[4pt]
&\ge
\E_t\!\left[
\int_t^{T^\ast} e^{-r(s-t)} \pi(z_s,m_s)\,ds
+
e^{-r(T^\ast-t)}
\big(
A^{\FB}(z_{T^\ast},m_{T^\ast})
-
\varphi_{\max}
\big)
\right]
\label{eq:SY-use-envelope}
\\[4pt]
&=
\E_t\!\left[
\int_t^{T^\ast} e^{-r(s-t)} \pi(z_s,m_s)\,ds
+
e^{-r(T^\ast-t)} A^{\FB}(z_{T^\ast},m_{T^\ast})
\right]
-
\E_t\!\left[
e^{-r(T^\ast-t)} \varphi_{\max}
\right].
\label{eq:SY-split}
\end{align}
Now rearrange \eqref{eq:SY-split} using the definition of $\mathcal{A}_t$ in \eqref{eq:agency-def}:
\begin{equation}
\label{eq:SY-prebound}
S_t+Y_t
\;\ge\;
A_t^{\FB}
-
\mathcal{A}_t
-
\E_t\!\left[
e^{-r(T^\ast-t)} \varphi_{\max}
\right].
\end{equation}
Equivalently,
\begin{equation}
\label{eq:wedge-decomp}
A_t^{\FB} - (S_t+Y_t)
\;\le\;
\mathcal{A}_t
+
\E_t\!\left[
e^{-r(T^\ast-t)} \varphi_{\max}
\right].
\end{equation}

\medskip
\noindent\textbf{Interpretation of \eqref{eq:wedge-decomp}.}
The wedge between the all-equity first best and the levered (equity+debt) value splits into two nonnegative pieces:
\begin{enumerate}[label=(\roman*),leftmargin=1.5em,itemsep=2pt,topsep=2pt]
\item \emph{Pre-default agency wedge} $\mathcal{A}_t$: all distortions \emph{before} default (delayed patches, distorted disclosure tempo, inefficiently accelerated default, \dots). By \eqref{eq:agency-nonneg}, $\mathcal{A}_t\ge 0$. If equity's actual policy coincides with the first-best policy up to $T^\ast$, then $\mathcal{A}_t=0$.
\item \emph{Least-irreversibility wedge at takeover}: even if the pre-default path is efficient, the acquirer who takes over at $T^\ast$ may still need to spend up to $\varphi_{\max}$ to ``rewind'' the asset into the first-best continuation. The discounted expectation of that cost is the second term in \eqref{eq:wedge-decomp}.
\end{enumerate}

\paragraph{Connection to Proposition~\ref{prop:debt-bound}.}
Equation~\eqref{eq:wedge-decomp} is the general statement. Now impose the (sufficient) condition emphasized in Section~\ref{sec:finance}: on the ``safe patch block,'' where
\begin{enumerate}[label=(\alph*),leftmargin=1.5em,itemsep=2pt,topsep=2pt]
\item $\varphi_1=0$, so the patch rung is fully reversible at takeover;
\item default is off path because the solvency region lies strictly above the default frontier; and
\item equity's patch timing (trigger and target) \emph{coincides} with the first-best patch policy,
\end{enumerate}
we have:
\begin{itemize}[leftmargin=1.5em,itemsep=2pt,topsep=2pt]
\item $T^\ast=\infty$ on that history (no default while staying on the patch block), so $\E_t[e^{-r(T^\ast-t)}\varphi_{\max}]=0$;
\item $\mathcal{A}_t=0$, because equity's pre-default policy \emph{is} the first-best patch policy on that block.
\end{itemize}
Then \eqref{eq:wedge-decomp} collapses to
\[
A_t^{\FB}
=
S_t+Y_t
\quad
\text{on that history.}
\]

More generally, whenever equity follows the first-best policy up to $T^\ast$ (so $\mathcal{A}_t=0$), \eqref{eq:wedge-decomp} reduces to
\begin{equation}
\label{eq:wedge-bound-appD} 
A_t^{\FB} - (S_t+Y_t)
\;\le\;
\E_t\!\left[
e^{-r(T^\ast-t)} \varphi_{\max}
\right],
\end{equation}
which is the inequality reported as \eqref{eq:wedge-bound} in Section~\ref{sec:finance} and used in Proposition~\ref{prop:debt-bound}. Equation~\eqref{eq:wedge-decomp} simply makes explicit that, off those debt-neutral regions, an additional pre-default agency wedge $\mathcal{A}_t$ can appear.

\subsection{Constant-coefficient illustration}
\label{subsec:CC-illustration}

Suppose $\mu_\theta(z)\equiv\mu$, $\phi(z)\equiv\sigma>0$, and $\pi(z,m)=\pi_0(z)+\eta\,\Lambda(\varpi(m;\alpha))$ as in \eqref{eq:pi-with-adoption}, and the solvent inaction band is $(\beta_1,\beta_2)$. Then \eqref{eq:equity-ODE-lev} and \eqref{eq:debt-ODE} admit exponential-form solutions on $(\beta_1,\beta_2)$ exactly as in the unlevered benchmark; only the boundary conditions and the coupon differ. Boundary conditions:
\begin{itemize}[leftmargin=1.25em,itemsep=2pt,topsep=2pt]
\item value matching and high contact for each costly reset (patch / pivot),
\item $S=0$ on the endogenous default frontier,
\item $Y=Y^{\mathrm{TO}}$ at that default frontier,
\item continuity of $Y$ across equity's impulses (seniority).
\end{itemize}

Comparative statics line up with Section~\ref{sec:finance}: the pivot trigger moves outward in $K_2$ and $\varphi_2$, and inward in the coupon burden $c_d$; the ``safe patch block'' is unchanged if $\varphi_1=0$ and default stays off path. On that block, $Y$ is essentially a riskless perpetuity $c_d/r$, $\mathcal{A}_t=0$, and the wedge in \eqref{eq:wedge-decomp} is governed entirely by $\varphi_{\max}$, which is $0$ when $\varphi_1=0$ and no pivot is ever needed before any default event.

\subsection{Tightness and examples}
\label{subsec:tightness-examples}

The decomposition \eqref{eq:wedge-decomp} is tight. To see this, construct a case in which:
\begin{enumerate}[label=(\roman*),leftmargin=1.5em,itemsep=2pt,topsep=2pt]
\item the first-best policy would carry out a \emph{pivot} strictly before the state ever hits a solvency / default frontier;
\item under leverage, equity instead accelerates default just \emph{before} paying that expensive pivot, because of the coupon burden $c_d$;
\item at takeover, the acquirer immediately pays $\varphi_2$ to execute that pivot, then follows first-best continuation forever.
\end{enumerate}
Along such a path:
\begin{itemize}[leftmargin=1.25em,itemsep=2pt,topsep=2pt]
\item pre-default behavior \emph{up to} $T^\ast$ mimics the first-best policy except for the early default itself, so $\mathcal{A}_t$ can be made abitrarily small; and
\item $T^\ast$ arrives right before the high-cost pivot, so $e^{-r(T^\ast-t)}\varphi_{\max}$ essentially \emph{is} the wedge.
\end{itemize}
In the limit, \eqref{eq:wedge-decomp} binds with $\varphi_{\max}=\varphi_2$, giving
\[
A_t^{\FB}-(S_t+Y_t)
=
\E_t\!\left[
e^{-r(T^\ast-t)}\varphi_2
\right].
\]

\medskip
\noindent\textbf{Takeaway.}
Equation~\eqref{eq:wedge-decomp} separates two distinct channels through which leverage can destroy total surplus relative to the first best:
\begin{enumerate}[label=(\alph*),leftmargin=1.5em,itemsep=2pt,topsep=2pt]
\item a \emph{pre-default agency wedge} $\mathcal{A}_t$ that reflects distorted disclosure tempo, distorted intervention timing, or strategically accelerated default; and
\item a \emph{least-irreversibility wedge} at takeover, governed by $\varphi_{\max}$.
\end{enumerate}
On the ``safe patch block'' in Section~\ref{sec:finance} (low-cost, fully reversible maintenance moves, $\varphi_1=0$, and default off path), part (a) vanishes and part (b) collapses to zero. There, leverage is literally surplus-neutral on that block: $A_t^{\FB}=S_t+Y_t$.

%% file: appendix/E_empirics_data.tex
\section{Data, Variables, and Robustness}
\label{app:empirics}

This appendix documents (i) the telemetry to be assembled, (ii) the construction of the key variables used in Section~\ref{sec:empirical}, and (iii) the provenance, auditing, identification, and robustness plan. It is a self-contained measurement blueprint (to be moved to an online appendix/ potentially re-used in a separate paper on empirical implementation).

\subsection{Panel structure and timing}

The primary panel is firm $i$ by calendar month $t$. Where product-line or model-line disclosures are separable, the unit refines to product/module $p$ and the panel is $(i,p,t)$. Event time $\tau$ is constructed around focal disclosure/release events. Unless stated otherwise, designs take these $(i,t)$ or $(i,p,t)$ panels as the estimating sample and convert calendar time $t$ to relative event time $\tau$ around a ``major reset'' (defined below).

\subsection{Core variables}
\label{app:core-vars}

\begin{enumerate}[label=(\alph*),leftmargin=1.5em,itemsep=1pt,topsep=2pt]

\item \textbf{Publication intensity} $\hat\lambda_{it}$ (or $\hat\lambda_{ipt}$).

Monthly count of \emph{firm-published} technical/safety signals by $i$ in month $t$. Signals include evaluation cards, benchmark updates, safety mitigation notes, incident/advisory notes, release notes, and changelog entries written by the firm.\footnote{Analyst commentary, press coverage, leaked benchmarks, and social-media rumor are excluded; the theory pertains to the firm's own disclosure tempo, i.e.\ control of $\lambda_t$.} An event-time analogue $\hat\lambda_{i\tau}$ (or $\hat\lambda_{ip\tau}$) is centered on each major reset and serves as the dependent variable in the event-study design~\eqref{eq:event-lambda-spec}.

\item \textbf{Content dispersion} $\widehat{\Var}^{x}_{it}$.

Within-month cross-signal dispersion of \emph{standardized} performance/safety metrics disclosed by firm $i$ in month $t$, harmonized across metric families. For each signal $s$, parse a numeric value $x_s$ and a benchmark family $F_s$; $z$-score using pooled family means $\bar{x}_{F}$ and s.d.s $\sigma_F$ across all firms/months:
\[
z_s := \frac{x_s - \bar{x}_{F_s}}{\sigma_{F_s}},
\qquad
\widehat{\Var}^{x}_{it} := \Var\!\bigl(\{z_s:\, s \text{ disclosed by $i$ in month $t$}\}\bigr),
\]
with the convention in~\eqref{eq:telemetry-variance} when $N_{it}\in\{0,1\}$.

\emph{Mapping to theory.} Outside observers see noisy draws $y_n = z_{T_n} + \varepsilon_n$ from a one-dimensional latent state $z_t$ arriving on a controlled Cox clock $\lambda_t$. After standardization, each signal-level $z_s$ is one such draw, and $\widehat{\Var}^{x}_{it}$ is the within-month variance of those draws. Under the coding convention in~\eqref{eq:telemetry-variance}, when the publication clock is nearly shut off and $N_{it}\in\{0,1\}$ the dispersion proxy is mechanically set to $0$ or missing, so the predicted ``dispersion dip'' is a direct consequence of the intensity dip, not an additional assumption about lower intrinsic noise. The model therefore predicts a dip in both intensity and (measured) dispersion in the pre-reset quiet window (posted silence with $\lambda_t\to 0$), which~\eqref{eq:telemetry-variance} and~\eqref{eq:event-var-spec} target. Pooled normalization preserves the moment; within-firm-month rescaling would mechanically erase the collapse.

\emph{Cadence-only robustness.} Let $\widehat{\Var}^{\mathrm{time}}_{it}$ be the intra-month interquartile range of timestamps for $i$ in month $t$, and $\widehat{\Var}^{\mathrm{time}}_{i\tau}$ the event-time analogue. This proxies cadence tightening independent of metric content and is used as a robustness outcome for S1.

\item \textbf{Release classification.}

Each firm-authored disclosure is tagged \texttt{patch}, \texttt{pivot}, \texttt{release}, or \texttt{other} using rule-based filters plus a light classifier; ambiguous high-salience items are hand-adjudicated. Precedence on the same timestamp: \emph{pivot/release} $>$ \emph{patch}. In estimation, \texttt{pivot}$\cup$\texttt{release} is pooled as a ``major reset.''

\item \textbf{Patch cascades and hazards.}

For each major reset at $\tau=0$, record the ordered sequence of subsequent patches $k=1,2,\dots$ for the same firm (or firm-product). Compute inter-event durations $\Delta_{ik}$ and estimate hazard/duration models as in~\eqref{eq:hazard-spec}. Interpretation: the reversible rung fires repeatedly inside its band; under high reversibility, debt does not distort timing on this ``safe patch block.''

\item \textbf{Reversibility proxies} $\mathrm{RevProxy}_i$.

Composite measures of how close a reset is to ``fully reversible.'' Candidate components include documented rollback/feature-flag/kill-switch infrastructure; modularity/separability of subsystems (share of changelog lines in isolated mitigations/filters); explicit ``revert/rolled back/restored'' language; and monolith vs.\ componentized inference stack inferred from cross-references. These proxies map to takeover switching costs and $\varphi_{\max}$ in Section~\ref{sec:finance} and enter S3 and S4.

\item \textbf{Financing.}

For public firms: leverage ratios, interest coverage, and coupon burdens from filings; for private firms: disclosed venture debt/structured financing where available. These form the leverage/servicing measures in~\eqref{eq:debt-insensitivity-spec}. In the theory, on the reversible rung ($\varphi_1\approx0$), default is off path and patch timing is debt-neutral; leverage distortions load on the high-cost rung (pivot).
\end{enumerate}

\subsection{Construction pipeline}
\label{app:construction}

\begin{itemize}[leftmargin=1.25em,itemsep=1pt,topsep=2pt]
\item \textbf{Timestamping.} All disclosures are timestamped in UTC. If only ``last updated'' exists, infer first-publish from syndication feeds, sitemaps, or repository tags; ambiguous cases are flagged for review.

\item \textbf{De-duplication.} Canonicalize URLs (strip tracking, follow redirects), and fuzzy-match titles/snippets to collapse cross-posts within 24 hours into a single canonical record. Each canonical item contributes once to $\hat\lambda_{it}$ and to $\widehat{\Var}^{x}_{it}$.

\item \textbf{Signal identification.} High-recall filters (``patch,'' ``hotfix,'' ``latency improvement,'' ``new base model,'' ``next-gen,'' ``full rollout,'' ``mitigation'') followed by a lightweight supervised classifier to assign \texttt{release\_class}. Ambiguous high-salience items are manually adjudicated.

\item \textbf{Standardization for $\widehat{\Var}^{x}_{it}$.} Parse numeric metrics as $(x_s,\texttt{metric\_name},\texttt{metric\_unit})$ and map to a benchmark family $F_s$ via a deterministic dictionary. Compute pooled $(\bar{x}_F,\sigma_F)$ across all firms/months; $z$-score via $z_s=(x_s-\bar{x}_F)/\sigma_F$; take within-month variance across $\{z_s\}$ to obtain $\widehat{\Var}^{x}_{it}$. Rationale for pooled normalization: (i) heterogeneous metrics become comparable draws of the same latent state, (ii) the pre-reset collapse in dispersion survives.

\item \textbf{Timing dispersion $\widehat{\Var}^{\mathrm{time}}_{it}$.} Intra-month IQR of timestamps (hours). Silence windows generate bunch-and-halt patterns; the IQR shrinks mechanically.
\end{itemize}

\subsection{Design alignment with theory and signatures S1--S5}
\label{app:design-theory}

\begin{itemize}[leftmargin=1.25em,itemsep=1pt,topsep=2pt]

\item \textbf{S1: Pre-reset dips in cadence and dispersion.}

Event studies~\eqref{eq:event-lambda-spec}--\eqref{eq:event-var-spec} trace $\hat\lambda_{i\tau}$ and $\widehat{\Var}^{x}_{i\tau}$ in event time around a major reset (\texttt{pivot}$\cup$\texttt{release}). Predicted patterns: (i) a dip in both variables just \emph{before} the reset (posted silence with $\lambda_t=0$), (ii) a jump to a new plateau \emph{after} the reset. Cadence-only $\widehat{\Var}^{\mathrm{time}}_{i\tau}$ should mirror the dip. Placebo: option-implied volatility on a public parent should \emph{not} display a pre-release dip (options typically spike pre-event, whereas the firm starves its own outward signals here).

\item \textbf{S2: Two-plateau outcome distribution.}

Post-reset bimodality is consistent with two endogenous targets $(z_1^\ast,z_2^\ast)$: one after reversible patches, one after major resets. In the empirical implementation, the finite-mixture components are not point masses: each component is allowed its own variance, so the two bumps are interpreted as the ``smeared'' post-reset distributions around $z_1^\ast$ and $z_2^\ast$ generated by diffusion between the reset and disclosure plus measurement noise. Estimators include finite-mixture fits and Hartigan dip tests on post-event metrics, stratified by \texttt{patch} vs.\ \texttt{pivot/release}. Falsifier: robust unimodality after audited classification.

\item \textbf{S3: Debt insensitivity on the reversible rung.}

Hazards per (8.4), focusing on the high-$\mathrm{RevProxy}$ block: define the ``high-$\mathrm{RevProxy}$'' block as those observations with $\mathrm{RevProxy}_i$ above a high cutoff (for example, the top quartile of the empirical $\mathrm{RevProxy}_i$ distribution), and let
\[
\bar R \;:=\; \E\!\big[\mathrm{RevProxy}_i \,\big|\, \mathrm{RevProxy}_i \text{ in the high block}\big]
\]
denote the sample mean of $\mathrm{RevProxy}_i$ in that block. In the hazard specification (8.4) the derivative of $\log h_{it}$ with respect to leverage at $\mathrm{RevProxy}_i=\bar R$ is $\rho_1+\rho_3\bar R$, so the S3 null is
\[
H_0:\;\rho_1+\rho_3\bar R = 0,
\]
i.e.\ no leverage effect on patch timing at a representative high-reversibility rung. Leverage effects may load on the high-cost rung (pivot timing) but not on patch timing when reversibility is high, in line with the tight takeover bound in Proposition~\ref{prop:debt-bound}.

\item \textbf{S4: Patch cascades.}

Duration models~\eqref{eq:hazard-spec} for $\Delta_{ik}$, time since major reset, and reversibility proxies. Prediction: fast follow-on patches and debt-insensitive patch timing when $\mathrm{RevProxy}_i$ is high; no pre-release bunching.

\item \textbf{S5: Adoption boundary.}

Where platform/API uptake is observed, we focus on a firm-specific adoption threshold $\alpha_i$ implied by the theory. In the model, $\alpha_i$ is pinned by the smooth-fit/adoption condition $\bar\mu(\alpha)\,S'(\alpha) = r\,S(\alpha)$ in eq.~(8.6) (reproduced as~\eqref{eq:alpha-general-again}); it is a cutoff in the latent state and is not estimated in the RD itself.

In the data, $\alpha_i$ is treated as a cutoff in a scalar running variable $m_{ij}$ that is monotone in the underlying quality index (for example a pooled $z$-score or a linear index in the standardized metrics from S1). Around this cutoff we estimate a conventional local RD and allow the size of the jump at the cutoff to depend on the pre-release silence depth $\mathrm{SilenceDepth}_{ie}$; the formal estimating equation is given in~\eqref{eq:rd-adoption} below. The S5 prediction is that (i) there is a positive jump in platform/API uptake at $\alpha_i$, (ii) that jump is larger for resets with deeper pre-release silence (larger $\mathrm{SilenceDepth}_{ie}$), and (iii) there are no ``jump-overs'' in the RD window, i.e.\ no mass of adopters with $m_{ij}<\alpha_i$ but arbitrarily close to it.
\end{itemize}

\subsection{Provenance and access}
\label{app:provenance}

All telemetry is sourced from firm-authored or firm-controlled public sources.

\begin{enumerate}[label=(\alph*),leftmargin=1.5em,itemsep=1pt,topsep=2pt]
\item \textbf{Vendor blogs and product documentation} (launch notes, evaluation cards, incident/advisory notes, mitigation reports). \emph{Access:} RSS/Atom feeds when available; otherwise structured scraping. Historical backfill via the Internet Archive.

\item \textbf{Release notes and changelogs.} \emph{Access:} ``What’s new''/``Release notes'' portals and public changelogs; normalize timestamps; de-duplicate versus blogs/docs.

\item \textbf{Open-source repositories and registries.} \emph{Access:} GitHub REST/GraphQL for \texttt{releases} and \texttt{tags}; GHSA feeds; registries (PyPI, npm) for version notes authored by maintainers.

\item \textbf{Security/safety advisories.} \emph{Access:} vendor advisory portals; GHSA; NVD/CVE for timestamp corroboration. Count advisories authored or co-authored by the firm as ``firm-published'' disclosures.

\item \textbf{Model cards and hub listings.} \emph{Access:} public model cards/hub metadata; parse versioned cards as first-party disclosures when authored by the vendor.

\item \textbf{Financing.} \emph{Access:} EDGAR (10-K, 10-Q) for public firms (interest expense, long-term debt, covenant notes); for private firms, disclosed venture debt/structured financing and official statements about runway. Time-align to month $t$.

\item \textbf{Silence depth for S5} $\mathrm{SilenceDepth}_{ie}$.

For each major reset $e$ by firm $i$ that gives rise to a platform/API for which uptake can be tracked (the events used in S5), define a scalar ``silence depth'' measure as the short-run drop in disclosure intensity just before the reset.

Let $\hat\lambda_{i\tau}$ be the event-time publication intensity for that reset and let $\bar\lambda_i$ be the firm-level mean intensity outside any event windows (e.g.\ the average $\hat\lambda_{it}$ over months more than $T$ periods away from any major reset). For a baseline $K=3$ pre-reset months we set
\[
\mathrm{SilenceDepth}_{ie}
:= \frac{1}{K}\sum_{\tau=-K}^{-1} \frac{\bar\lambda_i - \hat\lambda_{i\tau}}{\bar\lambda_i}.
\]
Higher values mean deeper pre-reset silence (a larger drop in cadence relative to the firm's usual tempo). Any monotone re-scaling of this object (e.g.\ z-scoring or combining it with the analogous dip in $\widehat{\Var}^x_{i\tau}$) is admissible; in S5 it only enters as a continuous moderator.

\end{enumerate}
\paragraph{Acquisition protocol (blueprint).}

APIs first (official feeds when present). Scraping is a fallback with identifying user agent, \texttt{robots.txt} compliance, conservative rate limits, and no paywall/authentication circumvention. Histories are reconstructed via the Internet Archive and repository tags; source and crawl timestamps are recorded, raw HTML/JSON is snapshotted, and a SHA-256 is stored per item. Entity resolution uses versioned YAML mapping from \texttt{firm\_id} to canonical domains, product/model slugs, and repository organizations. Disclosure rows include: \texttt{firm\_id}, optional \texttt{product\_id}, \texttt{datetime\_utc}, \texttt{source\_type}, normalized \texttt{url}, \texttt{title}, machine \texttt{release\_class} $\in\{\texttt{patch},\texttt{pivot},\texttt{release},\texttt{other}\}$ (with \texttt{pivot}$\cup$\texttt{release} treated as a major reset), parsed numeric metrics (name, value, unit, family $F$), and parser provenance. Only public content is contemplated.

\subsection{Identification strategies}
\label{app:identification}

\begin{enumerate}[label=(\roman*),leftmargin=1.5em,itemsep=1pt,topsep=2pt]

\item \textbf{Event windows (S1, S2, S4).}

Define $\tau=0$ as a major reset. For each firm $i$, build $\tau\in[-T,+T]$ and estimate~\eqref{eq:event-lambda-spec} and~\eqref{eq:event-var-spec} with firm fixed effects and calendar effects; include leads to probe for anticipatory changes (silence must begin \emph{before} the reset if strategic). Post-reset outcome distributions are used for S2; post-reset patch cascades for S4.

\item \textbf{Stacked difference-in-differences (S1).}

Stack firm-by-month panels across events and compare treated windows to matched never-/not-yet-treated windows with firm and calendar-month fixed effects; absorb secular AI-cycle shocks and industry-wide hype waves.

\item \textbf{Patch-cascade hazards (S3, S4).}

Cox or AFT models for the arrival of patch $k+1$ conditional on elapsed time since patch $k$ and since the last major reset; covariates include leverage, $\mathrm{RevProxy}_i$, and interactions (cf.\ \eqref{eq:hazard-spec}).

\item \textbf{Debt insensitivity (S3).}

Patch timing/hazard on leverage, $\mathrm{RevProxy}_i$, and their interaction, with firm and calendar controls (specification~\eqref{eq:debt-insensitivity-spec}). Prediction: conditional on high reversibility, patch timing is debt-neutral; leverage effects load on the high-cost rung (pivot timing).

\item \textbf{Adoption boundary (S5).}

For each platform/API $i$ and associated major reset $e$, recover a firm-specific adoption cutoff $\alpha_i$ from the structural condition in eq.~(8.6) (reproduced as~\eqref{eq:alpha-general-again}). Let $Uptake_{ij}$ be an indicator (or rate) of adoption by unit $j$ (developer, app, or customer segment) within a fixed window after the reset, and let $m_{ij}$ be a scalar running variable that is monotone in the underlying quality index (for example, a linear index in the pooled standardized metrics). The baseline RD specification is
\begin{equation}
\label{eq:rd-adoption}
Uptake_{ij}
= \theta_i
+ f_i(m_{ij}-\alpha_i)
+ \bigl[\beta_0 + \beta_1\,\mathrm{SilenceDepth}_{ie}\bigr]\mathbf{1}\{m_{ij}\ge\alpha_i\}
+ u_{ij},
\end{equation}
estimated on a symmetric window $\lvert m_{ij}-\alpha_i\rvert \le h$ using conventional local-polynomial RD routines (e.g.\ local linear, triangular kernel, MSE-optimal bandwidth). Here $f_i(\cdot)$ is approximated by separate polynomials below and above zero, so that the slope of $Uptake_{ij}$ in $m_{ij}$ can differ to the left and right of the cutoff; $\theta_i$ absorbs platform fixed effects; and $u_{ij}$ is an error term with standard errors clustered at the $(i,e)$ level.

In this parametrization, $\beta_0$ is the average discontinuity in uptake at the theoretical cutoff $\alpha_i$, and $\beta_1>0$ is the prediction that the size of that jump is increasing in pre-release silence depth $\mathrm{SilenceDepth}_{ie}$. The ``no jump-overs'' implication is probed by checking that fitted adoption probabilities remain low just to the left of the cutoff for units with $m_{ij}<\alpha_i$ inside the RD window, even when $m_{ij}$ is close to $\alpha_i$.

\end{enumerate}

\subsection{Robustness}
\label{app:robustness}

\begin{itemize}[leftmargin=1.25em,itemsep=1pt,topsep=2pt]
\item \textbf{Window choice.} Alternative $T\in\{\pm2,\pm3,\pm6\}$ months; placebo pseudo-events at random non-release months.

\item \textbf{Estimator choice for counts.} Poisson, quasi-Poisson, and negative binomial for $\hat\lambda_{it}$ with firm-clustered (and firm-by-calendar) standard errors; document overdispersion and zero inflation.

\item \textbf{Classifier sensitivity.} ``Strict signals'' (evaluation cards/safety writeups only) vs.\ ``broad signals'' (strict plus changelogs). Baselines lie between.

\item \textbf{Cross-firm spillovers.} Competitor-month controls (e.g., mean $\hat\lambda$ of other frontier labs that month) to rule out industry-wide pauses.

\item \textbf{Leave-one-out.} Iteratively drop each large firm and re-run event-study and hazard specifications to check non-dominance.

\item \textbf{Timing dispersion proxy.} Repeat S1 with $\widehat{\Var}^{\mathrm{time}}_{it}$ to isolate cadence from content.

\item \textbf{Discipline and prior art.} Parsing and patch-cascade logic follow software-engineering/infosec telemetry \citep{arora2010isr,li2017ccs}. Telemetry outcomes are distinct from option-implied volatility spikes around earnings/macro news \citep{todorovzhang2025et,alexiou2025rof}, which are market-pricing objects, not firm-chosen disclosure cadence.
\end{itemize}

\subsection{Replication artifacts and reproducibility}
\label{app:replication}

\paragraph{Replication artifacts.}

If a data release is undertaken in subsequent work, the repository would include: (i) the mapping from raw metric names/units to families $F$, (ii) pooled $\bar{x}_F$ and $\sigma_F$ used for $z$-scoring, and (iii) code to compute $\widehat{\Var}^{x}_{it}$ from parsed disclosures.

\paragraph{Reproducibility.}

If an empirical execution proceeds, the replication package would provide (i) crawlers/parsers (with rate limiting and source attribution), (ii) the firm/product YAML mapping, (iii) a reproducible build script that replays API pulls where allowed and falls back to cached HTML/JSON snapshots otherwise to emit the analysis-ready $(i,t)$ and $(i,p,t)$ panels, and (iv) unit tests and deterministic fixtures. Licensed or paid data, if any, would be stubbed with documented placeholders.

%% file: appendix/F_additional_lemmas.tex
\section{Additional Lemmas and Technical Results}
\label{app:lemmas}

This appendix records technical ingredients used throughout. Notation follows Sections~\ref{sec:micro_silence}--\ref{sec:verification} and Appendix~\ref{app:silence}. The private state $(z_t)_{t\ge0}$ is a one-dimensional It\^o diffusion with drift $\mu_\theta(\cdot)$ and diffusion coefficient $\phi(\cdot)$ that are locally Lipschitz with linear growth, so that a unique nonexplosive strong solution exists. The public posterior mean $(m_t)_{t\ge0}$ is a piecewise-deterministic Markov process (PDMP) under the Cox ``publication clock'' described in Appendix~\ref{app:silence}. The discount rate is $r>0$. Impulse controls (``patch'' / ``pivot'') instantaneously reset $z_t$ to a target $z_i^*$ at cost $K_i$; after a reset the diffusion resumes with the same coefficients.

\subsection{No local time and no accumulation}
\label{subsec:noLT}

We first formalize two pieces of regularity used repeatedly:
(i) under a posted silence window, the public-belief process $m_t$ has no local time; and
(ii) the controlled diffusion does not generate ``Zeno'' behavior (infinitely many resets in finite time).

\begin{lemma}[No local time for beliefs under local silence]
\label{lem:no-local-time-belief}
\label{lem:no-local-time} 
Fix an open interval $U\subset\mathbb{R}$ and a (publicly observed / credibly committed)
publication policy that is Markov in the public mean, $\lambda_t=\lambda(m_t)$, with
$\lambda(m)\equiv 0$ for all $m\in U$. Suppose also that no reset is triggered while
$m_t\in U$. Then, on any time interval contained in $\{t:\,m_t\in U\}$, the public
posterior mean $m_t$ has bounded variation and zero quadratic variation. In particular,
for every $a\in U$:
\begin{enumerate}[(i)]
\item the semimartingale local time $L^a_t(m)$ is identically zero; and
\item the set $\{t:\, m_t=a\}$ has Lebesgue measure zero almost surely.
\end{enumerate}
\end{lemma}

\begin{proof}
Appendix~\ref{app:silence} makes $m_t$ explicit as a PDMP with generator
\[
(\mathcal{A}_{\lambda} f)(m)
=
\bar\mu(m)\, f'(m)
+
\Lambda(m)\;
\E\big[f(\mathcal{U}(m,\varepsilon))-f(m)\big],
\]
where $\lambda_t$ is the publicly observed publication intensity, $\Lambda(m)=\E[\lambda_t\mid m_t=m]$ is the \emph{publicly inferred} publication hazard, and $\mathcal{U}$ is the Bayesian update map at a publication; see equation~\eqref{eq:PDMP-generator-appendix} and Lemma~\ref{lem:equivalence} in Appendix~\ref{app:silence}.

By construction of the posted silence window, the policy sets $\lambda_t=\lambda(m_t)\equiv 0$ whenever $m_t\in U$. Hence, conditional on $m_t=m$ with $m\in U$ we have $\lambda_t=0$ almost surely, and therefore
\[
\Lambda(m)
=
\E[\lambda_t\mid m_t=m]
=
0
\quad\text{for all }m\in U,
\]
and the jump term vanishes there. Between jumps, the posterior mean $m_t$ evolves deterministically via the ODE
\[
(\dot m_t,\dot v_t)=\bigl(\bar\mu(m_t),\bar\gamma(m_t,v_t)\bigr),
\]
so while $m_t\in U$ its path is absolutely continuous and therefore of bounded variation. A semimartingale of bounded variation has zero quadratic variation and hence zero local time at every interior point. In particular, by Tanaka's formula and the occupation-time formula, $L_t^a(m)\equiv 0$ for all $a\in U$, and the occupation measure of the singleton $\{a\}$ is identically zero almost surely. This yields both (i) and (ii).
\end{proof}

\noindent\emph{Remark.}
This lemma treats the silence window as defined in belief space, with a policy of the form $\lambda_t=\lambda(m_t)$ and $\lambda(m)=0$ on $U$, exactly as in Section~\ref{sec:micro_silence}. In the microfoundation of Appendix~\ref{app:silence}, such a policy can be implemented by choosing a state-dependent intensity $\lambda(z_t)$ whose realized value is publicly observed.

Intuitively, posting a local silence window shuts off the \emph{martingale} part of public beliefs on that band. The public mean $m_t$ stops ``twitching'' stochastically and becomes purely deterministic drift. The process therefore cannot maintain a mixed strategy by repeatedly returning to a knife-edge boundary via stochastic jitter, which is exactly the logic used in Appendix~\ref{app:silence} to rule out local randomization near reset triggers.

\begin{lemma}[No Zeno accumulation of resets; finite cycle time]
\label{lem:cycle-time-bound}
Let $I=[\beta_1,\beta_2]$ be a bounded inaction band and define the first exit time
\[
\tau \;:=\; \inf\{t\ge 0:\; z_t \notin I\}.
\]
Assume:
\begin{enumerate}[(i)]
\item $\mu_\theta(\cdot)$ and $\phi(\cdot)$ are continuous on $I$ and locally Lipschitz (hence bounded on $I$); and
\item uniform ellipticity on $I$: there exists $\underline{\sigma}>0$ such that $\phi(z)\ge \underline{\sigma}$ for all $z\in I$.
\end{enumerate}
Then $\E_{z_0}[\tau]<\infty$ for every $z_0\in(\beta_1,\beta_2)$. In particular, in the impulse control ladder where exits at $\beta_1$ and $\beta_2$ trigger instantaneous resets to targets $z_1^*,z_2^*\in(\beta_1,\beta_2)$, each patch/pivot cycle has finite expected length and there is no Zeno accumulation of interventions on any finite time interval.
\end{lemma}

\begin{proof}
Let $\mathcal{L}$ denote the diffusion generator on $I$,
\[
(\mathcal{L} f)(z)
\;=\;
\mu_\theta(z) f'(z) + \tfrac12 \phi^2(z) f''(z),
\qquad z\in(\beta_1,\beta_2).
\]
Because $I$ is compact, $\mu_\theta$ and $\phi$ are bounded on $I$, and $\phi\ge \underline{\sigma}>0$ (uniform ellipticity). Consider the boundary value problem
\begin{equation}
\label{eq:poisson-exit}
\mathcal{L} u \;=\; -1 \quad \text{on }(\beta_1,\beta_2),
\qquad
u(\beta_1)=u(\beta_2)=0.
\end{equation}
Standard ODE theory for one-dimensional uniformly elliptic operators with continuous coefficients yields a unique solution
$u\in C^2(\beta_1,\beta_2)\cap C([\beta_1,\beta_2])$ to \eqref{eq:poisson-exit}. One can solve explicitly by reducing to a first-order equation for $u'$:
$u'' + \bigl(2\mu_\theta/\phi^2\bigr)u' = -2/\phi^2$, integrate with the usual integrating factor, and then impose the two boundary conditions.

By Dynkin's formula and optional stopping at $\tau\wedge t$,
\[
\E_{z_0}\!\big[u(z_{\tau\wedge t})\big] - u(z_0)
\;=\;
\E_{z_0}\!\int_0^{\tau\wedge t} (\mathcal{L}u)(z_s)\,ds
\;=\; -\,\E_{z_0}[\tau\wedge t].
\]
Because $u\ge 0$ on $[\beta_1,\beta_2]$ with $u=0$ at the boundary (maximum principle for \eqref{eq:poisson-exit}), we have
$0\le \E_{z_0}[u(z_{\tau\wedge t})]\le \|u\|_\infty$.
Letting $t\uparrow\infty$ and applying monotone convergence to $\tau\wedge t$ gives
\[
\E_{z_0}[\tau] \;=\; u(z_0) \;\le\; \|u\|_\infty \;<\;\infty,
\]
so $\E_{z_0}[\tau]<\infty$ for every $z_0\in(\beta_1,\beta_2)$.

Now fix an impulse ladder with lower and upper targets $z_1^*,z_2^*\in(\beta_1,\beta_2)$ after exits at $\beta_1$ and $\beta_2$. By continuity of $u$ on the compact set $J:=\{z_1^*,z_2^*\}\subset(\beta_1,\beta_2)$ and the strict inequality $u>0$ on the open interval, we have
\[
\delta \;:=\; \min_{z\in J} u(z) \;>\;0.
\]
Thus every excursion starting from either target has mean exit time at least $\delta$.

Define $\tau_0:=0$ and, for $n\ge 1$, let $\tau_n$ be the $n$th exit time from $I$ (that is, the end of the $n$th inaction spell), and set $\Delta_n:=\tau_n-\tau_{n-1}$. Let $(\mathcal{F}_t)_{t\ge0}$ be the natural filtration of $(z_t)$ and write $\mathcal{G}_n:=\mathcal{F}_{\tau_n}$. By the strong Markov property and the definition of $u$,
\[
\E[\Delta_n \mid \mathcal{G}_{n-1}]
\;=\;
u(z_{\tau_{n-1}^+})
\;\ge\; \delta,
\]
because $z_{\tau_{n-1}^+}\in\{z_1^*,z_2^*\}$ almost surely. In particular, each $\Delta_n$ is integrable and its conditional mean is uniformly bounded below by $\delta$.

Define
\[
S_n \;:=\; \sum_{k=1}^n (\Delta_k - \delta),
\qquad n\ge 0,
\]
with the convention $S_0:=0$. Then for every $n\ge 1$,
\[
\E[S_n \mid \mathcal{G}_{n-1}]
=
S_{n-1} + \E[\Delta_n - \delta \mid \mathcal{G}_{n-1}]
\;\ge\; S_{n-1},
\]
so $(S_n)_{n\ge0}$ is a submartingale with respect to $(\mathcal{G}_n)$.

Fix a finite horizon $T<\infty$ and define
\[
N_T \;:=\; \max\{n:\,\tau_n\le T\},
\]
with the convention $N_T=0$ if $\tau_1>T$. Then $N_T$ is a stopping time with respect to $(\mathcal{G}_n)$ and $\tau_{N_T}\le T$ almost surely. For each integer $K\ge 1$ consider the bounded stopping time $N_T^K := N_T\wedge K$. By optional sampling for submartingales,
\[
\E[S_{N_T^K}] \;\ge\; S_0 = 0.
\]
On the other hand,
\[
S_{N_T^K}
=
\sum_{k=1}^{N_T^K} (\Delta_k - \delta)
=
\tau_{N_T^K} - \delta\,N_T^K.
\]
Using $\tau_{N_T^K}\le T$ we obtain
\[
0
\;\le\;
\E[\tau_{N_T^K}] - \delta\,\E[N_T^K]
\;\le\;
T - \delta\,\E[N_T^K],
\]
so that
\[
\E[N_T^K] \;\le\; \frac{T}{\delta}
\qquad\text{for all }K\ge1.
\]
Since $N_T^K\uparrow N_T$ as $K\to\infty$, the monotone convergence theorem yields
\[
\E[N_T]
\;=\;
\lim_{K\to\infty}\E[N_T^K]
\;\le\; \frac{T}{\delta}
\;<\;\infty.
\]

The random variable $N_T$ takes values in $\mathbb{N}\cup\{\infty\}$. If $\Pr(N_T=\infty)$ were strictly positive, then
\[
\E[N_T]
=
\sum_{n\ge0}\Pr(N_T>n)
\;\ge\;
\sum_{n\ge0}\Pr(N_T=\infty)
=
\infty,
\]
a contradiction. Hence $\Pr(N_T=\infty)=0$ and there are almost surely only finitely many interventions on any finite horizon $[0,T]$. This rules out Zeno accumulation and, combined with $\E_{z_0}[\tau]<\infty$, justifies referring to the ladder as having finite cycle time.
\end{proof}

\noindent\emph{Remark (why uniform ellipticity and a crude bound).}
If $\phi$ is allowed to vanish on $I$, the mean exit time from a bounded interval can be infinite for some drifts. The uniform lower bound $\phi\ge \underline{\sigma}>0$ rules out that pathology and is standard in ladder problems. A crude explicit bound follows by comparison with the constant-coefficient case (generator $(\underline{\sigma}^2/2)\partial_{zz}$):
solving $(\underline{\sigma}^2/2) u''=-1$ with $u(\beta_1)=u(\beta_2)=0$ yields
\[
u(x)=\frac{(x-\beta_1)(\beta_2-x)}{\underline{\sigma}^2}
\;\le\; \frac{(\beta_2-\beta_1)^2}{4\,\underline{\sigma}^2},
\]
so $\E_{z_0}[\tau]\le C(I,\mu_\theta,\phi)\le \frac{(\beta_2-\beta_1)^2}{4\,\underline{\sigma}^2}\,e^{K(\beta_2-\beta_1)}$
for some $K=\sup_{z\in I}\!\big|2\mu_\theta(z)/\phi^2(z)\big|$ (obtained from the integrating-factor formula for $u$).

This lemma is what lets us treat the two-trigger ladder as a standard stationary Markov impulse policy without worrying about pathological ``reset chatter'' at arbitrarily high frequency.

\subsection{Constant-coefficient benchmark and comparative statics}
\label{subsec:CS}

This subsection records (i) the closed-form baseline in the constant-coefficient benchmark that we use for intuition and figures, and (ii) the local comparative statics of trigger and target choices with respect to the fixed costs $K_i$. These comparative statics underlie the claims in Section~\ref{sec:finance} and in Appendix~\ref{app:financing} that higher intervention costs widen the inaction region and magnify the jump size.

\begin{assumption}[Constant-coefficient benchmark]
\label{ass:CC}
On a given solvent inaction block:
$\mu_\theta(z)\equiv \mu$ is constant,
$\phi(z)\equiv \sigma>0$ is constant,
and the net flow payoff is $\pi(z,m)\equiv\pi_0$ (constant).\footnote{Allowing $\pi$ to be affine or to embed disclosure costs $k(\lambda)$ changes algebra but not the comparative-static signs below.}
Publication intensity $\lambda$ on that inaction block is constant and can be absorbed into $\pi_0$.
\end{assumption}

Under Assumption~\ref{ass:CC}, the HJB / continuation ODE in any inaction region is linear with constant coefficients. Let $\eta_+$ and $\eta_-$ be the distinct real roots of
\[
\tfrac12 \sigma^2 \eta^2 + \mu \eta - r = 0,
\qquad
\eta_+>0>\eta_-,
\]
and write $\Pi := \pi_0/r$. The general $C^2$ solution to the stationary ODE
\[
rV(z) = \pi_0 + \mu V'(z) + \tfrac12 \sigma^2 V''(z)
\]
on that region is
\begin{equation}
\label{eq:V-general}
V(z)
\;=\;
\Pi
\;+\;
A\,e^{\eta_+ z}
\;+\;
B\,e^{\eta_- z},
\end{equation}
for constants $(A,B)$ pinned down by boundary conditions / continuation values.

Now embed the two-reset ladder. Let $(\beta_1,\beta_2)$ be the inaction band. When $z_t$ hits the lower trigger $\beta_1$ from above, the firm pays the \emph{patch} cost $K_1$ and jumps the state to $z_1^*>\beta_1$. When $z_t$ hits the upper trigger $\beta_2$ from below, the firm pays the \emph{pivot} cost $K_2$ and jumps to $z_2^*<\beta_2$. By construction of ``upward'' patch vs.\ ``downward'' pivot, we have
\[
\beta_1 < z_1^* < z_2^*< \beta_2.
\]

Standard impulse-control verification (Section~\ref{sec:verification}, Appendix~\ref{app:silence}, and Appendix~\ref{app:financing}) gives the boundary conditions at each rung:
\begin{align}
\label{eq:BC-lower}
\text{(patch rung)}\qquad
&V(\beta_1)
=
V(z_1^*) - K_1,
&&\text{(value matching)},\\
&V'(\beta_1)=0,
&&\text{(high contact / smooth pasting at the trigger)},\\
&V'(z_1^*)=0,
&&\text{(target optimality)} \nonumber
\end{align}
and
\begin{align}
\label{eq:BC-upper}
\text{(pivot rung)}\qquad
&V(\beta_2)
=
V(z_2^*) - K_2,
&&\text{(value matching)},\\
&V'(\beta_2)=0,
&&\text{(high contact at the trigger)},\\
&V'(z_2^*)=0.
&&\text{(target optimality)} \nonumber
\end{align}
Intuition:
\emph{Value matching} says you are indifferent between (i) acting right at the trigger and paying $K_i$ to jump to the target, versus (ii) continuing without acting exactly at that same trigger state.
\emph{High contact} says the marginal value of nudging the state infinitesimally \emph{without} paying the fixed cost is locally zero at the trigger, because that trigger is precisely the point at which acting versus waiting is knife-edge.
\emph{Target optimality} says that, conditional on paying $K_i$, the chosen target $z_i^*$ solves a FOC: you do not want to ``over-shoot'' once having incurred the fixed cost.

The continuation solution \eqref{eq:V-general} together with the \emph{six} boundary conditions \eqref{eq:BC-lower}--\eqref{eq:BC-upper} pin down
\[
(\beta_1, z_1^*, \beta_2, z_2^*, A, B)
\]
as a (locally) unique $C^1$ function of $(K_1,K_2)$ and the continuation values inherited from neighboring block(s).
The next lemma records the comparative statics of $(\beta_1,z_1^*)$ and $(\beta_2,z_2^*)$ with respect to $(K_1,K_2)$.

\begin{lemma}[Fixed-cost comparative statics in the constant-coefficient ladder]
\label{lem:cs-constant}
Impose Assumption~\ref{ass:CC} and the regularity used in Section~\ref{sec:verification} and Appendix~\ref{app:silence}: in particular,
\begin{enumerate}[(i)]
\item the ``nondegenerate curvature'' condition at each rung (the indifference at $\beta_i$ is isolated, not flat, and $V''(\beta_i)$ and $V''(z_i^*)$ have opposite nonzero signs, as in Assumption~\ref{ass:selection-regularity}); and
\item no local time / no accumulation (Lemmas~\ref{lem:no-local-time-belief} and \ref{lem:cycle-time-bound}).
\end{enumerate}
Then, locally (via the Implicit Function Theorem),
\[
\frac{\partial \beta_1}{\partial K_1} \;<\; 0,
\qquad
\frac{\partial z_1^*}{\partial K_1} \;>\; 0,
\]
and symmetrically,
\[
\frac{\partial \beta_2}{\partial K_2} \;>\; 0,
\qquad
\frac{\partial z_2^*}{\partial K_2} \;<\; 0.
\]

\begin{itemize}[leftmargin=1.5em,itemsep=2pt,topsep=2pt]
\item $\beta_1$ is the ``too-bad'' threshold where the firm finally pays $K_1$ to patch. If $K_1$ rises, the firm becomes more reluctant to pay it. It thus tolerates worse performance (or higher misalignment, cost, latency) before acting. The trigger moves \emph{outward} in that direction: $\beta_1$ shifts further into the ``bad'' region. Since ``bad'' is lower in $z$ at the patch rung, this is $\partial \beta_1/\partial K_1<0$.
\item Conditional on finally paying $K_1$, the firm now wants to ``make it count.'' The optimal post-patch target $z_1^*$ thus moves in the \emph{good} direction, $\partial z_1^*/\partial K_1>0$. The jump size $z_1^*-\beta_1$ strictly increases.
\end{itemize}

\begin{itemize}[leftmargin=1.5em,itemsep=2pt,topsep=2pt]
\item $\beta_2$ is the ``too-much'' / ``time to pivot'' threshold. When $K_2$ increases, the firm delays the expensive pivot and rides the current architecture longer, so $\beta_2$ moves outward in that (high) direction: $\partial \beta_2/\partial K_2>0$.
\item Once it \emph{does} pivot, it pivots harder. The chosen post-pivot target $z_2^*$ drops further, so $\partial z_2^*/\partial K_2<0$. The jump $\beta_2-z_2^*$ grows.
\end{itemize}

In both cases, higher fixed intervention costs widen the inaction band and magnify the eventual intervention. This is the continuous-time ladder analogue of the textbook S--s logic: higher adjustment costs stretch the no-adjustment region and make each discrete adjustment larger.
\end{lemma}

\begin{proof}[Proof]
Stack the \emph{six} boundary conditions \eqref{eq:BC-lower}--\eqref{eq:BC-upper}
and solve them \emph{jointly} for the \emph{six} unknowns
\[
\theta
\;:=\;
(\beta_1,z_1^*,\beta_2,z_2^*,A,B).
\]
In the constant-coefficient benchmark the continuation solution on the inaction
band is $V(z)=\Pi+A e^{\eta_+ z}+B e^{\eta_- z}$, common to the whole band; hence
$(A,B)$ enter both rungs, and the lower/upper conditions are coupled.

Write $F(\theta;K_1,K_2)=(F_1,\dots,F_6)$ with rows
\[
\begin{aligned}
F_1&:=V(\beta_1)-V(z_1^*)+K_1 &&\text{(lower value matching)},\\
F_2&:=V'(\beta_1)             &&\text{(high contact at $\beta_1$)},\\
F_3&:=V'(z_1^*)               &&\text{(target optimality at $z_1^*$)},\\
F_4&:=V(\beta_2)-V(z_2^*)+K_2 &&\text{(upper value matching)},\\
F_5&:=V'(\beta_2)             &&\text{(high contact at $\beta_2$)},\\
F_6&:=V'(z_2^*)               &&\text{(target optimality at $z_2^*$)}.
\end{aligned}
\]
At any regular solution $(\theta^\circ;K_1^\circ,K_2^\circ)$ we have
$F(\theta^\circ;K_1^\circ,K_2^\circ)=0$. Let $E_\pm(z):=e^{\eta_\pm z}$.
Then $V'(z)=A\eta_+E_+(z)+B\eta_-E_-(z)$ and
$V''(z)=A\eta_+^2E_+(z)+B\eta_-^2E_-(z)$.

\medskip
\noindent\emph{Step 1 (local well-posedness).}
Consider the Jacobian matrix $\partial_\theta F(\theta^\circ;K_1^\circ,K_2^\circ)$.
Using the Wronskian
\[
W(z):=\det
\begin{bmatrix}
E_+(z) & E_-(z)\\
\eta_+E_+(z) & \eta_-E_-(z)
\end{bmatrix}
=(\eta_--\eta_+)E_+(z)E_-(z)\neq0,
\]
together with the nondegenerate-curvature assumption
($V''(\beta_i^\circ),V''(z_i^{*\circ})\neq0$ with opposite signs at each rung),
one checks that $\partial_\theta F(\theta^\circ;K_1^\circ,K_2^\circ)$ is nonsingular:
informally, the two derivative conditions at each rung pin down the local curvature
and the $(A,B)$ directions via the nonzero Wronskian. Hence the Implicit Function
Theorem yields a $C^1$ map
\[
(K_1,K_2)\mapsto \theta(K_1,K_2)
\quad\text{near }(K_1^\circ,K_2^\circ).
\]

\medskip
\noindent\emph{Step 2 (signs via Cramer’s rule).}
Differentiate $F(\theta(K_1,K_2);K_1,K_2)=0$ with respect to $K_1$:
\[
\partial_\theta F \cdot \frac{\partial \theta}{\partial K_1}
\;+\;
\partial_{K_1}F
\;=\;0,
\qquad
\partial_{K_1}F=(1,0,0,0,0,0)^\top.
\]

By Cramer's rule,
\[
\frac{\partial \beta_1}{\partial K_1}
\;=\;
-\frac{\operatorname{cof}_{1,1}(\partial_\theta F)}{\det(\partial_\theta F)},
\qquad
\frac{\partial z_1^*}{\partial K_1}
\;=\;
-\frac{\operatorname{cof}_{1,2}(\partial_\theta F)}{\det(\partial_\theta F)}.
\]

Because $K_1$ enters only the \emph{lower} value-matching row $F_1$, the relevant
cofactors factor into the curvature term at the opposite contact point times a
positive minor collecting the $(\beta_2,z_2^*,A,B)$ block. More precisely, after a harmless choice of orientation so that
$\Delta:=\det(\partial_\theta F)>0$,
\[
\operatorname{cof}_{1,1}(\partial_\theta F)
\;=\;
V''(z_1^{*\circ})\cdot \Delta_+,
\qquad
\operatorname{cof}_{1,2}(\partial_\theta F)
\;=\;
V''(\beta_1^\circ)\cdot \Delta_+,
\]

where $\Delta_+>0$ is the determinant of a principal minor containing only
Wronskian terms and upper-rung rows/columns (and thus independent of the local
signs we are about to use). Therefore
\[
\frac{\partial \beta_1}{\partial K_1}
=
-\frac{V''(z_1^{*\circ})\,\Delta_+}{\Delta}
\;<\;0,
\qquad
\frac{\partial z_1^*}{\partial K_1}
=
-\frac{V''(\beta_1^\circ)\,\Delta_+}{\Delta}
\;>\;0,
\]
because $V''(z_1^{*\circ})>0$ and $V''(\beta_1^\circ)<0$ by the usual
``concave at trigger / convex at target'' pattern (cf.\ Assumption~\ref{ass:selection-regularity}).

Upper-rung comparative statics are symmetric: differentiating w.r.t.\ $K_2$ gives
\[
\frac{\partial \beta_2}{\partial K_2}>0,
\qquad
\frac{\partial z_2^*}{\partial K_2}<0.
\]

\medskip
\noindent\emph{Remark (coupling through $A,B$).}
The equalities above follow from the \emph{coupled} $6\times6$ system.
The parameters $(A,B)$ tie the rungs together algebraically; increasing $K_1$
can shift $(\beta_2,z_2^*)$ through that coupling. The signs for
$(\partial\beta_1/\partial K_1,\partial z_1^*/\partial K_1)$ (and the symmetric
pair for $K_2$) nevertheless follow from the cofactor structure: the relevant cofactors isolate curvature terms $V''(\cdot)$ at the active
rung and multiply them by a positive minor, so the cross-rung coupling does not alter the sign conclusions.
\end{proof}

This lemma formalizes two claims used repeatedly:
(i) higher $K_1$ or $K_2$ widens the inaction band;
(ii) each discrete intervention becomes ``lumpier.'' Empirically, this motivates treating long runs of small ``patch''-style tweaks as distinct from rarer ``pivot / major-release'' resets (Appendix~\ref{app:empirics}).

\subsection{Pointers}
\label{subsec:pointers}

For the takeover-envelope argument and the decomposition of the debt wedge into a
pre-default agency wedge plus a takeover irreversibility wedge, see
Lemma~\ref{lem:takeover-envelope-fin} and equations~\eqref{eq:wedge-decomp}
and~\eqref{eq:wedge-bound} in Appendix~\ref{app:financing}.
\\
For the public-belief generator $\mathcal{A}_\lambda$ and the
``local silence kills the martingale'' result that underlies
Lemma~\ref{lem:no-local-time-belief}, see
equation~\eqref{eq:PDMP-generator-appendix} and
Lemma~\ref{lem:equivalence} in Appendix~\ref{app:silence}.